\documentclass[a4paper,UKenglish,cleveref,autoref,thm-restate,colorlinks]{lipics-v2021}

\pdfoutput=1 \hideLIPIcs

\usepackage{microtype}
\usepackage[utf8]{inputenc}
\usepackage[T1]{fontenc}
\usepackage{amsmath, amssymb}
\usepackage{stmaryrd}
\usepackage{microtype}
\usepackage{IEEEtrantools}
\usepackage{array}
\usepackage{breakcites}

\usepackage{tikz}
\usetikzlibrary{automata,positioning,backgrounds,arrows,shapes.geometric,fit,matrix,decorations.pathreplacing,decorations.pathmorphing,overlay-beamer-styles}

\tikzset{
  >=stealth,
  initial text=,
  dotnode/.style={
    state,
    white,
    text=black,
    inner sep=0,
    align=center
  },
      diagonal fill/.style 2 args={
  fill=#2, path picture={
    \fill[#1, sharp corners] (path picture bounding box.south west) -| (path picture bounding box.north east) -- cycle;
    }
  },
  reversed diagonal fill/.style 2 args={
    fill=#2, path picture={
      \fill[#1, sharp corners] (path picture bounding box.north west) |- 
      (path picture bounding box.south east) -- cycle;
    }
  },
  emptynode/.style={
    minimum size=0,
    inner sep=0,
    outer sep=0
  },
}

\definecolor{custom-main}{rgb}{0.0, 0.0, 0.0}
\definecolor{custom-sub}{rgb}{0.0, 0.0, 0.0}
\definecolor{custom-purple}{rgb}{0.8, 0.1, 0.8} \definecolor{custom-amethyst}{rgb}{0.6, 0.4, 0.8}
\definecolor{custom-red}{rgb}{0.9, 0.1, 0.1}
\definecolor{custom-green}{rgb}{0.2, 0.7, 0.4}
\definecolor{custom-turquoise}{rgb}{0.3, 0.8, 0.8}

\newcommand{\init}{\mathsf{init}}

\newcommand{\decomp}{\mathcal{D}}

\newcommand{\mealyStateID}{\mealyStateSpace^{\playerSubset, \decomp}}

\newcommand{\mealyUpdateID}{\mealyUpdate^{\playerSubset, \decomp}}
\newcommand{\mealyNextIDpI}{\mealyNext_\playerIndex^{\playerSubset, \decomp}}
\newcommand{\mealyStateD}{\mealyStateSpace^{\decomp}}

\newcommand{\mealyUpdateD}{\mealyUpdate^{\decomp}}
\newcommand{\mealyNextDpI}{\mealyNext_\playerIndex^{\decomp}}

\newcommand{\IR}{\mathbb{R}}

\newcommand{\IN}{\mathbb{N}}

\newcommand{\IRbar}{\overline{\mathbb{R}}}

\newcommand{\INbar}{\overline{\mathbb{N}}}

\newcommand{\integerInterval}[1]{\llbracket{}#1\rrbracket{}}

\newcommand{\indexPosition}{\ell} \newcommand{\indexLast}{r} \newcommand{\numSegments}{k} \newcommand{\indexSegment}{j} \newcommand{\indexParity}{p} 

\newcommand{\nPlayer}{n}
\newcommand{\playerSet}{\integerInterval{\nPlayer}}
\newcommand{\playerSubset}{I}
\newcommand{\playerIndex}{i}
\newcommand{\indexPlayer}{\playerIndex}
\newcommand{\player}[1]{\mathcal{P}_{#1}}
\newcommand{\playerI}{\player{\playerIndex}}
\newcommand{\playerIAlt}{\player{\playerIndex'}}
\newcommand{\playerIAltAlt}{\player{\playerIndex''}}
\newcommand{\playerOne}{\player{1}}
\newcommand{\playerTwo}{\player{2}}
\newcommand{\playerIAdv}{\player{-\playerIndex}}

\newcommand{\game}{\mathcal{G}}
\newcommand{\arena}{\mathcal{A}}
\newcommand{\vertexSet}{V}
\newcommand{\vertexSetOne}{\vertexSet_1}
\newcommand{\vertexSetTwo}{\vertexSet_2}
\newcommand{\vertexSetI}{\vertexSet_\playerIndex}
\newcommand{\vertex}{v}
\newcommand{\vertexAlt}{w}
\newcommand{\vertexAltAlt}{u}
\newcommand{\edgeSet}{E}
\newcommand{\edge}{e}
\newcommand{\succSet}[1]{\mathsf{Succ}_\edgeSet(#1)}

\newcommand{\target}{T}
\newcommand{\targetI}{\target_\playerIndex}
\newcommand{\reach}[1]{\mathsf{Reach}(#1)}
\newcommand{\safe}[1]{\mathsf{Safe}(#1)}
\newcommand{\buchi}[1]{\mathsf{B\ddot{u}chi}(#1)}
\newcommand{\cobuchi}[1]{\mathsf{coB\ddot{u}chi}(#1)}

\newcommand{\vispos}[1]{\mathsf{VisPos}_{\game}(#1)}

\newcommand{\satpl}[1]{\mathsf{SatPl}_{\game}(#1)}

\newcommand{\objective}{\Omega}
\newcommand{\cost}[1]{\mathsf{cost}_{#1}}
\newcommand{\costI}{\cost{i}}
\newcommand{\costOne}{\cost{1}}
\newcommand{\costVector}{(\costI)_{\playerIndex\in\playerSet}}

\newcommand{\weight}{w}

\newcommand{\costReach}[2]{\mathsf{SPath}^{#1}_{#2}}

\newcommand{\spathIW}{\mathsf{SPath}^{\targetI}_{\weight_{\playerIndex}}}

\newcommand{\val}{\mathsf{val}}
\newcommand{\valI}{\mathsf{val}_\playerIndex}

\newcommand{\winningOne}[1]{W_{1}(#1)}
\newcommand{\winningTwo}[1]{W_{2}(#1)}
\newcommand{\winningI}[1]{W_{\playerIndex}(#1)}
\newcommand{\winningIAdv}[1]{W_{-\playerIndex}(#1)}

\newcommand{\arenaTuple}{((\vertexSetI)_{\playerIndex\in\playerSet}, \edgeSet)}
\newcommand{\gameTuple}{(\arena, \costVector)}

\newcommand{\playSet}[1]{\mathsf{Plays}(#1)}
\newcommand{\play}{\pi}
\newcommand{\playB}{\rho}
\newcommand{\segment}{\mathsf{sg}}
\newcommand{\historySet}[1]{\mathsf{Hist}(#1)}
\newcommand{\historySetI}[1]{\mathsf{Hist}_\playerIndex(#1)}
\newcommand{\historySetOne}[1]{\mathsf{Hist}_1(#1)}

\newcommand{\history}{h}
\newcommand{\hist}{\history}
\newcommand{\last}[1]{\mathsf{last}(#1)}
\newcommand{\first}[1]{\mathsf{first}(#1)}
\newcommand{\suffix}[2]{#1_{\geq #2}}
\newcommand{\prefix}[2]{#1_{\leq #2}}
\newcommand{\concat}[2]{\ensuremath{#1 \cdot #2}}

\newcommand{\prefHist}{w}

\newcommand{\stratProfile}{\sigma}
\newcommand{\stratProfileAlt}{\tau}
\newcommand{\strat}[1]{\sigma_{#1}}
\newcommand{\stratI}{\stratProfile_\playerIndex}
\newcommand{\stratOne}{\stratProfile_1}
\newcommand{\stratTwo}{\stratProfile_2}
\newcommand{\stratIAdv}{\stratProfile_{-\playerIndex}}
\newcommand{\stratAltAdvI}{\stratProfileAlt_{-\playerIndex}}
\newcommand{\stratAltAdvIb}{\stratProfileAlt_{-\playerIndex'}}

\newcommand{\stratAltI}{{\stratProfileAlt_{\playerIndex}}}
\newcommand{\stratAltOne}{\stratProfileAlt_1}

\newcommand{\outcome}[2]{\mathsf{Out}(#1, #2)}
\newcommand{\mealyMachine}{\mathcal{M}}
\newcommand{\mealyStateSpace}{M}
\newcommand{\mealyState}{m}
\newcommand{\mealyStateInit}{\mealyState_\init}
\newcommand{\mealyUpdate}{\mathsf{up}}
\newcommand{\mealyNext}{\mathsf{nxt}}
\newcommand{\mealyNextI}{\mealyNext_\playerIndex}
\newcommand{\mealyNextIAlt}{\mathsf{nxt}_{\playerIndex'}}
\newcommand{\mealyNextIAltAlt}{\mathsf{nxt}_{\playerIndex''}}
\newcommand{\mealyTuplePure}{(\mealyStateSpace, \mealyStateInit, \mealyUpdate, \mealyNextI)}

\tikzset{
  >=stealth,
  left sided/.style={
    draw=none,
    append after command={
      [shorten <= -0.5\pgflinewidth]
      (\tikzlastnode.north west) edge[dashed](\tikzlastnode.south west)
    }
  },
  two sided/.style={
    draw=none,
    append after command={
      [shorten <= -0.5\pgflinewidth]
      (\tikzlastnode.north west) edge[dashed](\tikzlastnode.south west)
      (\tikzlastnode.north east) edge[dashed](\tikzlastnode.south east)
    }
  },
  right sided/.style={
    draw=none,
    append after command={
      [shorten <= -0.5\pgflinewidth]
      (\tikzlastnode.north east) edge[dashed](\tikzlastnode.south east)
    }
  }
}

\tikzstyle{stochasticc} = [fill, circle, minimum size=0.1cm, inner sep=0.05cm, outer sep=0cm]
\tikzstyle{stochastics} = [fill, rectangle, minimum size=0.1cm, inner sep=0.05cm, outer sep=0cm]

\makeatletter
\def\squarecorner#1{
\pgf@x=\the\wd\pgfnodeparttextbox \pgfmathsetlength\pgf@xc{\pgfkeysvalueof{/pgf/inner xsep}}\advance\pgf@x by 2\pgf@xc \pgfmathsetlength\pgf@xb{\pgfkeysvalueof{/pgf/minimum width}}\ifdim\pgf@x<\pgf@xb \pgf@x=\pgf@xb \fi \pgf@y=\ht\pgfnodeparttextbox \advance\pgf@y by\dp\pgfnodeparttextbox \pgfmathsetlength\pgf@yc{\pgfkeysvalueof{/pgf/inner ysep}}\advance\pgf@y by 2\pgf@yc \pgfmathsetlength\pgf@yb{\pgfkeysvalueof{/pgf/minimum height}}\ifdim\pgf@y<\pgf@yb \pgf@y=\pgf@yb \fi \ifdim\pgf@x<\pgf@y \pgf@x=\pgf@y \else
        \pgf@y=\pgf@x \fi
\pgf@x=#1.5\pgf@x \advance\pgf@x by.5\wd\pgfnodeparttextbox \pgfmathsetlength\pgf@xa{\pgfkeysvalueof{/pgf/outer xsep}}\advance\pgf@x by#1\pgf@xa \pgf@y=#1.5\pgf@y \advance\pgf@y by-.5\dp\pgfnodeparttextbox \advance\pgf@y by.5\ht\pgfnodeparttextbox \pgfmathsetlength\pgf@ya{\pgfkeysvalueof{/pgf/outer ysep}}\advance\pgf@y by#1\pgf@ya }
\makeatother

\pgfdeclareshape{square}{
    \savedanchor\northeast{\squarecorner{}}
    \savedanchor\southwest{\squarecorner{-}}

    \foreach \x in {east,west} \foreach \y in {north,mid,base,south} {
        \inheritanchor[from=rectangle]{\y\space\x}
    }
    \foreach \x in {east,west,north,mid,base,south,center,text} {
        \inheritanchor[from=rectangle]{\x}
    }
    \inheritanchorborder[from=rectangle]
    \inheritbackgroundpath[from=rectangle]
}

\title{Arena-Independent Memory Bounds for Nash Equilibria in Reachability Games}

\author{{James C.~A.} Main}{UMONS -- Université de Mons, Belgium}{}{https://orcid.org/0009-0000-8471-4833}{James {C.~A.} Main was supported by an F.R.S.-FNRS Research Fellowship during the preparation of this work..}

\authorrunning{{J.~C.~A.} Main} 

\Copyright{{James C.~A.} Main} 



\category{} 

\relatedversiondetails[cite={DBLP:conf/stacs/Main24}]{Conference version (STACS 2024)}{https://doi.org/10.4230/LIPIcs.STACS.2024.50}

\funding{This work has been supported by the Fonds de la Recherche Scientifique - FNRS under Grant n° T.0188.23 (PDR ControlleRS).} 

\acknowledgements{I thank Thomas Brihaye, Aline Goeminne and Mickael Randour for fruitful discussions and their comments on a preliminary version of this paper.}

\nolinenumbers

\EventEditors{John Q. Open and Joan R. Access}
\EventNoEds{2}
\EventLongTitle{42nd Conference on Very Important Topics (CVIT 2016)}
\EventShortTitle{CVIT 2016}
\EventAcronym{CVIT}
\EventYear{2016}
\EventDate{December 24--27, 2016}
\EventLocation{Little Whinging, United Kingdom}
\EventLogo{}
\SeriesVolume{42}
\ArticleNo{23}

\begin{document}
\maketitle
  
\begin{abstract}
  We study the memory requirements of Nash equilibria in turn-based multiplayer games on possibly \textit{infinite graphs} with reachability, safety, shortest-path, Büchi and co-Büchi objectives.
  
  We present constructions for \textit{finite-memory} Nash equilibria in these games that apply to arbitrary game graphs, bypassing the finite-arena requirement that is central in existing approaches.
  We show that, for these five types of games, from any Nash equilibrium, we can derive another Nash equilibrium where all strategies are finite-memory such that all objectives satisfied by the outcome of the original equilibrium also are by the outcome of the derived equilibrium, without increasing costs for shortest-path games.
  
  Furthermore, we provide memory bounds that are \textit{independent of the size of the game graph} for reachability, safety and shortest-path games.
  These bounds depend only on the number of players.
  
  To the best of our knowledge, we provide the first results pertaining to finite-memory constrained Nash equilibria in infinite arenas and the first arena-independent memory bounds for Nash equilibria.
\end{abstract}

\section{Introduction}\label{section:intro}
\subparagraph*{Games on graphs.}
\textit{Games on graphs} are a prevalent framework to model reactive systems, i.e., systems that continuously interact with their environment. 
Typically, this interaction is modelled as an infinite-duration \textit{two-player (turn-based) zero-sum game} played on an arena (i.e., a game graph) where a system player and an environment player are adversaries competing for opposing goals (e.g.,~\cite{DBLP:conf/dagstuhl/2001automata,DBLP:reference/mc/BloemCJ18,gog23}), which can be modelled, e.g., by numerical costs for the system player.
Determining whether the system can enforce some specification boils down to computing how low of a cost the system player can guarantee.
We then construct an \textit{optimal strategy} for the system which can be seen as a formal blueprint for a controller of the system to be implemented~\cite{rECCS,DBLP:reference/mc/BloemCJ18}.
For implementation purposes, strategies should have a finite representation.
We consider \textit{finite-memory strategies} (e.g.,~\cite{DBLP:journals/lmcs/BouyerLORV22,DBLP:journals/iandc/MainR24}) which are strategies defined by Mealy machines, i.e., finite automata with outputs on their edges.

\subparagraph*{Nash equilibria.}
Models in which players are not competing are also relevant for controller design: these are called \textit{non-zero-sum games}, in which there can be more than two players.
For instance, if the goal is to control a system consisting of several components, each with its own objective, it may prove too restrictive to assume that all components are adversarial to one another.
In this work, we focus on \textit{Nash equilibria}~\cite{Nash50}, a classical formalisation of rational behaviour in multiplayer games.
Intuitively, a Nash equilibrium is a contract between the players, described by one strategy per player, such that no player can benefit by unilaterally breaking it.
In the context of controller synthesis, as mentioned above, we are interested in finite-memory Nash equilibria.
The goal of this work is to understand how much memory is needed for Nash equilibria in given classes of games.

\subparagraph*{Reachability games.}
We study games on possibly \textit{infinite} arenas in which all players share objectives of the same type.
We focus on variants of \textit{reachability objectives}; reachability objectives constitute a fundamental class of objectives in reactive synthesis~\cite{DBLP:conf/fsttcs/BrihayeGMR23}.
A \textit{reachability objective} models the goal of visiting a given set of target vertices.
The complement of a reachability objective is called a \textit{safety objective}: it requires that a set of undesirable vertices (e.g., unsafe configurations of the system) be avoided.
A \textit{Büchi objective} requires that a target set be visited infinitely often instead of only once as for reachability.
The complement of a Büchi objective is called a \textit{co-Büchi objective}; it models the requirement that undesirable configurations eventually no longer occur.

We also study games with a quantitative variant of the reachability objective.
We consider arenas with edges labelled by non-negative integer weights (modelling, e.g., the passage of time), and define the cost of a player as the sum of weights prior to the first visit of a target, or infinity if the target is never reached.
This cost function is called the \textit{shortest-path} cost function.
The goal of a player is to minimise it (i.e., to reach a target as soon as possible).

NEs are guaranteed to exist in games with these objectives or cost functions.
For instance, see~\cite{DBLP:conf/lfcs/BrihayePS13,depril2013} for reachability games and for shortest-path games in finite arenas and~\cite{DBLP:conf/fsttcs/Ummels06} for safety, Büchi and co-Büchi games.
We also show in Appendix~\ref{appendix:existence} that there always exist an NE in shortest-path games on infinite arenas.

In finite arenas, it has also been shown that finite-memory NEs are guaranteed to exist in the games we consider (e.g.,~\cite{DBLP:conf/lfcs/BrihayePS13,DBLP:conf/fossacs/Ummels08,DBLP:journals/jcss/BrihayeBGT21}).
However, existing constructions yield finite-memory NEs whose memory has a size that depends on the size of the arena.
Therefore, these constructions do not directly generalise to infinite arenas.
The main idea of these approaches is as follows.
First, one shows that there exist plays resulting from NEs with a finite representation, typically, a lasso.
This play is then encoded in a Mealy machine.
If some player deviates from the play, the other players switch to a (finite-memory) punishing strategy to sabotage the deviating player; this enforces the stability of the equilibrium.
This punishment mechanism is inspired by the proof of the folk theorem for NEs in repeated games~\cite{Fri71,OR94}.

\subparagraph*{Contributions.}
Our contributions are twofold.
First, we present constructions for \textit{finite-memory} NEs for reachability, shortest-path, safety, Büchi and co-Büchi games that apply to arbitrary arenas, bypassing the finite-arena requirement that is central in existing approaches.
More precisely, for these types of games, we show that from any NE, we can derive another NE where all strategies are finite-memory and such that the objectives satisfied by the outcome of original NE also are by the outcome of the derived NE, without increasing their cost for shortest-path games.
In other words, our constructions are general and can be used to \textit{match or improve any NE cost profile}.

Second, for reachability, shortest-path and safety games, we provide memory bounds that are \textit{independent} of the size of the arena and depend only on the number of players.
For Büchi and co-Büchi games, we show that similar arena-independent memory bounds cannot be obtained: we provide a family of two-player games where NEs with an outcome winning for the second player require a memory whose size is linear in the size of the (finite) arena (Proposition~\ref{prop:buchi:dep}).
We summarise our memory upper bounds for each considered type of game in Table~\ref{table:nash memory}.

\begin{table}
  \centering
  \bgroup
  \def\arraystretch{1.2}
  \scalebox{1}{
    \begin{tabular}{|c|c|c|}
      \hline      
      Arena size
      & Finite & Infinite \\
      \hline
      Reachability & \multicolumn{2}{c|}{$\nPlayer^2$ (Thm.~\ref{thm:reach:ne})} \\
      \hline
      Shortest-path &  \multicolumn{2}{c|}{$\nPlayer^2+2\nPlayer$ (Thm.~\ref{thm:short:ne})} \\
      \hline
      Safety & \multicolumn{2}{c|}{$\nPlayer^2 + 2\nPlayer$ (Thm.~\ref{thm:safe:ne})} \\
      \hline
      Büchi & $|\vertexSet| + \nPlayer^2+\nPlayer$ (Thm.~\ref{thm:buchi:fmne:lasso}) & Finite-memory (Thm.~\ref{thm:buchi:fmne:lasso},~\ref{thm:buchi:fmne:nolasso})  \\
      \hline
      Co-Büchi & $|\vertexSet| + 2\nPlayer$ (Thm.~\ref{theorem:cobuchi:fm-ne}) & Finite-memory (Thm.~\ref{theorem:cobuchi:fm-ne})
      \\
      \hline
    \end{tabular}
  }
  \egroup
  \caption{Sufficient amount of memory to construct an NE that improves or matches the cost profile of an NE in $\nPlayer$-player games.
    The set $\vertexSet$ denotes the set of vertices of the arena.
  }\label{table:nash memory}
\end{table}

We build finite-memory NEs by using a relaxation of the classical punishment mechanism.
We build NEs from plays that can be decomposed into segments (e.g., finite histories, simple plays or ultimately periodic plays) -- see Lemmas~\ref{lem:spath:simple decomposition},~\ref{lem:safety:simple decomposition},~\ref{lem:buchi:simplelasso},~\ref{lem:buchi:simpleinfinite} and~\ref{lem:cobuchi:outcomes}.
Intuitively, we construct finite-memory strategies in which the players keep track of a segment of the decomposition (the one that they are currently attempting to construct) and the latest player to have moved.
While this memory structure does not allow players to detect deviations within a segment of the decomposition, we impose conditions on decompositions such that no such deviation is profitable.
The players switch to punishing strategies whenever there is a player who moves to a vertex outside of the current segment; the players retain enough information to detect such deviations.

Combined with the aforementioned NE existence results, our results imply that there always exist finite-memory NEs in the games we study.
Memory is necessary in general to construct an NE if we require that its outcome satisfies the objectives of some players (see Example~\ref{example:ne:roles of memory}); our results provide upper bounds on the memory needed in these cases.
Whether there always exists a memoryless NE in the games we consider is a closely related open question.

\subparagraph*{Related work.}
We refer to the survey~\cite{DBLP:conf/dlt/Bruyere17} for an extensive bibliography on games played on finite graphs, to~\cite{DBLP:conf/fsttcs/BrihayeGMR23} for a survey centred around reachability games and to~\cite{gog23} as a general reference on games on graphs.
We discuss three research directions related to this work.

The first direction is related to \textit{computational problems} for NEs.
In the settings we consider, NEs are guaranteed to exist.
However, NEs where no player satisfy their objective can coexist with NEs where all players satisfy their objective~\cite{DBLP:conf/fsttcs/Ummels06,DBLP:conf/fossacs/Ummels08}.
A classical problem is the \textit{constrained NE existence problem} which asks whether there exists an NE such that certain players satisfy their objective in the qualitative case or such that the cost incurred by players is bounded from above in the quantitative case (e.g.,~\cite{DBLP:conf/csl/BruyereMR14,DBLP:journals/corr/BouyerBMU15}).
For games on finite arenas, the constrained NE existence problem is \textsf{NP}-complete for reachability and shortest-path games~\cite{DBLP:journals/jcss/BrihayeBGT21}, \textsf{NP}-complete for safety games~\cite{DBLP:conf/icalp/ConduracheFGR16}, \textsf{NP}-complete in co-Büchi games~\cite{DBLP:conf/fossacs/Ummels08}, and is in \textsf{P} for Büchi games~\cite{DBLP:conf/fossacs/Ummels08}.
Constrained equilibrium existence problems have also been studied in models other than finite arenas, e.g., in concurrent games~\cite{DBLP:conf/fsttcs/BouyerBMU11,DBLP:phd/hal/Brenguier12} and timed games~\cite{DBLP:conf/concur/BouyerBM10,DBLP:conf/formats/BrihayeG20}.

Second, the construction of our finite-memory NEs rely on \textit{characterisations} of plays resulting from NEs.
Their purpose is to ensure that the punishment mechanism described above can be used to guarantee the stability of an equilibrium.
In general, these characterisations can be useful from an algorithmic perspective; solving the constrained NE existence problem boils down to finding a play that satisfies the characterisation.
Characterisations appear in the literature for NEs~\cite{DBLP:conf/fossacs/Ummels08,DBLP:conf/concur/UmmelsW11,DBLP:journals/corr/BouyerBMU15}, but also for other types of equilibria, e.g., subgame perfect equilibria~\cite{DBLP:journals/lmcs/BrihayeBGRB20} and secure equilibria~\cite{DBLP:conf/csl/BruyereMR14}.

Finally, there exists a body of work dedicated to better understanding the complexity of optimal strategies in zero-sum games.
We mention~\cite{GZ05} for memoryless strategies, and~\cite{DBLP:journals/lmcs/BouyerLORV22} and~\cite{DBLP:journals/theoretics/BouyerRV23} for finite-memory strategies in finite and infinite arenas respectively.
In finite arenas, for the finite-memory case, a key notion is \textit{arena-independent} finite-memory strategies, i.e., strategies based on a memory structure that is sufficient to win in all arenas whenever possible.
In this work, the finite-memory strategies we propose actually depend on the arena; only their size does not.
We also mention~\cite{DBLP:journals/iandc/RouxP18}: in games on finite arenas with objectives from a given class, finite-memory NEs exist if certain conditions on the corresponding zero-sum games hold.

\subparagraph*{Outline.}
This paper is an extended version of the conference paper~\cite{DBLP:conf/stacs/Main24}.
It presents the contributions of the conference paper with detailed proofs, and additionally presents results for games with safety and co-Büchi objectives.
This work is structured as follows.
In Section~\ref{section:prelim}, we summarise prerequisite definitions.
We establish the existence of memoryless punishing strategies by studying zero-sum games in Section~\ref{section:zero-sum}.
Characterisations of outcomes of NEs are provided in Section~\ref{section:outcomes}.
We prove our main results on finite-memory NEs in reachability and shortest-path games in Section~\ref{section:reach}.
Section~\ref{section:safety} presents our results for safety games.
Finally, Section~\ref{section:buchi} is dedicated to the corresponding result for Büchi and co-Büchi games.

\section{Preliminaries}\label{section:prelim}
\subparagraph*{Notation.} We write $\IN$, $\IR$ for the sets of natural and real numbers respectively, and let $\IRbar =\IR\cup \{+\infty, -\infty\}$ and $\INbar =\IN\cup \{+\infty\}$. For any $\nPlayer\in\IN$, $\nPlayer\geq 1$, we let $\integerInterval{\nPlayer} = \{1, \ldots, \nPlayer\}$ denote the set of positive integers up to $\nPlayer$.

\subparagraph*{Arenas and plays.} 
Let $(\vertexSet, \edgeSet)$ be a directed graph where $\vertexSet$ is a (possibly infinite) set of vertices and $\edgeSet\subseteq \vertexSet\times\vertexSet$ is an edge relation.
For any $\vertex\in\vertexSet$, we write $\succSet{\vertex}=\{\vertex'\in\vertexSet\mid (\vertex, \vertex')\in\edgeSet\}$ for the set of successor vertices of $\vertex$.
Let $\nPlayer\in\IN$.
An \textit{$\nPlayer$-player arena} (on $(\vertexSet, \edgeSet)$) is a tuple $\arena= \arenaTuple$, where $(\vertexSetI)_{\playerIndex\in\playerSet}$ is a partition of $\vertexSet$ and, for all $\playerIndex\in\playerSet$, we say that the vertices in $\vertexSetI$ are controlled by $\playerI$.
We assume that there are no deadlocks in the arenas we consider, i.e., for all $\vertex\in\vertexSet$, $\succSet{\vertex}$ is not empty.
We write $\playerI$ for player $\playerIndex$.

A play starts in an initial vertex and proceeds as follows.
At each round of the game, the player controlling the current vertex selects a successor of this vertex and the current vertex is updated accordingly. 
The play continues in this manner infinitely.
Formally, a \textit{play} of $\arena$ is an infinite sequence $\vertex_0\vertex_1\ldots\in\vertexSet^\omega$ such that $(\vertex_\indexPosition, \vertex_{\indexPosition+1})\in E$ for all $\indexPosition\in\IN$.
For a play $\play=\vertex_0\vertex_1\ldots$ and $\indexPosition\in\IN$, we let $\suffix{\play}{\indexPosition} = \vertex_\indexPosition\vertex_{\indexPosition+1}\ldots$ denote the suffix of $\play$ from position $\indexPosition$ and $\prefix{\play}{\indexPosition}=\vertex_0\ldots\vertex_\indexPosition$ denote the prefix of $\play$ up to position $\indexPosition$.
A \textit{history} is any finite non-empty prefix of a play.
We write $\playSet{\arena}$ and $\historySet{\arena}$ for the set of plays and histories of $\arena$ respectively.
For $\playerIndex\in\playerSet$, we let $\historySetI{\arena} = \historySet{\arena}\cap \vertexSet^*\vertexSetI$.
For any history $\history = \vertex_0\ldots\vertex_\indexLast$, we let $\first{\history}$ and $\last{\history}$ respectively denote $\vertex_0$ and $\vertex_\indexLast$.
For any play $\play\in\playSet{\arena}$, $\first{\play}$ is defined similarly as the first state of $\play$.

Given two histories $\history =\vertex_0\ldots\vertex_\indexPosition$ and $\history' = \vertex_\indexPosition\vertex_{\indexPosition+1}\ldots\vertex_{\indexLast}$, we let $\concat{\history}{\history'} = \vertex_0\ldots\vertex_\indexPosition\vertex_{\indexPosition+1}\ldots\vertex_{\indexLast}$; we say that $\concat{\history}{\history'}$ is the \textit{combination} of $\history$ and $\history'$.
The combination $\concat{\history}{\play}$ of a history $\history$ and a play $\play$ such that $\last{\history}=\first{\play}$ is defined similarly.

A play or history is \textit{simple} if no vertex occurs twice (or more) within.
A \textit{cycle} is a history of the form $\hist = \hist'\first{\hist'}$ where $\hist'$ is a history, and $\hist=\hist'\first{\hist'}$ is a \textit{simple cycle} if $\hist'$ is a simple history.
A play $\play$ is a \textit{lasso} if there exist $p\in\vertexSet^*$ and $c\in\historySet{\arena}$ such that $\play = pc^\omega$.
For any $p\in\vertexSet^*$ and $c\in\historySet{\arena}$ such that $pc^\omega\in\playSet{\arena}$, we say that the lasso $pc^\omega$ is simple if $pc$ is a simple history.

A \textit{segment} of a play $\play$ is either a suffix $\suffix{\play}{\indexPosition}$ of $\play$ for some $\indexPosition\in\IN$ or any history of the form $\vertex_\indexPosition\ldots\vertex_{\indexPosition'}$ where $\indexPosition\leq\indexPosition'$ (i.e., an infix of $\play$).
We denote segments by $\segment$ to avoid distinguishing finite and infinite segments of plays in the following.
A segment is \textit{simple} if it is a simple history, a simple play or a simple lasso.
A segment is \textit{trivial} if it consists of a single vertex.

A \textit{(segment) decomposition} of $\play\in\playSet{\arena}$ is a (finite or infinite) sequence $\decomp= (\segment_\indexSegment)_{\indexSegment=1}^{\numSegments}$ of segments of $\play$ such that, for all $\indexSegment < \numSegments$, $\last{\segment_\indexSegment}=\first{\segment_{j+1}}$, and $\play = \concat{\concat{\segment_1}{\ldots}}{\segment_\numSegments}$ if $\numSegments\in\IN$ and $\play = \concat{\concat{\segment_1}{\segment_2}}{\ldots}$ if $\numSegments=\infty$.
The segment decomposition $\decomp$ is \textit{simple} if all segments within are simple.
A decomposition is \textit{trivial} if it contains only one element.

\subparagraph*{Strategies.}
Strategies describe the decisions of players during a play.
These choices may depend on the entire past of the play, i.e., not only on the current vertex of the play.
Formally, a \textit{strategy} of $\playerI$ in an arena $\arena$ is a function $\stratI\colon\historySetI{\arena}\to \vertexSet$ such that for all histories $\history\in\historySetI{\arena}$, $(\last{\history}, \stratI(\history))\in E$.

A \textit{strategy profile} is a tuple $\stratProfile = (\stratI)_{\playerIndex\in\playerSet}$, where $\stratI$ is a strategy of $\playerI$ for all $\playerIndex\in\playerSet$.
To highlight the role of $\playerI$, we sometimes write strategy profiles as $\stratProfile = (\stratI, \stratIAdv)$, where $\stratIAdv$ denotes the strategy profile of the players other than $\playerI$.

A play $\play = \vertex_0\vertex_1\vertex_2\ldots$ is \textit{consistent} with a strategy $\stratI$ of $\playerI$ if for all $\indexPosition\in\IN$, $\vertex_\indexPosition\in\vertexSetI$ implies $\vertex_{\indexPosition+1} = \stratI(\prefix{\play}{\indexPosition})$.
A play is \textit{consistent} with a strategy profile if it is consistent with all strategies of the profile.
Given an initial vertex $\vertex_0$ and a strategy profile $\stratProfile$, there is a unique play $\outcome{\stratProfile}{\vertex_0}$ from $\vertex_0$ that is consistent with $\stratProfile$, called the \textit{outcome} of $\stratProfile$ from $\vertex_0$.

There are two classes of strategies of interest in this work.
A strategy $\stratI$ is \textit{memoryless} if the moves it prescribes depend only on the current vertex, i.e., if for all $\history$, $\history'\in\historySetI{\arena}$, if $\last{\history} = \last{\history'}$, then $\stratI(\history)=\stratI(\history')$.
We view memoryless strategies as functions $\vertexSetI\to\vertexSet$.

A strategy is \textit{finite-memory} if it can be encoded by a Mealy machine, i.e., a finite automaton with outputs.
A \textit{Mealy machine} (of $\playerI$) is a tuple $\mealyMachine =\mealyTuplePure$ where $\mealyStateSpace$ is a finite set of memory states, $\mealyStateInit$ is an initial memory state, $\mealyUpdate\colon \mealyStateSpace\times\vertexSet\to\mealyStateSpace$ is a memory update function and $\mealyNextI\colon\mealyStateSpace\times\vertexSetI\to\vertexSet$ is a next-move function.

To describe the strategy induced by a Mealy machine, we first define the iterated update function $\widehat{\mealyUpdate}\colon \vertexSet^*\to \mealyStateSpace$ by induction.
We write $\varepsilon$ for the empty word.
We let $\widehat{\mealyUpdate}(\varepsilon) = \mealyStateInit$ and for all $w\in \vertexSet^*$ and $\vertex\in\vertexSet$, we let $\widehat{\mealyUpdate}(\prefHist\vertex) = \mealyUpdate(\widehat{\mealyUpdate}(\prefHist), \vertex)$.
The strategy $\stratI^\mealyMachine$ induced by $\mealyMachine$ is defined, for all histories $\history = \history'\vertex\in\historySetI{\arena}$, by $\stratI^\mealyMachine(\history) = \mealyNextI(\widehat{\mealyUpdate}(\history'), \vertex)$.

We say that a finite-memory strategy $\stratI$ has \textit{memory size at most} $b\in\IN$ if there is some Mealy machine $\mealyTuplePure$ inducing $\stratI$ with $|\mealyStateSpace| \leq b$.

We say that a strategy profile is memoryless (resp.~finite-memory) if all strategies in the profile are memoryless (resp.~finite-memory).

\begin{remark}
  We define Mealy machines with vertex-based updates (as in, e.g.,~\cite{DBLP:journals/acta/ChatterjeeRR14,DBLP:reference/mc/BloemCJ18,DBLP:journals/jcss/BrihayeBGT21}).
  Some authors define the update functions of Mealy machines with edges instead of vertices (see, e.g.,~\cite{DBLP:journals/lmcs/BouyerLORV22,gog23}).
  Vertex-update Mealy machines can be seen as a special case of edge-update Mealy machines.

  In \textit{finite arenas}, these two variants of Mealy machines share the same \textit{expressiveness}.
  Given an edge-update Mealy machine, we can transform it into an equivalent vertex-update Mealy machine by augmenting the memory states with the last vertex to have been visited.
  With this last vertex and the current one, the last edge can be inferred.
  This allows memory updates to be imitated one step later than in the original edge-update Mealy machine.
  In general, edge-update Mealy machines may be more concise than vertex-update ones as they have access to more information in memory updates.
  It can be shown that transforming an edge-update Mealy machine into an equivalent vertex-update Mealy machine may require a vertex-update Mealy machine whose size is proportional to the product of the sizes of the edge-update Mealy machine size and the arena.

  In \textit{infinite arenas}, these two Mealy machine models are no longer equivalent.
  For instance, an edge-based Mealy machine can be used to check if a self-loop is taken in a vertex, but a vertex-update Mealy machine cannot, as this could require memorising the previous state.
  The construction described above for finite arenas fails as it yields an infinite memory state space.

  In the sequel, we show that finite-memory strategies induced by \textit{vertex-update Mealy machines} suffice for given specifications and provide upper bounds on the sufficient memory size (smaller bounds can be obtained with edge-update Mealy machines).
  As we consider possibly infinite arenas, this is a stronger result that showing that edge-update Mealy machines suffice for these specifications.
  Furthermore, our arena-independent memory bounds for vertex-update Mealy machines cannot not be obtained by first considering edge-update Mealy machines due to the possible blow-up in size whenever a translation is possible.
  \hfill$\lhd$
\end{remark}

\subparagraph*{Games.}
We formalise the goal of a player in two ways.
In the qualitative case, we describe the goal of a player by a set of plays, called an \textit{objective}.
We say that a play $\play$ satisfies an objective $\objective$ if $\play\in\objective$.
For quantitative specifications, we assign a quantity to each play using a \textit{cost function} $\costI\colon \playSet{\arena}\to \overline{\IR}$ that $\playerI$ intends to minimise.
Any goal expressed by an objective $\objective$ can be encoded using a cost function $\costI$ which assigns $0$ to plays in $\objective$ and $1$ to others (i.e., the indicator of $\playSet{\arena}\setminus\objective$); aiming to minimise this cost is equivalent to aiming to satisfy the objective.
To avoid redundancy, we present definitions that are applicable for both objectives and cost functions by using cost functions, and separately present notions that are specific to objectives.

A \textit{game} is an arena augmented with the goals of each player.
Formally, a game is a tuple $\game=\gameTuple$ where $\arena$ is an $\nPlayer$-player arena and, for all $\playerIndex\in\playerSet$, $\costI$ is the cost function of $\playerI$.
The \textit{cost profile} of a play $\play$ in $\game$ is $(\costI(\play))_{\playerIndex\in\playerSet}$.
Given two plays $\play$ and $\play'$, we say that the cost profile of $\play$ is preferable to that of $\play'$ if $\costI(\play)\leq \costI(\play')$ for all $\playerIndex\in\playerSet$.
When the goal of each player is given by an objective, we denote games as tuples $\game = (\arena, (\objective_\playerIndex)_{\playerIndex\in\playerSet})$ where $\objective_\playerIndex$ is the objective of $\playerI$ for all $\playerIndex\in\playerSet$.

\subparagraph*{Objectives and costs.}
We study five classes of games with goals related to \textit{reachability}.
Let $\target\subseteq\vertexSet$ denote a set of target vertices.
We refer to the set of target vertices $\target$ as a \textit{target}.

The reachability objective formalises the goal of reaching the target.
Formally, the \textit{reachability objective} (for $\target$) $\reach{\target}$ is defined by $\{\vertex_0\vertex_1\vertex_2\ldots\in\playSet{\arena}\mid\exists\,\indexPosition\in\IN,\, \vertex_\indexPosition\in\target\}$.
The complement of the reachability objective $\safe{\target}=\playSet{\arena}\setminus\reach{\target}$ is called the \textit{safety} objective (for $\target$) which requires avoiding $\target$.

Next, we introduce a cost function for the goal of reaching a target as soon as possible.
In this context, we assign (non-negative integer) weights to edges via a \textit{weight function} $\weight\colon E\to \IN$, which model, e.g., the time taken when traversing an edge.
The weight function is extended to histories as follows: for $\history = \vertex_0\ldots\vertex_\indexLast\in\historySet{\arena}$, we let $\weight(\history)= \sum_{\indexPosition=0}^{\indexLast-1}\weight((\vertex_\indexPosition, \vertex_{\indexPosition+1}))$.
We define the \textit{shortest-path} cost function (for $\target$ and $\weight$), for all plays $\play = \vertex_0\vertex_1\ldots\in\playSet{\arena}$, by $\costReach{\target}{\weight}(\play) = \weight(\prefix{\play}{\indexLast})$ if $\indexLast = \min\{\indexPosition\in\IN\mid \vertex_{\indexPosition}\in\target\}$ exists and $\costReach{\target}{\weight}(\play)=+\infty$ otherwise.
By assigning an infinite cost to plays that do not visit $\target$, we model the idea that not reaching $\target$ is the least desirable situation.

Finally, the Büchi objective requires reaching the target infinitely often.
Formally, the \textit{Büchi objective} (for $\target$) $\buchi{\target}$ is defined by $\{\vertex_0\vertex_1\vertex_2\ldots\in\playSet{\arena}\mid\forall\,\indexPosition\in\IN,\,\exists\,\indexPosition'\geq\indexPosition,\, \vertex_{\indexPosition'}\in\target\}$.
The complement of a Büchi objective is a co-Büchi objective: the \textit{co-Büchi} objective (for $\target$) is defined as $\cobuchi{\target}=\playSet{\arena}\setminus\buchi{\target}$ and requires visiting $\target$ finitely often.

We refer to games where all players have a reachability (resp.~safety, Büchi, co-Büchi) objective as \textit{reachability} (resp.~\textit{safety}, \textit{Büchi}, \textit{co-Büchi}) \textit{games} and games where all players have a shortest-path cost function as \textit{shortest-path games}.

We introduce some notation that is useful in the sequel.
Let $\game$ be a game on the arena $\arena$ and let $\play\in\playSet{\arena}$ be a play.
First, assume that $\game = (\arena, (\objective_\playerIndex)_{\playerIndex\in\playerSet})$ is a game in which the goals of all players are modelled by objectives.
We let $\satpl{\play} = \{\playerIndex\in\playerSet\mid\play\in\objective_\playerIndex\}$ denote the set of players whose objectives are satisfied by $\play$.
Second, assume that $\game = (\arena, (\spathIW)_{\playerIndex\in\playerSet})$ is a shortest-path game where $\target_1, \ldots, \target_\nPlayer\subseteq\vertexSet$ are targets for each player and $\weight_1$, \ldots, $\weight_\nPlayer\colon\edgeSet\to\IN$ are weight functions for all players.
In this case, we let $\satpl{\play} = \{\playerIndex\in\playerSet\mid \play\in\reach{\target_\playerIndex}\}$ denote the set of players whose targets are visited in $\play$.
Finally, if $\game$ is a reachability, safety or shortest-path game in which $\target_\indexPlayer\subseteq\vertexSet$ is the target of $\playerI$ for all $\playerIndex\in\playerSet$, we let $\vispos{\play} = \{\min\{\indexPosition\in\IN\mid \vertex_\indexPosition\in\target_\playerIndex\}\mid \playerIndex\in\satpl{\play}\}$ denote the set of earliest positions at which targets are visited along $\play$.

\subparagraph*{Nash equilibria.}
Let $\game = (\arena, \costVector)$ be a game and $\vertex_0$ be an initial vertex.
Given a strategy profile $\stratProfile = (\stratI)_{\playerIndex\in\playerSet}$, we say that a strategy $\stratAltI$ of $\playerI$ is a \textit{profitable deviation} (with respect to $\stratProfile$ from $\vertex_0$) if $\costI(\outcome{(\stratAltI, \stratIAdv)}{\vertex_0}) < \costI(\outcome{\stratProfile}{\vertex_0})$.
A \textit{Nash equilibrium} (NE) from $v_0$ is a strategy profile such that no player has a profitable deviation.
Equivalently, $\stratProfile$ is an NE from $\vertex_0$ if, for all $\playerIndex\in\playerSet$ and all plays $\play$ consistent with $\stratIAdv$ starting in $\vertex_0$, $\costI(\play) \geq \costI(\outcome{\stratProfile}{\vertex_0})$.
For the sake of conciseness, we refer to the cost profile of the outcome of an NE as the \textit{cost profile of the NE}.

NEs are guaranteed to exist in the five classes of games we consider; see, e.g.~\cite{DBLP:conf/lfcs/BrihayePS13,depril2013} for reachability games and for shortest-path games in finite arenas, Appendix~\ref{appendix:existence} for shortest-path games in general and~\cite{DBLP:conf/fsttcs/Ummels06} for safety, Büchi and co-Büchi games. Even though NEs exist from all initial vertices, there are no guarantees on the quality of their cost profiles (this is a core motivation of the constrained NE existence problem mentioned in Section~\ref{section:intro}).
In a given game, there can exist several NEs from a given initial vertex with incomparable cost profiles.
Similarly, in a given game with objectives, from a fixed initial vertex, there can exist NEs with which all players satisfy their objective and NEs with which no players satisfy their objective~\cite{DBLP:conf/fsttcs/Ummels06,DBLP:conf/fossacs/Ummels08}.
In general, memory may be necessary to construct an NE with a cost profile that is preferable to a given vector (i.e., a solution to the constrained existence problem).

We provide two illustrative examples below.
First, we provide a shortest-path game in which there exists two NEs with incomparable cost profiles, i.e., there need not exist an NE with a minimal cost profile (with respect to the componentwise ordering).
Second, we provide a reachability game in which memory is needed to obtain an NE such that a given subset of players win.

\begin{figure}
  \centering
  \begin{subfigure}[b]{0.48\textwidth}
    \centering
    \begin{tikzpicture}[node distance=8mm]
      \node[state, align=center] (v0) {$\vertex_0$};
      \node[state, align=center, right = of v0] (t12) {$t_{12}$};
      \node[state, square, align=center, right = of t12] (v1) {$\vertex_1$};
      \node[state, align=center, below = of v1] (v2) {$\vertex_2$};
      \node[state, align=center, right = of v1] (t1) {$t_1$};
      \path[->] (v0) edge node[below] {$3$} (t12);
      \path[->] (t12) edge (v1);
      \path[->] (v0) edge[bend left] (v1);
      \path[->] (t1) edge[bend left] (v1);
      \path[->] (v1) edge[bend left] (t1);
      \path[->] (v1) edge (v2);
      \path[->] (v2) edge [loop left] (v2);
    \end{tikzpicture}
    \caption{A two-player weighted arena in which there are two memoryless NE with incomparable payoff profiles.
      Edges are labelled by their weight and unlabelled edges have a weight of $1$.}
    \label{fig:ex:bonus segment}
  \end{subfigure}
  \hfill
  \begin{subfigure}[b]{0.48\textwidth}
    \centering
    \begin{tikzpicture}[node distance=0.4cm]
      \node[state, align=center] (v0) {$\vertex_0$};
      \node[state, square, align=center, left = of v0] (t2) {$t_2$};
      \node[state, diamond, align=center, right = of v0] (v1) {$\vertex_1$};
      \node[state, regular polygon, regular polygon sides=6, align=center, right = of v1] (v2) {$\vertex_2$};
      \node[state, align=center, right = of v2] (t1) {$t_1$};
      \node[inner sep=0cm, node distance=0.2cm, right = of v1] (ref) {};
      \node[state, align=center, node distance=12mm, below = of ref] (v5) {$\vertex_3$};
      \node[state, diamond, align=center, left = of v5] (t3) {$t_3$};
      \node[state, regular polygon, regular polygon sides=6, align=center, right = of v5] (t4) {$t_4$};
      \path[->] (v0) edge[bend left] (v1);
      \path[->] (v1) edge[bend left] (v0);
      \path[->] (v1) edge[bend left] (v2);
      \path[->] (v2) edge[bend left] (v1);
      \path[->] (v2) edge[bend left] (t1);
      \path[->] (t1) edge[bend left] (v2);
      \path[->] (v0) edge (t2);
      \path[->] (v1) edge (v5);
      \path[->] (v2) edge (v5);
      \path[->] (v5) edge (t3);
      \path[->] (v5) edge (t4);
      \path[->] (t2) edge[loop below] (t2);
      \path[->] (t3) edge[loop left] (t3);
      \path[->] (t4) edge[loop right] (t4);
    \end{tikzpicture}
    \caption{An arena in which memory is required to obtain an NE from $\vertex_0$ with which $\playerOne$ and $\playerTwo$ win in the reachability game where the target of $\playerI$ is $\{t_\indexPlayer\}$ for all $\indexPlayer\in\integerInterval{4}$.}\label{figure:roles of memory}
  \end{subfigure}
  \hfill

  \caption{Two arenas. Circles, squares, diamonds and hexagons respectively denote $\playerOne$, $\playerTwo$, $\player{3}$ and $\player{4}$ vertices.
  }
\end{figure}

\begin{example}  
  Consider the shortest-path game played on the arena depicted in Figure~\ref{fig:ex:bonus segment} where $\target_1= \{t_{12}, t_1\}$ and $\target_2 = \{t_{12}\}$ are the respective targets of $\playerOne$ and $\playerTwo$.
  The memoryless strategy profile $(\stratOne, \stratTwo)$ with $\stratOne(\vertex_0)=t_{12}$ and $\stratTwo(\vertex_1)=\vertex_2$ is an NE from $\vertex_0$ with cost profile $(3, 3)$.
  To check that it is an NE, we observe that if $\playerOne$ deviates and moves from $\vertex_0$ to $\vertex_1$, then their target is not visited and the target of $\playerTwo$ is visited before they act, and thus their cost cannot change even if they deviate.

  Another NE from $\vertex_0$ would be the memoryless strategy profile $(\stratOne', \stratTwo')$ such that $\stratOne'(\vertex_0)=\vertex_1$ and $\stratTwo'(\vertex_1)= t_1$.
  The cost profile of this NE is $(2, +\infty)$, which is incomparable with the cost profile $(3, 3)$ of the above NE.
  \hfill $\lhd$
\end{example}

\begin{example}\label{example:ne:roles of memory}
  We consider the four-player reachability game played on the arena of Figure~\ref{figure:roles of memory} where the target of $\playerI$ is $\{t_\indexPlayer\}$ for all $\indexPlayer\in\integerInterval{4}$.
  We show that there exists an NE from $\vertex_0$ such that $t_1$ and $t_2$ are visited in its outcome, and that any such NE requires memory.

  Intuitively, it suffices to construct a strategy profile resulting in the outcome $\vertex_0\vertex_1\vertex_2t_1\vertex_2\vertex_1\vertex_0t_2^\omega$ such that the decisions of $\playerOne$ in $\vertex_3$ prevent $\player{3}$ and $\player{4}$ from having a profitable deviation.
  This can be ensured by having $\playerOne$ avoiding the target of the owner of the vertex from which $\vertex_3$ was reached (as they must have deviated).
  Formally, we consider a strategy $\stratOne$ of $\playerOne$ such that $\stratOne(\vertex_0) = \vertex_1$, $\stratOne(\vertex_0\vertex_1\vertex_2t_1\vertex_2\vertex_1\vertex_0) = t_2$ and for all histories $\hist$ such that $\last{\hist}=\vertex_3$, we have $\stratOne(\hist) = t_4$ if the penultimate vertex of $\hist$ is $\vertex_1$ and $\stratOne(\hist) = t_3$ otherwise (in particular, when the penultimate vertex of $\hist$ is $\vertex_2$).
  We let $\stratTwo$ denote the unique strategy of $\playerTwo$ in the arena.
  For $\player{3}$, we consider a strategy $\strat{3}$ such that $\strat{3}(\vertex_0\vertex_1) = \vertex_2$ and $\strat{3}(\vertex_0\vertex_1\vertex_2t_1\vertex_2\vertex_1) = \vertex_0$.
  Finally, for $\player{4}$, we let $\strat{4}$ be such that $\strat{4}(\vertex_0\vertex_1\vertex_2) = t_1$ and $\strat{4}(\vertex_0\vertex_1\vertex_2t_1\vertex_2) = \vertex_1$.
  The strategy profile $\stratProfile = (\stratOne, \stratTwo, \strat{3}, \strat{4})$ is an NE from $\vertex_0$ with an outcome satisfying the objectives of $\playerOne$ and $\playerTwo$.
  The two players whose objectives are satisfied have no profitable deviations.
  The other two players can only visit their target by going through $\vertex_3$ and the strategy of $\playerOne$ prevents any deviation moving to $\vertex_3$ from being profitable.

  Without memory, it is not possible to obtain an NE from $\vertex_0$  such that  that $t_1$ and $t_2$ occur in its outcome.
  To visit both $t_1$ and $t_2$, memory is necessary as there is no simple history starting in $\vertex_0$ that visits both vertices (outcomes of memoryless strategy profiles are simple lassos).
  Therefore, memory is required to visit targets that cannot be connected via a simple history from the initial vertex.
  Memory is also required for another reason: it prevents $\player{3}$ and $\player{4}$ from having profitable deviations.
  Indeed, if $\playerOne$ plays without memory, then one of $\player{3}$ or $\player{4}$ would have an incentive to move to $\vertex_3$ to fulfil their objective.
  Thus, memory is useful to prevent the existence of profitable deviations (in this case via a punishment mechanism for players that deviate from the intended outcome).
  The role of memory for constrained NEs is explained further in~\cite{DBLP:conf/fsttcs/BrihayeGMR23}.
  \hfill$\lhd$
\end{example}

The previous example shows that for Pareto optimal NEs (i.e., NEs whose cost profiles cannot be improved for a player without harming the cost of another), memory is necessary in general.
Whether there always exists a memoryless NE in the games we consider is an open question.

\subparagraph*{Zero-sum games.} In a zero-sum game, two players compete with opposing goals.
Formally, a two-player game $\game = (\arena, (\cost{1}, \cost{2}))$ is a \textit{zero-sum game} if $\cost{2} = -\cost{1}$.
We usually shorten the notation of a zero-sum game to $\game = (\arena, \cost{1})$ due to the definition.

Let $\vertex_0\in V$.
If
$\inf_{\stratOne}\sup_{\stratTwo}\costOne(\outcome{(\stratOne, \stratTwo)}{\vertex_0}) = \sup_{\stratTwo}\inf_{\stratOne}\costOne(\outcome{(\stratOne, \stratTwo)}{\vertex_0}),$
where $\stratI$ is quantified over the strategies of $\playerI$, we refer to the above as the \textit{value} of $v_0$ and denote it by $\val(v_0)$.
A game is \textit{determined} if the value is defined for all vertices.

A strategy $\stratOne$ of $\playerOne$ (resp.~$\stratTwo$ of $\playerTwo$) is said to ensure $\alpha\in\overline{\IR}$ from a vertex $\vertex_0$ if all plays $\pi$ consistent with $\stratOne$ (resp.~$\stratTwo$) from $\vertex_0$ are such that $\costOne(\pi)\leq \alpha$ (resp.~$\costOne(\pi)\geq \alpha$).
A strategy of $\playerI$ is \textit{optimal} from $\vertex_0\in\vertexSet$ if it ensures $\val(\vertex_0)$ from $\vertex_0$.
A strategy is a \textit{uniformly optimal} strategy if it ensures $\val(\vertex)$ from $\vertex$ for all $\vertex\in\vertexSet$.
Optimal strategies do not necessarily exist, even if the value does.

\begin{figure}
  \centering
  \begin{tikzpicture}[node distance=0.45cm]
    \node[state, align=center] (v) {$v_0$};
    \node[state, square, align=center, right = of v] (inf) {$v_\infty$};
    \node[state, align=center, below = of inf] (one) {$v_1$};
    \node[state, align=center, right = of one] (two) {$v_2$};
    \node[state, align=center, right = of two] (three) {$v_3$};
    \node[align=center, right = of three] (dots) {\ldots};
    \node[state, align=center, left = of one] (target) {$t$};
\path[<->] (one) edge (inf);
\path[<->] (two) edge (inf);
\path[<->] (three) edge (inf);
    \path[->] (two) edge (one);
    \path[->] (three) edge (two);
    \path[->] (one) edge (target);
    \path[->] (dots) edge (three);
    \path[<->] (dots) edge[bend right] (inf);
    \path[->] (v) edge (inf);
    \path[->] (v) edge[loop left] (v);
    \path[->] (target) edge[loop left] (target);
  \end{tikzpicture}
  \caption{An infinite weighted arena where there is no $\playerTwo$ optimal strategy from $\vertex_\infty$ in the zero-sum shortest path game with $\target=\{t\}$. All edges have a weight of $1$. Circle and squares respectively denote $\playerOne$ and $\playerTwo$ vertices.}
  \label{fig:ex:non-opt}
\end{figure}

\begin{example}\label{ex:no opti}
  Consider the two-player zero-sum game played on the weighted arena illustrated in Figure~\ref{fig:ex:non-opt} where the cost function of $\playerOne$ is $\costReach{\{t\}}{\weight}$.
  Let $\alpha\in\IN\setminus\{0\}$.
  It holds that $\val(\vertex_\alpha)=\alpha$.
  On the one hand, $\player{1}$ can ensure a cost of $\alpha$ from $\vertex_\alpha$ by moving leftward in the illustration.
  On the other hand, $\player{2}$ can ensure a cost of $\alpha$ from $\vertex_\alpha$ with the memoryless strategy that moves from $\vertex_\infty$ to $\vertex_\alpha$.
  It follows that this same memoryless strategy of $\player{2}$ ensures $\alpha+1$ from $\vertex_\infty$.
  We conclude that $\val(\vertex_\infty)=+\infty$.
  However, $\playerTwo$ cannot prevent $t$ from being reached from $\vertex_\infty$, despite its infinite value.
  Therefore, $\playerTwo$ does not have an optimal strategy.
  \hfill$\lhd$
\end{example}

If the goal of $\player{1}$ is formulated by an objective $\objective$, we say that a strategy $\stratOne$ of $\playerOne$ (resp.~$\stratTwo$ of $\playerTwo$) is \textit{winning} from $\vertex_0$ if all plays consistent with it from $\vertex_0$ satisfy $\objective$ (resp.~$\playSet{\arena}\setminus\objective$).
The set of vertices from which $\playerOne$ (resp.~$\playerTwo$) has a winning strategy is called their \textit{winning region} denoted by $\winningOne{\objective}$ (resp.~$\winningTwo{\playSet{\arena}\setminus\objective}$).
A strategy $\stratOne$ of $\playerOne$ (resp.~$\stratTwo$ of $\playerTwo$) is a \textit{uniformly winning} strategy if it is winning from all vertices in $\winningOne{\objective}$ (resp.~$\winningTwo{\playSet{\arena}\setminus\objective}$).

Given an $\nPlayer$-player game $\game = \gameTuple$ where $\arena = \arenaTuple$, we define the \textit{coalition game (for $\playerI$)} as the game opposing $\playerI$ to the coalition of the other players, formally defined as the two-player zero-sum game $\game_\playerIndex= (\arena_\playerIndex, \costI)$ where $\arena_\playerIndex = ((\vertexSetI, \vertexSet\setminus\vertexSetI), \edgeSet)$.
We write $\playerIAdv$ to refer to the coalition of players other than $\playerI$

As we have done for multiplayer nonzero-sum games, we refer to two-player zero-sum games where the objective of $\playerOne$ is a reachability objective as a zero-sum reachability game, and similarly for the other objectives we consider and the shortest-path function.
We remark that zero-sum reachability (resp.~Büchi) games are the dual of zero-sum safety (resp.~co-Büchi) games, as the former can be transformed into the latter by exchanging the two players.

\section{Zero-sum games: punishing strategies}\label{section:zero-sum}
In this section, we provide an overview of relevant results regarding strategies in zero-sum games.
They are of interest to construct Nash equilibria with the classical punishment mechanism (previously described in Section~\ref{section:intro}).
Intuitively, this mechanism functions as follows: if some player deviates from the intended outcome of the NE, the other players coordinate as a coalition to ensure that no deviation is profitable by punishing a player who deviates.
The strategy of the coalition used to sabotage the deviating player is called a \textit{punishing strategy}, and punishing strategies can be obtained via the zero-sum coalition games.

To construct finite-memory NEs, we require punishing strategies that are finite-memory and that behave well independently of the vertex at which a deviation occurs.
In the following, we explain that we can always find memoryless punishing strategies.
In Section~\ref{section:zero-sum:reach}, we recall classical results on zero-sum reachability, safety, Büchi and co-Büchi games.
In Section~\ref{section:zero-sum:short}, we show that memoryless punishing strategies exist in zero-sum shortest-path games, even though the adversary in a zero-sum shortest-path game need not have an optimal strategy.

We fix a two-player arena $\arena = ((\vertexSetOne, \vertexSetTwo), E)$ and a target $\target\subseteq\vertexSet$ for the remainder of this section.

\subsection{Reachability, safety, Büchi and co-Büchi games}\label{section:zero-sum:reach}
Zero-sum reachability, safety, Büchi and co-Büchi games enjoy \textit{memoryless determinacy}: they are determined and for both players, there exist \textit{memoryless uniformly winning strategies}.
A proof of this property can be found in~\cite{DBLP:conf/dagstuhl/Mazala01} for reachability and safety games and in~\cite{DBLP:conf/focs/EmersonJ88} for Büchi and co-Büchi games (Büchi and co-Büchi objectives are special cases of parity objectives).
We summarise this result in the following theorem.
\begin{theorem}\label{thm:qualitative:ML strat}
  Both players have memoryless uniformly winning strategies in zero-sum reachability, safety, Büchi and co-Büchi games.
\end{theorem}

Let $\game = (\arena, \safe{\target})$ be a zero-sum safety game.
In $\game$, the winning region $\winningOne{\safe{\target}}$ of $\playerOne$ is a \textit{trap} for $\playerTwo$: for all $\vertex\in\winningOne{\safe{\target}}$, the set $\succSet{\vertex}\cap\winningOne{\safe{\target}}$ is non-empty if $\vertex\in\vertexSetOne$ and $\succSet{\vertex}\subseteq\winningOne{\safe{\target}}$ if $\vertex\in\vertexSetTwo$.
In particular, strategies of $\playerOne$ that only select successor vertices that are in $\winningOne{\safe{\target}}$ whenever possible are uniformly winning strategies of $\playerOne$ in $\game$ as they are guaranteed to avoid $\target\subseteq\winningTwo{\reach{\target}}$ from any vertex in $\winningOne{\safe{\target}}$.
We summarise this classical property in the following lemma for later.

\begin{lemma}\label{lemma:qualitative:safety}
  Let $\game = (\arena, \safe{\target})$ be a zero-sum safety game.
  The winning region $\winningOne{\safe{\target}}$ of $\playerOne$ in $\game$ is a trap for $\playerTwo$.
  Furthermore, for all strategies $\stratOne$ of $\playerOne$ such that, for all $\hist\in\historySetOne{\arena}$, $\last{\hist}\in\winningOne{\safe{\target}}$, implies that $\stratOne(\hist)\in\winningOne{\safe{\target}}$, it holds that $\stratOne$ is a uniformly winning strategy of $\playerOne$.
\end{lemma}

\subsection{Shortest-path games}\label{section:zero-sum:short}
We now discuss properties of zero-sum shortest-path games.
Let $\weight\colon\edgeSet\to\IN$ be a weight function and $\game = (\arena, \costReach{\target}{\weight})$ be a zero-sum shortest-path game (fixed for this whole section).
First, we note that $\game$ is determined.
This can be shown using the determinacy of games with open objectives~\cite{gale1953infinite}.
An objective $\objective$ is \textit{open} if it can be written as a union of sets of continuations of histories, i.e., a union of the form $\bigcup_{\hist\in\mathcal{H}}\{\concat{\hist}{\play}\mid\play\in\playSet{\arena} \text{ and } \last{\hist}=  \first{\play}\}$ for some set $\mathcal{H}\subseteq\historySet{\arena}$.
The following argument also shows that $\playerOne$ has optimal strategies from all vertices and $\playerTwo$ has optimal strategies from vertices with a finite value.

\begin{restatable}{lemma}{lemShortDeterminacy}\label{lem:short:determinacy}
  The shortest-path game $\game$ is determined.
  For all $\vertex_0\in\vertexSet$, $\playerOne$ has an optimal strategy from $\vertex_0$ and, if $\val(\vertex_0) < +\infty$, then $\playerTwo$ has an optimal strategy from $\vertex_0$.
\end{restatable}
\begin{proof}
  For $\alpha\in\IN$, let us consider the objective $\{\costReach{\target}{\weight}\leq\alpha\} = \{\play\in\playSet{\arena}\mid \costReach{\target}{\weight}(\play)\leq\alpha\}$.
  Let $\alpha\in\IN$.
  The objective $\{\costReach{\target}{\weight}\leq\alpha\}$ is open: it is the set of continuations of histories ending in $\target$ with weight at most $\alpha$.
  We obtain that $\game_\alpha=(\arena, \{\costReach{\target}{\weight}\leq\alpha\})$ is determined~\cite{gale1953infinite}.
  Therefore, in $\game_\alpha$, from any vertex $\vertex_0\in\vertexSet$, there is $\playerIndex\in\{1, 2\}$ such that $\playerI$ has a winning strategy $\stratI$ from $v_0$.
  If $\playerOne$ wins from $\vertex_0$ in $\game_\alpha$, then $\playerOne$ can ensure $\alpha$ in $\game$, whereas if $\playerTwo$ wins from $\vertex_0$ in $\game_\alpha$, then $\playerTwo$ can ensure $\alpha+1$ in $\game$.

  Let $\vertex_0\in\vertexSet$.
  We consider two cases.
  First, assume that for all $\alpha\in\IN$, $\playerTwo$ wins from $\vertex_0$ in $\game_\alpha$.
  It follows that $\val(\vertex_0)=+\infty$ and that all strategies of $\playerOne$ are optimal from $\vertex_0$.
  We now assume that $\playerOne$ wins from $\vertex_0$ in $\game_\alpha$ for some $\alpha\in\IN$.
  We have $\val(\vertex_0) = \min\{\alpha\in\IN\mid\playerOne\text{ wins from }\vertex_0\text{ in } \game_\alpha\}$, as both players have a strategy ensuring this number from $\vertex_0$ by the above (and the fact that all strategies of $\playerTwo$ ensure $0$ from any vertex in $\game$).
  In particular, both players have optimal strategies in $\game$ from $\vertex_0$.
\end{proof}

On the one hand, it can be shown that $\playerOne$ has a uniformly optimal memoryless strategy.
We defer the proof of this statement to Appendix~\ref{appendix:thm:short:pOneOpt}.
\begin{restatable}{theorem}{thmShortPOneOpt}\label{thm:short:pOneOpt}
  In the shortest-path $\game$, $\playerOne$ has a uniformly optimal memoryless strategy that is uniformly winning in the zero-sum reachability game $(\arena, \reach{\target})$.
\end{restatable}

On the other hand, we have illustrated that $\playerTwo$ does not necessarily have an optimal strategy from states with an infinite value in Example~\ref{ex:no opti}.
We further observe (via the same game) that, in general, there need not be memoryless strategies of $\playerTwo$ that are optimal from all vertices of finite value.
\begin{example}[Example~\ref{ex:no opti} continued]
  We consider the game of Example~\ref{ex:no opti} and build on what we have previously shown.
  We prove that $\player{2}$ does not have a memoryless strategy in this game that ensures $\alpha$ from $\vertex_\alpha$ for all $\alpha\in\IN\setminus\{0\}$.
    Consider the memoryless strategy $\stratTwo$ of $\playerTwo$ such that $\stratTwo(\vertex_\infty)= \vertex_\alpha$ for some $\alpha\in\IN\setminus\{0\}$.
    This strategy cannot ensure more than $\alpha+2$ from the vertex $\vertex_{\alpha+3}$; if $\playerOne$ moves from $\vertex_{\alpha+3}$ to $\vertex_\infty$, then moves leftwards from $\vertex_{\alpha}$, the cost of the resulting outcome is $\alpha+2<\val(\vertex_{\alpha+3})$.
    Therefore, there is no memoryless strategy of $\playerTwo$ in this game that ensures, from all finite-value vertices, their value.\hfill $\lhd$
\end{example}

The previous example highlights that uniformly optimal strategies need not exist in zero-sum shortest-path games for the second player.
However, to implement the punishment mechanism, we do not need optimal strategies: it suffices to use strategies that sufficiently punish deviating players no matter where a deviation occurs.
Formally, we show that there exists a family $(\stratTwo^\alpha)_{\alpha\in\IN}$ of memoryless strategies of $\playerTwo$ in $\game$ such that, for all $\alpha\in\IN$, $\stratTwo^\alpha$ is winning from any vertex in the winning region of $\playerTwo$ in the reachability game $(\arena, \reach{\target})$ and ensures the minimum of $\alpha$ and the value of the vertex from any other vertex.
Intuitively, the parameter $\alpha$ quantifies by how much $\playerOne$ should be sabotaged (uniformly).

Let $\alpha\in\IN$.
We define $\stratTwo^\alpha$ as follows.
On $\winningTwo{\safe{\target}}$, we let $\stratTwo^\alpha$ coincide with a uniformly winning memoryless strategy of $\playerTwo$ in $(\arena, \reach{\target})$.
For all $\vertex\in\vertexSetTwo\setminus\winningTwo{\safe{\target}}$, we $\stratTwo^\alpha(\vertex)$ be a successor $\vertex'$ of $\vertex$ maximising $\weight(\vertex,\vertex')+\val(\vertex')$ if there exists one, and,  otherwise let $\stratTwo^\alpha(\vertex)$ be a successor $\vertex'$ of $\vertex$ such that $\weight(\vertex,\vertex')+\val(\vertex')\geq\alpha$.
This ensures that from any vertex with an infinite value outside of $\winningTwo{\safe{\target}}$, $\playerTwo$ ensures a cost of at least $\alpha$.
In particular, from any vertex $\vertex$ with $\val(\vertex)<\alpha$, $\playerOne$ cannot exploit the decisions of $\playerTwo$ in infinite-value vertices to obtain a cost less than $\val(\vertex)$ from $\vertex$ (as illustrated in Example~\ref{ex:no opti}).
We formally establish the required properties of $\stratTwo^\alpha$ in the following proof.

\begin{restatable}{theorem}{thmShortMLStrat}\label{thm:short:ML strat}
  For all $\alpha\in\IN$, there exists a memoryless strategy $\stratTwo^\alpha$ of $\playerTwo$ such that, for all $\vertex\in\vertexSet$:
  \begin{enumerate}[(i)]
  \item $\stratTwo^\alpha$ is winning from $\vertex$ for $\playerTwo$ in the game $(\arena, \reach{\target})$ if $\vertex\in\winningTwo{\safe{\target}}$ and\label{item:short:ML strat:one}
  \item $\stratTwo^\alpha$ ensures a cost of at least $\min\{\val(v), \alpha\}$.\label{item:short:ML strat:two}
  \end{enumerate}
\end{restatable}

\begin{proof}
  Let $\stratTwo^{\safe{\target}}$ be a memoryless uniformly winning strategy of $\playerTwo$ in $(\arena, \reach{\target})$ (cf.~Theorem~\ref{thm:qualitative:ML strat}).
  For $\vertex\in\vertexSetTwo$, we let $\stratTwo^\alpha(\vertex)=\stratTwo^{\safe{\target}}(\vertex)$ if $\vertex\in\winningTwo{\safe{\target}}$, otherwise, if $\max_{\vertex'\in\succSet{\vertex}}(\weight(\vertex, \vertex')+\val(\vertex'))$ is defined, we let $\stratTwo^\alpha(\vertex)$ be a vertex attaining this maximum, and, otherwise, we let $\stratTwo^\alpha(\vertex)=\vertex'$ where $\vertex'\in\succSet{\vertex}$ is such that $\val(\vertex') + \weight(\vertex, \vertex')\geq\alpha$.

  We prove that $\stratTwo^\alpha$ satisfies the claimed properties.
  To show~\ref{item:short:ML strat:one}, it suffices to prove that any play starting in $\winningTwo{\safe{\target}}$ consistent with $\stratTwo^{\safe{\target}}$ never leaves $\winningTwo{\safe{\target}}$.
  We can prove this by contradiction.
  Assume that there there exists $\vertex\in\winningTwo{\safe{\target}}$ such that $\stratTwo^{\safe{\target}}(\vertex)\notin\winningTwo{\safe{\target}}$.
  By determinacy of safety games, $\playerOne$ can force a visit to $\target$ from $\stratTwo^{\safe{\target}}(\vertex)$.
  This implies that $\stratTwo^{\safe{\target}}$ is not winning from $\vertex$, contradicting the fact that $\stratTwo^{\safe{\target}}$ is uniformly winning.
  
  To establish~\ref{item:short:ML strat:two}, we show the following property: for any history $\history=\vertex_0\ldots\vertex_\indexLast$ that is consistent with $\stratTwo^\alpha$ such that for all $\indexPosition<\indexLast$, $\vertex_\indexPosition\notin\target$, it holds that $\weight(\history) + \min\{\val(\vertex_\indexLast),\alpha\} \geq \min\{\val(\vertex_0),\alpha\}.$

  We proceed by induction on the length of histories.
  For all $\vertex_0\in\vertexSet$, the property is immediate for the history $\history=\vertex_0$.
  We now consider a suitable history $\history =\vertex_0\ldots\vertex_\indexLast$ with $\indexLast\geq 1$ and assume the property holds for $\history'=\vertex_0\ldots\vertex_{\indexLast-1}$ by induction (note that $\history'$ is of the suitable form as well).
  Let $\edge=(\vertex_{\indexLast-1}, \vertex_{\indexLast})$.
  
  We discuss two cases depending on whether $\val(\vertex_{\indexLast-1})$ is finite and split each case depending on whom controls $\vertex_{\indexLast-1}$.
  We first assume that $\val(\vertex_{\indexLast-1})$ is finite.
  By Lemma~\ref{lem:short:determinacy}, both players have optimal strategies from any vertex with a finite value.

  We observe that
  \begin{align*}
    \weight(\history) + \min\{\val(\vertex_\indexLast),\alpha\} & = \weight(\history') + \weight(\edge) + \min\{\val(\vertex_\indexLast),\alpha\} \\
                                                                & \geq \weight(\history') + \min\{\val(\vertex_\indexLast)+\weight(\edge),\alpha\}.
  \end{align*}
  To conclude by induction, it suffices to show that $\val(\vertex_\indexLast)+\weight(\edge)\geq\val(\vertex_{\indexLast-1})$.
  If $\vertex_{\indexLast-1}\in\vertexSetOne$, $\playerOne$ can ensure a cost of $\val(\vertex_\indexLast)+\weight(\edge)$ from $\vertex_{\indexLast-1}$ by moving from $\vertex_{\indexLast-1}$ to $\vertex_\indexLast$ then playing optimally from there, yielding the desired inequality.
  Assume now that $\vertex_{\indexLast-1}\in\vertexSetTwo$. 
  For all $\vertex'\in\succSet{\vertex_{\indexLast-1}}$, it holds that if $\playerTwo$ can ensure $\beta\in\IN$ from $\vertex'$, then $\playerTwo$ can ensure $\beta + \weight(\vertex_{\indexLast-1}, \vertex')$ from $\vertex_{\indexLast-1}$, i.e., $\val(\vertex_{\indexLast-1}) \geq \beta + \weight(\vertex_{\indexLast-1}, \vertex')$.
  Because $\val(\vertex_{\indexLast-1})$ is finite, it follows that the value of all successors of $\vertex_{\indexLast-1}$ also is.
  The definition of $\stratTwo^\alpha$ and $\vertex_\indexLast=\stratTwo^\alpha(\vertex_{\indexLast-1})$ imply that $\val(\vertex_{\indexLast-1}) = \val(\vertex_\indexLast) + \weight(\vertex_{\indexLast-1}, \vertex_\indexLast)$, ending the proof of this case.

  Let us now assume that $\val(\vertex_{\indexLast-1})=+\infty$.
  We first consider the case $\vertex_{\indexLast-1}\in\vertexSetOne$.
  Then all successors of $\vertex_{\indexLast-1}$ have an infinite value, otherwise $\playerOne$ could ensure some natural number from $\vertex_{\indexLast-1}$ by moving to a successor with finite value and playing optimally from there.
  We must therefore show that $\weight(\history) + \alpha\geq\min\{\val(\vertex_0), \alpha\}$, which follows directly from $\alpha\geq\min\{\val(\vertex_0), \alpha\}$.
  
  Next, we assume that $\vertex_{\indexLast-1}\in\vertexSetTwo$.
  It follows from $\val(\vertex_{\indexLast-1})=+\infty$ and the definition of $\stratTwo^\alpha$ that $\val(\vertex_\indexLast) + \weight(\edge)\geq \alpha$.
  The desired inequality follows from $\weight(\history) \geq \weight(\edge)$, ending the induction proof.

  Let $\vertex\in\vertexSet\setminus\winningTwo{\safe{\target}}$.
  We now use the previous property to conclude that $\stratTwo^\alpha$ ensures $\min\{\val(\vertex), \alpha\}$ from $\vertex$.
  Let $\play$ be a play consistent with $\stratTwo^\alpha$ starting in $\vertex$.
  If $\play$ does not visit $\target$, then $\costReach{\target}{\weight}(\play)=+\infty$.
  Otherwise, let $\history$ be the prefix of $\play$ up to the first vertex in $\target$ included.
  Then, we have $\costReach{\target}{\weight}(\play)=\weight(\history)\geq \min\{\val(\vertex), \alpha\}$ by the previous property.
  This shows that $\stratTwo^\alpha$ ensures $\min\{\val(\vertex), \alpha\}$ from $\vertex$, ending the proof.
\end{proof}

\begin{remark}[Optimal strategies for $\playerTwo$]\label{rmk:opti:ptwo:branching}
  The proof above suggests a class of arenas in which $\playerTwo$ has a memoryless uniformly optimal strategy for the $\costReach{\target}{\weight}$ cost function.
  We show the following: if $\succSet{\vertex}$ is a finite set for all $\vertex\in\vertexSetTwo$, then $\playerTwo$ has a memoryless uniformly optimal strategy.
  In particular, this property holds in finitely-branching arenas.
  
  Assume that $\succSet{\vertex}$ is a finite set for all $\vertex\in\vertexSetTwo$.
  In this case, the definition of $\stratTwo^\alpha$ is independent of $\alpha$.
  Let $\stratTwo = \stratTwo^0$ and let us show that $\stratTwo$ is optimal from all vertices.
  It follows from the proof above that $\stratTwo$ is optimal from all vertices with finite value and all vertices in $\winningTwo{\safe{\target}}$.
  To close our argument that $\stratTwo$ is uniformly optimal, we prove that all $\vertex\in\vertexSet$ must either satisfy $\val(\vertex)\in\IN$ or $\vertex\in\winningTwo{\safe{\target}}$.
  
  We assume towards a contradiction that there exists $\vertex\in\vertexSet$ such that $\val(\vertex)=+\infty$ and $\vertex\in\winningOne{\reach{\target}}$.
  Fix a strategy $\stratOne$ that is winning from $\vertex$ for $\playerOne$ in the reachability game $(\arena, \reach{\target})$.
  All plays starting in $\vertex_0$ that are consistent with $\stratOne$ eventually reach $\target$.
  The set of their prefixes up to the first occurrence of a vertex of $\target$ can be seen as a finitely branching tree (branching occurs only when $\playerTwo$ selects a move).
  Since $\vertex$ has an infinite value, there are histories in the tree with arbitrarily large weight.
  Therefore, the tree must be infinite.
  By König's lemma~\cite{konig1927schlussweise}, there must be an infinite branch in this tree, i.e., a play consistent with $\stratOne$ that does not visit $\target$.
  This contradicts the assumption that $\stratOne$ is winning from $\vertex$ for $\playerOne$.
  \hfill $\lhd$
\end{remark}

\section{Characterising Nash equilibria outcomes}\label{section:outcomes}
We provide characterisations of plays that are outcomes of NEs in the five classes of games we consider.
These characterisations relate to the corresponding zero-sum games.
Intuitively, in most cases, a play is an NE outcome if and only if values in the coalition games of the vertices in the play do not suggest the existence of a profitable deviation for any player.
We provide characterisations of NE outcomes for reachability, safety, Büchi and co-Büchi games in Section~\ref{section:outcomes:reach} and a characterisation for shortest-path games in Section~\ref{section:outcomes:short}.
We fix an arena $\arena = \arenaTuple$ and targets $\target_1, \ldots, \target_\nPlayer\subseteq\vertexSet$ for this entire section.

\subsection{Reachability, safety, Büchi and co-Büchi games}\label{section:outcomes:reach}
We provide characterisations for NE outcomes in games with objectives.
The characterisations provided in this section are analogous to the NE outcome characterisations for finite arenas of~\cite{DBLP:conf/icalp/ConduracheFGR16}.
These characterisations can be proven in the same way in an arbitrary arena as they can be in finite arenas; we include a proof below to illustrate the construction of NEs via a punishment mechanism.

First, we consider games with reachability, Büchi and co-Büchi objectives.
Let $\game = (\arena, (\objective_\playerIndex)_{\playerIndex\in\playerSet})$ be a game with $\objective_\indexPlayer\in\{\reach{\target_\playerIndex}, \buchi{\target_\playerIndex}, \cobuchi{\target_\playerIndex}\}$ for all $\indexPlayer\in\playerSet$.
We denote by $\winningI{\objective_\playerIndex}$ the winning region of the first player of the coalition game $\game_\playerIndex = (\arena_\playerIndex, \objective_\playerIndex)$, in which $\playerI$ is opposed to the other players.

Let $\play=\vertex_0\vertex_1\vertex_2\ldots\in\playSet{\arena}$ be a play.
The play $\play$ is an NE outcome if and only if the players whose objectives are not satisfied cannot enforce their objective from any vertex along $\play$, i.e., if for all $\indexPlayer\in\playerSet\setminus\satpl{\play}$, $\vertex_\indexPosition\notin\winningI{\objective_\playerIndex}$ for all $\indexPosition\in\IN$.
On the one hand, if some player whose objective is not satisfied can enforce it from some vertex $\vertex_\indexPosition$ along $\play$, then they have a profitable deviation by switching to a winning strategy of their coalition game from $\vertex_\indexPosition$; the resulting outcome is winning for this player as we consider reachability, Büchi and co-Büchi objectives.
Conversely, one constructs a Nash equilibrium as follows. 
The players follow the play $\play$, and, if $\playerI$ deviates from $\play$, then all other players conform to a (memoryless) uniformly winning strategy for the second player in the coalition $\game_i$ for the objective $\playSet{\arena}\setminus\objective_\playerIndex$.
This ensures no player has a profitable deviation.

We formally state the characterisation below.
The following statement differs slightly from its finite-arena counterpart in~\cite{DBLP:conf/icalp/ConduracheFGR16}: we consider games in which players need not have objectives of the same type, whereas~\cite{DBLP:conf/icalp/ConduracheFGR16} considers games in which either all players have reachability objectives or all players have prefix-independent objectives (which includes Büchi and co-Büchi objectives).
We provide a proof for the sake of completeness.

\begin{restatable}{theorem}{thmCharacReachBuchi}\label{thm:charac:reach:buchi}
  Let $\game= (\arena, (\objective_\indexPlayer)_{\indexPlayer\in\playerSet})$ be a game such that $\objective_\indexPlayer\in\{\reach{\target_\playerIndex}, \buchi{\target_\playerIndex}, \cobuchi{\target_\playerIndex}\}$ for all $\indexPlayer\in\playerSet$.
  Let $\play=\vertex_0\vertex_1\ldots$ be a play.
  Then $\play$ is the outcome of an NE from $\vertex_0$ if and only if, for all $\playerIndex\in\playerSet\setminus\satpl{\play}$, $\vertex_\indexPosition\notin\winningI{\objective_\playerIndex}$ for all $\indexPosition\in\IN$.
\end{restatable}
\begin{proof}
  First, assume that there exist $\playerIndex\in\playerSet\setminus\satpl{\play}$ and $\indexPosition\in\IN$ such that $\vertex_\indexPosition\in\winningI{\objective_\indexPlayer}$.
  Let $\stratProfile=(\strat{\playerIndex'})_{\playerIndex'\in\playerSet}$ be a strategy profile such that $\outcome{\stratProfile}{\vertex_0} = \play$.
  We show that $\stratProfile$ is not an NE from $\vertex_0$ by highlighting a profitable deviation of $\playerI$.
  We let $\stratAltI$ be a strategy of $\playerI$ that agrees with $\stratI$ until $\vertex_\indexPosition$ is visited and that agrees with a memoryless uniformly winning strategy of $\playerI$ over all histories containing $\vertex_\indexPosition$.
  The outcome $\outcome{(\stratAltI, \stratIAdv)}{\vertex_0}$ can be written as a combination of a prefix of $\play$ up to the first occurrence of $\vertex_\indexPosition$ and an outcome of a uniformly winning strategy of $\playerI$ from $\winningI{\objective_\indexPlayer}$.
  Thus, the play $\outcome{(\stratAltI, \stratIAdv)}{\vertex_0}$ has a suffix satisfying $\objective_\playerIndex$, which implies that $\outcome{(\stratAltI, \stratIAdv)}{\vertex_0}$ itself satisfies $\objective_\playerIndex$, as we consider reachability, Büchi and co-Büchi objectives.
  We obtain that $\stratAltI$ is a profitable deviation of $\playerI$, and thus that $\stratProfile$ is not an NE from $\vertex_0$.
  We have proven that $\play$ is not an NE outcome.

  We now prove the other implication.
  For all $\playerIndex\in\playerSet$, we fix a memoryless uniformly winning strategy $\stratAltAdvI$ of $\playerIAdv$ in the zero-sum coalition game $\game_\playerIndex = (\arena_\indexPlayer, \objective_\indexPlayer)$.
  Assume that for all $\playerIndex\in\playerSet\setminus\satpl{\play}$, $\vertex_\indexPosition\notin\winningI{\objective_\playerIndex}$ for all $\indexPosition\in\IN$.
  We construct an NE $\stratProfile$ from $\vertex_0$ such that $\outcome{\stratProfile}{\vertex_0} = \play$ as follows.
  Let $\playerIndex\in\playerSet$.
  We provide a partial definition of a strategy $\stratI$ of $\playerI$ over histories starting in $\vertex_0$.
  Let $\hist\in\historySetI{\arena}$ such that $\first{\hist} = \vertex_0$.
  If $\hist = \prefix{\play}{\indexPosition}$ for some $\indexPosition\in\IN$, we let $\stratI(\hist) = \vertex_{\indexPosition+1}$.
  We now assume that $\hist$ is not a prefix of $\play$.
  We can write $\hist = \concat{\hist'}{\hist''}$ where $\hist'$ is the longest common prefix of $\hist$ and $\play$.
  Let $\indexPlayer'\in\playerSet$ such that $\last{\hist'}\in\vertexSet_{\playerIndex'}$.
  If $\indexPlayer'=\playerIndex$, we leave $\stratI(\hist)$ arbitrary, and otherwise, we let $\stratI(\hist) = \stratAltAdvI(\last{\hist})$.

  We let $\stratProfile = (\stratI)_{\playerIndex\in\playerSet}$.
  By construction, we obtain that $\play = \outcome{\stratProfile}{\vertex_0}$.
  To conclude, we show that $\stratProfile$ is an NE from $\vertex_0$.
  Players in $\satpl{\play}$ cannot have profitable deviations.
  It therefore remains to show that the other players do not have profitable deviations.
  We let $\playerIndex\in\playerSet\setminus\satpl{\play}$.
  It suffices to show that for all $\play'\in\playSet{\arena}$ starting in $\vertex_0$, if $\play'$ is consistent with $\stratIAdv$, then $\play'\notin\objective_\indexPlayer$.
  Let $\play'$ be a play starting in $\vertex_0$ that is consistent with $\stratIAdv$.
  It follows from the definition of $\stratProfile$ that $\play'$ can be written as the combination of a prefix $\prefix{\play}{\indexPosition}$ of $\play$ and a play $\play''$ consistent with $\stratAltAdvI$ starting in the winning region $\winningIAdv{\playSet{\arena}\setminus\objective_\indexPlayer}$ of $\playerIAdv$ in $\game_\indexPlayer$.
  In particular, we have $\play''\notin\objective_\indexPlayer$.
  If $\objective_\indexPlayer$ is a reachability objective, the target of $\playerI$ does not occur in $\prefix{\play}{\indexPosition}$ (as $\play\notin\objective_\indexPlayer$) and does not occur in $\play''$ either (as $\play''\notin\objective_\indexPlayer$).
  If $\objective_\indexPlayer$ is a Büchi or co-Büchi objective, we have $\play'\in\objective_\indexPlayer$ if and only if $\play''\in\objective_\indexPlayer$.
  In both cases, we obtain that $\play'\notin\objective_\indexPlayer$.
  This shows that $\playerI$ does not have a profitable deviation, and thus that $\stratProfile$ is an NE from $\vertex_0$.
\end{proof}

For safety objectives, we have to slightly adapt the characterisation of Theorem~\ref{thm:charac:reach:buchi}.
If a safety objective is violated by a play, this violation is witnessed by a prefix of the play.
It follows that there can be no profitable deviations for some $\playerI$ that occur at a point of an outcome after a vertex unsafe for $\playerI$ is traversed.
In other words, in a safety game $\game = (\arena, (\safe{\target_\indexPlayer})_{\indexPlayer})$, a play $\play$ is an NE outcome if and only if, for all $\playerIndex\in\playerSet\setminus\satpl{\play}$, all vertices of $\play$ that occur before $\target_\indexPlayer$ is visited are outside of $\winningI{\safe{\target_\indexPlayer}}$.
We omit a proof of this result: it can be shown in the same way as the corresponding finite-arena characterisation of~\cite{DBLP:conf/icalp/ConduracheFGR16} (or with an argument analogous to that of Theorem~\ref{thm:charac:reach:buchi}).

\begin{restatable}{theorem}{thmCharacSafe}\label{thm:charac:safe}
  Let $\game= (\arena, (\safe{\target_\indexPlayer})_{\indexPlayer\in\playerSet})$ be a safety game.
  Let $\play=\vertex_0\vertex_1\ldots$ be a play.
  Then $\play$ is the outcome of an NE from $\vertex_0$ if and only if, for all $\playerIndex\in\playerSet\setminus\satpl{\play}$, $\vertex_\indexPosition\notin\winningI{\objective_\playerIndex}$ for all $\indexPosition\leq\min\{\indexLast\in\IN\mid\vertex_\indexLast\in\target_\indexPlayer\}$.
\end{restatable}

\subsection{Shortest-path games}\label{section:outcomes:short}
We now characterise NE outcomes in shortest-path games.
We provide a characterisation for shortest-path games in which each player has their own weight function.
This setting is more general that the class of shortest-path games considered in Section~\ref{section:reach}; our memory bounds for NEs apply to shortest-path games in which costs of the players are built on the same weight function.

For all $\indexPlayer\in\playerSet$, let $\weight_\indexPlayer\colon\edgeSet\to\IN$ be a weight function for $\playerI$.
We consider the shortest-path game $\game = (\arena, (\costReach{\target_\playerIndex}{\weight_\indexPlayer})_{\playerIndex\in\playerSet})$ on $\arena$.
For any $\vertex\in\vertexSet$, we denote by $\valI(\vertex)$ the value of $\vertex$ in the coalition game $\game_\playerIndex = (\arena_\playerIndex, \costReach{\target_\playerIndex}{\weight_\indexPlayer})$.
We also reuse the notation $\winningI{\reach{\target_\playerIndex}}$ of the previous section.

In reachability, safety, Büchi and co-Büchi games, Theorems~\ref{thm:charac:reach:buchi} and~\ref{thm:charac:safe} indicate that is to sufficient to refer to the value in coalition games (i.e., who wins) to characterise NE outcomes.
For shortest-path games on finite arenas, we can also characterise NE outcomes by only using coalition game values (see~\cite[Thm.~15]{DBLP:journals/jcss/BrihayeBGT21}).
However, this is no longer the case in infinite arenas.
\begin{example}\label{ex:stronger charac}
    Let us consider the arena depicted in Figure~\ref{fig:ex:non-opt} and let $\target_1=\{t\}$ and $\target_2=\{\vertex_0\}$.
    It holds that $\val_1(\vertex_0)=+\infty$ (it follows from $\val_1(\vertex_\infty)=+\infty$ which is shown in Example~\ref{ex:no opti}).
    Therefore, the cost of all suffixes of the play $\vertex_0^\omega$ for $\playerOne$ matches the value of their first vertex $\vertex_0$.
    However, for any strategy profile resulting in $\vertex_0^\omega$ from $\vertex_0$, $\playerOne$ has a profitable deviation in moving to $\vertex_\infty$ and using a reachability strategy to ensure a finite cost.
    \hfill$\lhd$
\end{example}

A value-based characterisation fails because of vertices $\vertex\in\winningI{\reach{\target_\playerIndex}}$ such that $\valI(\vertex)$ is infinite.
Despite the infinite value of such vertices, $\playerI$ has a strategy such that their cost is finite no matter the behaviour of the others.
To obtain a characterisation of NE outcomes, we must impose additional conditions on players whose targets are not visited in the outcome: these players must not be able to force a visit to their target from any vertex of the outcome.

We show that a play in a shortest-path game is the outcome of an NE if and only if it is an outcome of an NE for the reachability game $(\arena, (\reach{\target_\playerIndex})_{\playerIndex\in\playerSet})$ and that, for players whose targets are visited, the values $\valI$ suggest they do not have a profitable deviation before their target is reached (after this point, their cost is fixed, similarly to safety games).
We follow a proof approach similar to that of Theorem~\ref{thm:charac:reach:buchi}.

\begin{restatable}{theorem}{thmCharacShort}\label{thm:charac:short}
  Let $\play = \vertex_0\vertex_1\ldots$ be a play.
  Then $\play$ is an outcome of an NE from $\vertex_0$ in $\game$ if and only
  \begin{enumerate}[(i)]
  \item for all $\playerIndex\in\playerSet\setminus\satpl{\play}$ and $\indexPosition\in\IN$, we have $\vertex_\indexPosition\notin\winningI{\reach{\target_\playerIndex}}$ and\label{item:charac:short:one}
  \item for all $\playerIndex\in\satpl{\play}$, we have $\costReach{\target_\playerIndex}{\weight_\indexPlayer}(\suffix{\play}{\indexPosition})\leq \valI(\vertex_\indexPosition)$ for all $\indexPosition\leq \min\{\indexLast\in\IN\mid \vertex_\indexLast\in \target_\playerIndex\}$.\label{item:charac:short:two}
  \end{enumerate}
\end{restatable}

\begin{proof}
  We start by proving that if~\ref{item:charac:short:one} or~\ref{item:charac:short:two} does not hold, then $\play$ cannot be the outcome of an NE.
  Let $\stratProfile = (\stratI)_{\playerIndex\in\playerSet}$ such that $\play=\outcome{\stratProfile}{\vertex_0}$.
  We show that some player has a profitable deviation.
  First, assume that~\ref{item:charac:short:one} does not hold.
  Let $\playerIndex\notin\satpl{\play}$ and $\indexPosition\in\IN$ such that $\vertex_\indexPosition\in\winningI{\reach{\target_\playerIndex}}$.
  Any strategy with which $\prefix{\play}{\indexPosition}$ is consistent that switches to a uniform memoryless winning strategy in the zero-sum reachability game $(\arena_\playerIndex, \reach{\target_\playerIndex})$ once $\vertex_\indexPosition$ is reached is a profitable deviation: it guarantees a finite cost for $\playerI$.
  Therefore $\stratProfile$ is not an NE from $\vertex_0$.

  We now assume that \ref{item:charac:short:two} does not hold.
  Let $\playerIndex\in\satpl{\play}$, $\indexLast_\playerIndex=\min\{\indexLast\in\IN\mid \vertex_\indexLast\in \target_\playerIndex\}$ and $\indexPosition\leq \indexLast_\playerIndex$ such that $\costReach{\target_\playerIndex}{\weight_\indexPlayer}(\suffix{\play}{\indexPosition})>\val(\vertex_\indexPosition)$.
  Any strategy with which $\prefix{\play}{\indexPosition}$ is consistent and that commits to a strategy of the first player of $\game_i$ that is optimal from $\vertex_\indexPosition$ once $\vertex_\indexPosition$ is reached ensures a cost of at most
  \[\weight_\indexPlayer(\prefix{\play}{\indexPosition}) + \val(\vertex_\indexPosition) < \weight_\indexPlayer(\prefix{\play}{\indexPosition}) + \costReach{\target_\playerIndex}{\weight_\indexPlayer}(\suffix{\play}{\indexPosition}) = \costReach{\target_\playerIndex}{\weight_\indexPlayer}(\play),\]
  thus it is a profitable deviation of $\playerI$, i.e., $\stratProfile$ is not an NE from $\vertex_0$.

  We now show the converse implication.
  Let $\stratProfile=(\stratI)_{\playerIndex\in\playerSet}$ be a strategy profile such that all players follow $\play$, and if $\playerI$ deviates from $\play$, the coalition $\playerIAdv$ switches to a memoryless strategy $\stratAltAdvI$ that is uniformly winning for the second player of the coalition reachability game $(\arena_\playerIndex, \reach{\target_\playerIndex})$ and that ensures at cost of at most $\min\{\valI(\vertex), \costReach{\target_\playerIndex}{\weight_\indexPlayer}(\play)\}$ from all $\vertex\in\vertexSet$ (such a strategy exists by Theorem~\ref{thm:short:ML strat} in $\game_\playerIndex$).
  
  We now show that for all $\playerIndex\in\playerSet$, $\playerI$ does not have a profitable deviation.
  This is equivalent to showing that for all plays $\play'$ starting in $\vertex_0$, if $\play'$ is consistent with $\stratIAdv$, then $\costReach{\target_\playerIndex}{\weight}(\play')\geq\costReach{\target_\playerIndex}{\weight}(\play)$.
  We let $\play' = \vertex_0'\vertex_1'\ldots\in\playSet{\arena}$ such that $\vertex_0'=\vertex_0$ and $\play'$ is consistent with $\stratIAdv$.
  We assume that $\play'\neq\play$, as the inequality is direct otherwise, and let $\indexLast=\min\{\indexPosition\in\IN\mid\vertex_\indexPosition\neq\indexPosition\}$.
  By construction of $\stratProfile$ and choice of $\indexLast$, it follows that $\suffix{\play'}{\indexLast-1}$ is consistent with the punishing strategy $\stratAltAdvI$.
  
  First, we assume that $\playerIndex\notin\satpl{\play}$.
  In this case, we must show that $\play'\notin\reach{\target_\playerIndex}$.
  On the one hand, no elements of $\target_\playerIndex$ occur in $\suffix{\play'}{\indexLast-1}$ by choice of $\stratAltAdvI$, as~\ref{item:charac:short:one} yields $\vertex'_{\indexLast-1}=\vertex_{\indexLast-1}\notin\winningI{\reach{\target_\playerIndex}}$.
  On the other hand, no target vertices of $\playerI$ occur in $\prefix{\play'}{\indexLast-1}=\prefix{\play}{\indexLast-1}$ as $\play\notin\reach{\target_\playerIndex}$.
  This ends the proof of the first case.

  We now assume that $\playerIndex\in\satpl{\play}$.
  If a target vertex of $\playerI$ occurs in $\prefix{\play}{\indexLast-1}$, then $\costReach{\target_\playerIndex}{\weight_\playerIndex}(\play)=\costReach{\target_\playerIndex}{\weight_\playerIndex}(\play')$ (which is sufficient to end the proof).
  We therefore assume that no vertex in $\prefix{\play}{\indexLast-1}$ is in $\target_\playerIndex$ and further distinguish two cases.
  We first assume that $\valI(\vertex_{\indexLast-1})\geq \costReach{\target_\playerIndex}{\weight_\indexPlayer}(\play)$.
  We have
  \[\costReach{\target_\playerIndex}{\weight_\playerIndex}(\play') =
    \weight(\prefix{\play}{\indexLast-1}) +
    \costReach{\target_\playerIndex}{\weight_\indexPlayer}(\suffix{\play'}{\indexLast-1}) \geq
    \costReach{\target_\playerIndex}{\weight_\indexPlayer}(\play),\]
  as $\stratAltAdvI$ ensures a cost of at least $\costReach{\target_\playerIndex}{\weight_\indexPlayer}(\play)$ from $\vertex_{\indexLast-1}$.
  We now assume that $\valI(\vertex_{\indexLast-1}) < \costReach{\target_\playerIndex}{\weight_\indexPlayer}(\play)$.
  In this case, we obtain that
  \begin{align*}
    \costReach{\target_\playerIndex}{\weight_\playerIndex}(\play')
    & =
      \weight(\prefix{\play}{\indexLast-1}) +
      \costReach{\target_\playerIndex}{\weight_\indexPlayer}(\suffix{\play'}{\indexLast-1}) \\
    & \geq
      \weight(\prefix{\play}{\indexLast-1}) + \valI(\vertex_{\indexLast-1}) \\
    & \geq \weight(\prefix{\play}{\indexLast-1})
      + \costReach{\target_\playerIndex}{\weight_\indexPlayer}(\suffix{\play}{\indexLast-1}) \\
    & =
      \costReach{\target_\playerIndex}{\weight_\indexPlayer}(\play),
  \end{align*}
  where the penultimate inequality follows from~\ref{item:charac:short:two}.
\end{proof}

\begin{remark}
Informally,~\cite[Thm.~15]{DBLP:journals/jcss/BrihayeBGT21} states that a play is the outcome of an NE if and only if condition~\ref{item:charac:short:two} of Theorem~\ref{thm:charac:short} holds for all players (with the minimum over vertex indices in the condition replaced by an infinimum).
This characterisation only fails when there are vertices $\vertex\in\winningI{\reach{\target_\playerIndex}}$ with $\valI(\vertex)=+\infty$.
However, such vertices do not exist if $\succSet{\vertex}$ is a finite set for all $\vertex\in\vertexSet$ (refer to Remark~\ref{rmk:opti:ptwo:branching}).
Therefore, the finite-arena characterisation of~\cite[Thm.~15]{DBLP:journals/jcss/BrihayeBGT21} extends to finitely-branching arenas.
\hfill$\lhd$
\end{remark}

\section{Finite-memory Nash equilibria in reachability and shortest-path games}\label{section:reach}
In this section, we study reachability and shortest-path games (with a single non-negative integer weight function for all players).
The main results of this section state that given an NE in a reachability or shortest-path game, we can derive from it a finite-memory NE in which the same targets are visited, the costs are not increased for any player (in shortest-path games) and we provide an upper bound depending solely on the number of players on the amount of memory used by all of the involved strategies.
In other words, we obtain \textit{arena-independent} upper bounds on the sufficient memory for solutions to the constrained NE existence problem (see Section~\ref{section:intro}).
To obtain these arena-independent memory bounds, we do not fully implement the punishment mechanism as in the proofs of Theorems~\ref{thm:charac:reach:buchi} and~\ref{thm:charac:short}.
Instead, intuitively, following a deviation of $\playerI$, the coalition $\playerIAdv$ does not necessarily switch to a punishing strategy for $\playerI$.
Instead, they may attempt to keep following a suffix of the equilibrium's original outcome if the deviation does not appear to prevent it.

This section is structured as follows.
We illustrate the constructions for reachability and shortest-path games with examples in Section~\ref{section:fm:example}.
In Section~\ref{section:fm:simple outcomes}, we show that we can derive, from any NE outcome, an NE outcome that admits a finite simple decomposition.
We impose technical properties on these decompositions such that we can construct finite-memory NEs based on them.
We introduce (finite-memory) strategies based on simple play decompositions in Section~\ref{section:fm:general}.
Sections~\ref{section:fm:reach} and~\ref{section:fm:short} explain how to instantiate the Mealy machines of Section~\ref{section:fm:general} to construct finite-memory NEs in reachability and shortest-path games respectively.

We fix an arena $\arena=\arenaTuple$, target sets $\target_1$, \ldots, $\target_\nPlayer\subseteq\vertexSet$ and a weight function $\weight\colon\edgeSet\to\IN$ for the remainder of this section.

\subsection{Examples}\label{section:fm:example}
In this section, we illustrate the upcoming construction for finite-memory NEs for both settings of interest.
We construct NEs from (well-chosen) simple segment decompositions of plays.
The main role of memory in the following is to keep track of two pieces of information: a segment of the play in the considered decomposition (which describes the moves to be made), and the last player to have moved (to track who to punish if necessary).
We start with a reachability game.
\begin{example}\label{ex:reach ne}
  We consider the game on the arena depicted in Figure~\ref{fig:ex:reach ne:game} where the objective of $\playerI$ is $\reach{\{t_\playerIndex\}}$ for all $\playerIndex\in\integerInterval{4}$.
  This game has a similar structure to that of the game of Example~\ref{example:ne:roles of memory}.
  To illustrate the idea behind the upcoming construction, we present a finite-memory NE from $\vertex_0$ in which the targets of $\playerOne$ and $\playerTwo$ are visited in its outcome.
  In this example, we construct an NE from the NE outcome $\play=\vertex_0\vertex_1\vertex_2t_1\vertex_2\vertex_1\vertex_0t_2^\omega$.
  Intuitively, this is the simplest NE outcome starting in $\vertex_0$ that visits both $t_1$ and $t_2$.

  \begin{figure}
     \centering
     \begin{subfigure}[b]{0.55\textwidth}
       \centering
       \begin{tikzpicture}[node distance=0.4cm]
         \node[state, align=center] (v0) {$\vertex_0$};
         \node[state, square, align=center, left = of v0] (t2) {$t_2$};
         \node[state, diamond, align=center, right = of v0] (v1) {$\vertex_1$};
         \node[state, regular polygon, regular polygon sides=6, align=center, right = of v1] (v2) {$\vertex_2$};
         \node[state, align=center, right = of v2] (t1) {$t_1$};
         \node[state, regular polygon, regular polygon sides=6, align=center, below = of v1] (v3) {$\vertex_4$};
         \node[state, diamond, align=center, below = of v2] (v4) {$\vertex_3$};
         \node[inner sep=0cm, node distance=0.2cm, right = of v3] (ref) {};
         \node[state, align=center, node distance=6mm, below = of ref] (v5) {$\vertex_5$};
         \node[state, diamond, align=center, left = of v5] (t3) {$t_3$};
         \node[state, regular polygon, regular polygon sides=6, align=center, right = of v5] (t4) {$t_4$};
         \path[->] (v0) edge[bend left] (v1);
         \path[->] (v1) edge[bend left] (v0);
         \path[->] (v1) edge[bend left] (v2);
         \path[->] (v2) edge[bend left] (v1);
         \path[->] (v2) edge[bend left] (t1);
         \path[->] (t1) edge[bend left] (v2);
         \path[->] (v0) edge (t2);
         \path[->] (v1) edge (v3);
         \path[->] (v2) edge (v4);
         \path[->] (v3) edge (v5);
         \path[->] (v4) edge (v5);
         \path[->] (v5) edge (t3);
         \path[->] (v5) edge (t4);
         \path[->] (t2) edge[loop below] (t2);
         \path[->] (t3) edge[loop left] (t3);
         \path[->] (t4) edge[loop right] (t4);
       \end{tikzpicture}
       \caption{An arena. Circles, squares, diamonds and hexagons are respectively~$\playerOne$, $\playerTwo$, $\player{3}$, $\player{4}$ vertices.}
       \label{fig:ex:reach ne:game}
     \end{subfigure}
     \hfill
     \begin{subfigure}[b]{0.42\textwidth}
       \centering
       \begin{tikzpicture}[node distance=1.2cm]
         \node[initial left, state, rectangle, rounded corners] (1p3) {$(\player{3}, 1)$};
         \node[state, rectangle, rounded corners, below = of 1p3] (1p4) {$(\player{4}, 1)$};
         \node[state, rectangle, rounded corners, right = of 1p3] (2p3) {$(\player{3}, 2)$};
         \node[state, rectangle, rounded corners, below = of 2p3] (2p4) {$(\player{4}, 2)$};
\begin{scope}[on background layer]
           \node[rectangle, fit=(1p3)(1p4), fill=black!10, rounded corners] (m1) {};
           \node[rectangle, fit=(2p3)(2p4), fill=black!10, rounded corners] (m2) {};
         \end{scope}
         
         \path[->] (1p3) edge[bend right] node[left] {$\vertex_2$} (1p4);
         \path[->] (1p4) edge[bend right] node[right] {$\vertex_1$} (1p3);
         \path[->] (1p4) edge node[above] {$t_1$} (2p4);
         \path[->] (1p3) edge node[above] {$t_1$} (2p3);
         \path[->] (2p3) edge[bend right] node[left] {$\vertex_2$} (2p4);
         \path[->] (2p4) edge[bend right] node[right] {$\vertex_1$} (2p3);
       \end{tikzpicture}
       \caption{An illustration of the update scheme of a Mealy machine. Transitions that do not change the memory state are omitted.}
       \label{fig:ex:reach ne:mealy}
     \end{subfigure}
     \caption{A reachability game and a representation of a Mealy machine update scheme suitable for an NE from $\vertex_0$.}
     \label{fig:ex:reach ne}
   \end{figure}

      First, observe that $\play$ can be seen as the combination of the simple history $\segment_1=\vertex_0\vertex_1\vertex_2t_1$ and the simple lasso $\segment_2=t_1\vertex_2\vertex_1\vertex_0t_2^\omega$.
   The simple history $\segment_1$ connects the initial vertex to the first visited target, and the simple lasso $\segment_2$ connects the first target to the second and contains the suffix of the play.
   Therefore, if $\player{3}$ and $\player{4}$ were not looking for a profitable deviation, the outcome $\play$ could be obtained by using a finite-memory strategy profile where all strategies are defined by a Mealy machine with state space $\integerInterval{2}$.
   Intuitively, these strategies would follow $\segment_1$ while remaining in their first memory state $1$, then, when $t_1$ is visited, they would update their memory state to $2$ and follow $\segment_2$.

  We build on these simple Mealy machines with two states.
  We include additional information in each memory state.
  We depict a suitable Mealy machine state space and update scheme in Figure~\ref{fig:ex:reach ne:mealy}.
  The rectangles grouping together states $(\player{3}, \indexSegment)$ and $(\player{4}, \indexSegment)$ represent the memory state $\indexSegment$ of the simpler Mealy machine for $\indexSegment\in\integerInterval{2}$.
  The additional information roughly encodes the last player to act among the players whose objective is not satisfied in $\play$.
  More precisely, an update is performed from the memory state $(\playerI, \indexSegment)$ only if the vertex fed to the Mealy machine appears in $\segment_\indexSegment$ for $\indexSegment\in\integerInterval{2}$.
  
  By construction, if $\playerI$ (among $\player{3}$ and $\player{4}$) deviates and exits the set of vertices of $\segment_\indexSegment$ when in a memory state of the form $(\cdot, \indexSegment)$, then the memory updates to $(\playerI, \indexSegment)$ and does not change until the play returns to some vertex of $\segment_\indexSegment$ (which is not possible here due to the structure of the arena, but may be in general).
  For instance, assume that $\player{3}$ moves from $\vertex_1$ to $\vertex_4$ after the history $\history=\vertex_0\vertex_1\vertex_2t_1\vertex_2\vertex_1$.
  Then the memory after $\history$ is in state $(\player{3}, 2)$ and no longer changes from there on.

  It remains to explain how the next-move function of the Mealy machine should be defined to ensure an NE.
  Essentially, for a state of the form $(\playerI, \indexSegment)$ and vertices in $\segment_\indexSegment$, we assign actions as in the simpler two-state Mealy machine described previously.
  On the other hand, for a state of the form $(\playerI, \indexSegment)$ and a vertex not in $\segment_\indexSegment$, we use a memoryless punishing strategy against $\playerI$.
  In this particular case, we need only specify what $\playerOne$ should do in $\vertex_5$.
  Naturally, in memory state $(\playerI, \indexSegment)$, $\playerOne$ should move to the target of the other player.
  It is essential to halt memory updates for vertices $\vertex_3$ and $\vertex_4$ to ensure the correct player is punished.

  We close this example with comments on the structure of the Mealy machine.
  Assume that the memory state is of the form $(\playerI, \indexSegment)$.
  If a deviation occurs and leads to a vertex of $\segment_\indexSegment$ other than the intended one, then the other players will continue trying to progress along $\segment_\indexSegment$ and do not specifically try punishing the deviating player.
  Similarly, if after a deviation leaving the set of vertices of $\segment_\indexSegment$ (from which point the memory is no longer updated until this set is rejoined), a vertex of $\segment_\indexSegment$ is visited again, then the players resume trying to progress along this history and memory updates resume.
  In other words, these finite-memory strategies do not pay attention to all deviations and do not have dedicated memory that commit to punishing deviating players for the remainder of a play after a deviation.
  \hfill$\lhd$
\end{example}

We now give an example for shortest-path games.
In this case, we use slightly larger Mealy machines: it may be necessary to commit to a punishing strategy if the set of vertices of the segment the players want to progress along is left.
This requires additional memory states.
Furthermore, our example also highlights that it may be necessary to punish deviations from players whose targets are visited, as they can improve their cost (in contrast to reachability games).

\begin{example}\label{ex:short ne}
  We consider the shortest-path game on the weighted arena depicted in Figure~\ref{fig:ex:short ne:game} where the target of $\playerI$ is $\target_\playerIndex=\{t, t_{12}\}$ for all $\playerIndex\in\integerInterval{2}$ and $\target_3=\{t\}$ for $\player{3}$.
  Our goal is to construct an NE from $\vertex_0$ with the cost profile $(5, 5, 5)$; this can be done by designing an NE with outcome $\play=\vertex_0\vertex_1\vertex_3t^\omega$.
  We argue that the idea described in Example~\ref{ex:reach ne} would not work in this setting. We then provide an alternative construction that builds on the same ideas.

  \begin{figure}
    \centering
    \begin{subfigure}[b]{0.48\textwidth}
      \centering
      \begin{tikzpicture}[node distance=0.45cm]
        \node[state, align=center] (v0) {$\vertex_0$};
        \node[state, align=center, right = of v0] (v1) {$\vertex_1$};
        \node[state, square, align=center, below = of v1] (v2) {$\vertex_2$};
        \node[state, diamond, align=center, right = of v1] (v3) {$\vertex_3$};
        \node[state, square, align=center, right = of v3] (t123) {$t$};
        \node[state, align=center, below = of t123] (v4) {$\vertex_4$};
        \node[state, align=center, right = of v2] (t12) {$t_{12}$};
        \path[->] (v0) edge node[above] {$3$} (v1);
        \path[->] (v0) edge (v2);
        \path[->] (v1) edge (v3);
        \path[->] (v2) edge (v3);
        \path[->] (v2) edge (t12);
        \path[->] (v3) edge (t123);
        \path[->] (v3) edge (v4);
        \path[->] (v3) edge[bend right=45] (v0);
        \path[->] (t12) edge[loop below] (t12);
        \path[->] (t123) edge[loop right] (t123);
        \path[->] (v4) edge[loop right] (v4);
      \end{tikzpicture}
      \caption{A weighted arena. Circles, squares and diamonds are respectively~$\playerOne$, $\playerTwo$ and $\player{3}$ vertices. Edge labels indicate their weight. Unlabelled edges have a weight of $1$.}
      \label{fig:ex:short ne:game}
    \end{subfigure}
    \hfill
    \begin{subfigure}[b]{0.5\textwidth}
      \centering
      \begin{tikzpicture}[node distance=1.2cm]
        \node[initial left, state, rectangle, rounded corners] (p1) {$(\player{1}, 1)$};

        \node[state, rectangle, rounded corners, right = of p1] (p2) {$(\player{2}, 1)$};

        \node[inner sep=0pt, node distance=0.75cm, right = of p1] (ref) {};
        \node[state, rectangle, rounded corners, below = of ref] (p3) {$(\player{3}, 1)$};
        \node[state, rectangle, rounded corners, below = of p3] (p3b) {$\player{3}$};
        \node[state, rectangle, rounded corners, node distance=0.5cm, right = of p3b] (p2b) {$\player{2}$};
        \node[state, rectangle, rounded corners, node distance=0.5cm, left = of p3b] (p1b) {$\player{1}$};
\begin{scope}[on background layer]
          \node[rectangle, fit=(p1)(p2)(p3), fill=black!10, rounded corners] (m) {};
        \end{scope}

        \path[->] (p1) edge[] node[below] {$t$} (p2);
        \path[->] (p2) edge[bend right=20] node[above, align=center] {$\vertex_0$, $\vertex_1$} (p1);
        \path[->] (p3) edge[bend left=20] node[below left, align=center] {$\vertex_0$\\ $\vertex_1$} (p1);
        \path[->] (p1) edge[] node[above right] {$\vertex_3$} (p3);
        \path[->] (p2) edge[] node[above left] {$\vertex_3$} (p3);
        \path[->] (p3) edge[bend right=20] node[below right] {$t$} (p2);
        \path[->] (p1) edge[bend right] node[left, align=center] {$\vertex_2$ \\ $t_{12}$ \\ $\vertex_4$} (p1b);     
        \path[->] (p2) edge[bend left] node[right, align=center] {$\vertex_2$ \\ $t_{12}$ \\ $\vertex_4$} (p2b);   
        \path[->] (p3) edge node[left, align=center] {$\vertex_2$ \\ $t_{12}$ \\ $\vertex_4$} (p3b);
      \end{tikzpicture}
      \caption{An illustration of the update scheme of a Mealy machine. Transitions that do not change the memory state are omitted.}
      \label{fig:ex:short ne:mealy}
    \end{subfigure}
    \caption{A shortest-path game and a representation of a Mealy machine update scheme suitable for some NE from $\vertex_0$.}
    \label{fig:ex:short ne}
  \end{figure}

     In this case, $\play$ is a simple lasso, much like the second part of the play in the previous example.
    First, let us assume a Mealy machine similar to that of Example~\ref{ex:reach ne}, i.e., such that it tries to progress along $\play$ whenever it is in one of the vertices of $\play$.
    The update scheme of such a Mealy machine would be obtained by removing the transitions to states of the form $\playerI$ from Figure~\ref{fig:ex:short ne:mealy} (replacing them by self-loops).
  
  If $\player{3}$ uses a strategy based on such a Mealy machine, then $\playerOne$ has a profitable deviation from $\vertex_0$.
  Indeed, if $\playerOne$ moves from $\vertex_0$ to $\vertex_2$, then either $\playerOne$ incurs a cost of $2$ if $\playerTwo$ moves to $t_{12}$ from $\vertex_2$ or a cost of $3$ if $\playerTwo$ moves to $\vertex_3$ as $\player{3}$ would then move to $t$ by definition of their Mealy machine.
  To circumvent this issue, we add new memory states $\playerI$ for all players and if $\playerI$ exits the set of vertices of $\play$, we update the memory to the punishment state $\playerI$.
  This results in the update scheme depicted in Figure~\ref{fig:ex:short ne:mealy}.
  Next-move functions to obtain an NE can be defined as follows, in addition to the expected behaviour to obtain $\play$: 
  for $\playerTwo$, $\mealyNext_2((\playerOne, 1), \vertex_2) = \mealyNext_2(\playerOne, \vertex_2) = \vertex_3$ and for $\player{3}$, $\mealyNext_3(\playerOne, \vertex_3)=\vertex_4$.

  Similarly to the previous example, players do not explicitly react to deviations that move to vertices of $\play$; if $\player{3}$ deviates after reaching $\vertex_3$ and moves back to $\vertex_0$, the memory of the other players does not update to state $\player{3}$.
  Intuitively, there is no need to switch to a punishing strategy for $\player{3}$ as going back to the start of the intended outcome is more costly than conforming to it, preventing the existence of a profitable deviation.

  We remark that this example differs slightly from the general construction below.
  According to the general construction, we should decompose $\play$ into two parts: a history $\vertex_0\vertex_1\vertex_3t$ from the initial vertex to the first target and the suffix $t^\omega$ of the play after all targets are visited.
    Furthermore, we can argue that such a split is sometimes necessary (see Example~\ref{ex:bonus segment}).
    \hfill$\lhd$
\end{example}

\subsection{Well-structured Nash equilibrium outcomes}\label{section:fm:simple outcomes}

The examples of finite-memory NEs presented on the previous section were built from well-chosen decompositions of an NE outcome.
In this section, we show that NE outcomes that admit a decomposition from which we can construct an NE can always be found in reachability and shortest-path games.
More precisely, we show that we can always derive, from a given NE outcome, another NE outcome that admits a finite simple decomposition such that each history segment of this decomposition connects the different targets visited along the play and has minimal weight.
In particular, the derived NE outcome has a preferable cost profile to that of the original outcome.

\begin{figure}
  \centering
  \scalebox{1}{
    \begin{tikzpicture}[node distance=4mm, every dotnode/.style={minimum size=6mm},every state/.style={minimum size=6mm}]
      \node[state, fill=blue!20] (v0) {};
      \node[state, right =of v0, fill=blue!20] (v1) {};
      \node[dotnode,right = of v1] (dots1) {\ldots};
      \node[state, right = of dots1, diagonal fill={red!20}{blue!20}, accepting] (v2) {};
      \node[state, right =of v2, fill=red!20] (v3) {};
      \node[dotnode,right = of v3] (dots2) {\ldots};
      \node[state, right = of dots2, diagonal fill={yellow!20}{red!20}, accepting] (v4) {};
      \node[dotnode,right = of v4] (dots3) {\ldots};
      \node[state, right = of dots3, fill=custom-green!20] (v5) {};
      \node[state, right = of v5, diagonal fill={orange!20}{custom-green!20},  accepting] (v6) {};
      \node[state, right = of v6, fill=orange!20] (v7) {};
      \node[state, right = of v7,fill=orange!20] (v8) {};
      \node[dotnode,right = of v8] (dots4) {\ldots};

      \path[->] (v0) edge (v1);
      \path[->] (v1) edge (dots1);
      \path[->] (dots1) edge (v2);
      \path[->] (v2) edge (v3);
      \path[->] (v3) edge (dots2);
      \path[->] (dots2) edge (v4);
      \path[->] (v4) edge (dots3);
      \path[->] (dots3) edge (v5);
      \path[->] (v5) edge (v6);
      \path[->] (v6) edge (v7);
      \path[->] (v7) edge (v8);
      \path[->] (v8) edge (dots4);

      \node[emptynode,below = 2mm of v0.south west] (brace1l) {};
      \node[emptynode,below = 2mm of v2.south east] (brace1r) {};
      \node[emptynode,below = 2mm of v2.south west] (brace2l) {};
      \node[emptynode,below = 2mm of v4.south east] (brace2r) {};
      \node[emptynode,below = 2mm of dots3.south east] (brace3l) {};
      \node[emptynode,below = 2mm of v6.south east] (brace3r) {};
      \node[emptynode,below = 2mm of v6.south west] (brace4l) {};
      \node[emptynode,below = 2mm of dots4.south east] (brace4r) {};

      \draw[decorate,blue,decoration={brace,amplitude=5pt,mirror}]
      (brace1l) -- (brace1r) node [midway,blue,below,yshift=-5pt] {$\segment_1$};
      \draw[decorate,red,decoration={brace,amplitude=5pt,mirror}]
      (brace2l) -- (brace2r) node [midway,red,below,yshift=-5pt] {$\segment_2$};
      \draw[decorate,custom-green,decoration={brace,amplitude=5pt,mirror}]
      (brace3l) -- (brace3r) node [midway,custom-green,below,yshift=-5pt] {$\segment_{\numSegments-1}$};
      \draw[decorate,orange,decoration={brace,amplitude=5pt,mirror}]
      (brace4l) -- (brace4r) node [midway,orange,below,yshift=-5pt] {$\segment_\numSegments$};

      \node[emptynode, below = 3mm of v4] (simpl1) {};
      \node[emptynode, below = 25mm of simpl1] (simpl2) {};
      \path[->,draw=black,thick,decorate,decoration={snake,post length=3pt}] (simpl1) -- (simpl2) node[midway,right,align=center,xshift=5pt] {Simplification for\\ each segment};

      \node[state, below = 3mm of simpl2, diagonal fill={yellow!20}{red!20}, accepting] (w4) {};
      
      \node[dotnode, left = of w4] (dotsw2) {\ldots};
      \node[state, left =of dotsw2, fill=red!20] (w3) {};

      \node[state, left = of w3, diagonal fill={red!20}{blue!20}, accepting] (w2) {};
      \node[dotnode,left = of w2] (dotsw1) {\ldots};
      \node[state, left =of dotsw1, fill=blue!20] (w1) {};
      \node[state, left = of w1, fill=blue!20] (w0) {};

      \node[dotnode,right = of w4] (dotsw3) {\ldots};
      \node[state, right = of dotsw3, fill=custom-green!20] (w5) {};
      \node[state, right = of w5, diagonal fill={orange!20}{custom-green!20},  accepting] (w6) {};
      \node[state, right = of w6, fill=orange!20] (w7) {};
      \node[state, right = of w7,fill=orange!20] (w8) {};
      \node[dotnode,right = of w8] (dotsw4) {\ldots};

      \path[->] (w0) edge (w1);
      \path[->] (w1) edge (dotsw1);
      \path[->] (dotsw1) edge (w2);
      
      \path[->] (w2) edge (w3);
      \path[->] (w3) edge (dotsw2);
      \path[->] (dotsw2) edge (w4);
      \path[->] (w4) edge (dotsw3);
      \path[->] (dotsw3) edge (w5);
      \path[->] (w5) edge (w6);
      \path[->] (w6) edge (w7);
      \path[->] (w7) edge (w8);
      \path[->] (w8) edge (dotsw4);

      \node[emptynode,below = 2mm of w0.south west] (bracew1l) {};
      \node[emptynode,below = 2mm of w2.south east] (bracew1r) {};
      \node[emptynode,below = 2mm of w2.south west] (bracew2l) {};
      \node[emptynode,below = 2mm of w4.south east] (bracew2r) {};
      \node[emptynode,below = 2mm of dotsw3.south east] (bracew3l) {};
      \node[emptynode,below = 2mm of w6.south east] (bracew3r) {};
      \node[emptynode,below = 2mm of w6.south west] (bracew4l) {};
      \node[emptynode,below = 2mm of dotsw4.south east] (bracew4r) {};

      \draw[decorate,blue,decoration={brace,amplitude=5pt,mirror}]
      (bracew1l) -- (bracew1r) node [midway,blue,below,yshift=-5pt,align=center] {$\segment_1'$\\simple\\history};
      
      \draw[decorate,red,decoration={brace,amplitude=5pt,mirror}]
      (bracew2l) -- (bracew2r) node [midway,red,below,yshift=-5pt,align=center] {$\segment_2'$\\simple\\ history};
      \draw[decorate,custom-green,decoration={brace,amplitude=5pt,mirror}]
      (bracew3l) -- (bracew3r) node [midway,custom-green,below,yshift=-5pt,align=center] {$\segment_{\numSegments-1}'$\\simple\\ history};
      \draw[decorate,orange,decoration={brace,amplitude=5pt,mirror}]
      (bracew4l) -- (bracew4r) node [midway,orange,below,yshift=-5pt,align=center] {$\segment_\numSegments'$\\simple lasso or \\simple play};

      \path[->,red] (w1) edge[bend left] node {$\mathbin{\tikz [x=1.4ex,y=1.4ex,line width=.2ex] \draw[-] (0,0) -- (1,1) (0,1) -- (1,0);}$} node[above,align=center] {No shortcuts\\via $\segment_1$ vertices} (w2);
    \end{tikzpicture}
  }
  \caption{Simplification process for an NE outcome in a multi-player shortest-path game.
    Doubly circled vertices denote the first occurrence of a target vertex of some player in the play.
  }\label{figure:simplification:reach}
\end{figure}

We focus on \textit{shortest-path games} in this section.
For reachability games, the existence of suitable NE outcomes can be recovered by noting that a shortest-path game with a constant zero weight function is equivalent to its unweighted reachability game.
This derivation can be seen as a simplification process, illustrated in Figure~\ref{figure:simplification:reach}: we decompose a given NE outcome into segments up connecting the first visit to target sets, and then simplify each segment individually while minimising the weight of the simple history segments.
This simplification process yields outcomes admitting a decomposition satisfying the properties outlined in the following lemma.

\begin{restatable}{lemma}{lemSimpleDecOutcomes}\label{lem:spath:simple decomposition}
  Let $\play'$ be the outcome of an NE from $\vertex_0\in\vertexSet$ in a shortest-path game $\game=(\arena, (\costReach{\target_\playerIndex}{\weight})_{\playerIndex\in\playerSet})$.
  Let $\numSegments = |\vispos{\play'}\setminus\{0\}|$.
  There exists an NE outcome $\play$ from $\vertex_0$ with $\vispos{\play}\setminus\{0\} = \{\indexPosition_1 < \ldots < \indexPosition_\numSegments\}$ that admits a finite simple segment decomposition $(\segment_1, \ldots, \segment_{\numSegments+1})$ such that
  \begin{enumerate}[(i)]
  \item if $\numSegments>0$, then $(\segment_1, \ldots, \concat{\segment_\numSegments}{\segment_{\numSegments+1}})$ is also a simple decomposition of $\play$;\label{item:spath:simple decomposition:one}
  \item if $\numSegments>0$, then for all $\indexSegment\in\integerInterval{\numSegments}$, $\concat{\concat{\segment_1}{\ldots}}{\segment_\indexSegment}=\prefix{\play}{\indexPosition_\indexSegment}$;\label{item:spath:simple decomposition:two}
  \item if $\numSegments>0$, then for all $\indexSegment\in\integerInterval{\numSegments}$, $\weight(\segment_\indexSegment)$ is minimal among all histories that share their first and last vertex with $\segment_\indexSegment$ and traverse a subset of the vertices occurring in $\segment_\indexSegment$; and\label{item:spath:simple decomposition:three}
  \item for all $\playerIndex\in\playerSet$, $\costReach{\target_\playerIndex}{\weight}(\play)\leq \costReach{\target_\playerIndex}{\weight}(\play')$.\label{item:spath:simple decomposition:four}
  \end{enumerate}
\end{restatable}
\begin{proof}
  We distinguish two cases in this proof.
  First, let us assume that $\numSegments=0$, i.e., $\vispos{\play'}\subseteq\{0\}$.
  If $\play'$ is a simple play, it suffices to choose $\play = \play'$; it admits the trivial simple decomposition $(\play)$.
  Therefore, we assume that $\play'$ is not a simple play. 
  This implies that there is a simple lasso $\play\in\playSet{\arena}$ starting from $\vertex_0$ that only uses vertices that occur in $\play'$.
  It follows that $\vispos{\play}\subseteq\{0\}$.
  By Theorem~\ref{thm:charac:short}, $\play$ is an NE outcome; the first condition of the characterisation of Theorem~\ref{thm:charac:short} follows from it holding for $\play'$ and the second condition holds because $\vispos{\play}=\vispos{\play'}\subseteq\{0\}$.

  Second, we assume that $\numSegments> 0$.
  We define $\play$ by describing the simple decomposition $\decomp=(\segment_1, \ldots, \segment_{\numSegments+1})$.
  Let $\indexPosition_1' <\ldots < \indexPosition_\numSegments'$ be the elements of $\vispos{\play'}\setminus\{0\}$ and $\indexPosition_0' = 0$.
  For $\indexSegment\in\integerInterval{\numSegments}$, we let $\segment'_\indexSegment$ denote the segment of $\play$ between positions $\indexPosition_{\indexSegment-1}'$ and $\indexPosition_\indexSegment'$.
  We let $\segment_\indexSegment$ be a simple history that shares its first and last vertex with $\segment_\indexSegment'$ and traverses a subset of the vertices occurring in $\segment_\indexSegment'$, with minimal weight among all such histories.
  This definition ensures that, for all $\indexSegment\in\integerInterval{\numSegments}$, the minimum in condition~\ref{item:spath:simple decomposition:three} is realised by $\segment_\indexSegment$, and thus condition~\ref{item:spath:simple decomposition:three} holds.
  It remains to define the segment $\segment_{\numSegments+1}$. 
  We let $\segment_{\numSegments+1}$ be $\suffix{\play'}{\indexPosition_\numSegments}$ if $\concat{\segment_\numSegments}{\suffix{\play'}{\indexPosition_\numSegments}}$ is a simple play, and otherwise we let $\segment_{\numSegments+1}$ be any play starting in $\last{\segment_\numSegments}$ such that $\concat{\segment_{\numSegments}}{\segment_{\numSegments+1}}$ is a simple lasso in which only vertices of $\play'$ occur.
  It follows from this choice of $\segment_{\numSegments+1}$ that $\play=\concat{\concat{\segment_1}{\ldots}}{\segment_{\numSegments+1}}$ and $\decomp=(\segment_1, \ldots, \segment_{\numSegments+1})$ satisfy~\ref{item:spath:simple decomposition:one}.

  We now prove that $\play$ is an NE outcome satisfying~\ref{item:spath:simple decomposition:two} and~\ref{item:spath:simple decomposition:four}.
  Let $\play=\vertex_0\vertex_1\ldots$, and let $\indexPosition_1 <\ldots < \indexPosition_\numSegments$ be the elements of $\vispos{\play}\setminus\{0\}$ and $\indexPosition_0=0$.
  To show that $\play$ is an NE outcome, we rely on Theorem~\ref{thm:charac:short}.
  Because $\play'$ is an NE outcome and all vertices occurring in $\play$ occur in $\play'$, it follows that the first condition of the characterisation of Theorem~\ref{thm:charac:short} holds for $\play$.
  
  For the second condition of this characterisation, we fix $\playerIndex\in\satpl{\play}$ and $\indexSegment_\playerIndex\leq \numSegments$ such that $\indexPosition_{\indexSegment_\playerIndex} = \min\{\indexPosition\in\IN\mid\vertex_\indexPosition\in\target_\playerIndex\}$.
  We show that for all $\indexPosition\leq \indexPosition_{\indexSegment_\playerIndex}$, we have $\costReach{\target_\playerIndex}{\weight}(\suffix{\play}{\indexPosition})\leq\valI(\vertex_\indexPosition)$ where $\valI(\vertex_\indexPosition)$ is the value of $\vertex_\indexPosition$ in the coalition game $\game_\playerIndex = (\arena_\playerIndex, \costReach{\target_\playerIndex}{\weight})$.
  
  Let $\indexPosition\leq \indexPosition_{\indexSegment_\playerIndex}$ and $\indexSegment\leq \indexSegment_\playerIndex$ such that $\indexPosition_{\indexSegment} \leq\indexPosition< \indexPosition_{\indexSegment+1}$.
  By construction, there is an occurrence of $\vertex_\indexPosition$ in the segment $\segment'_\indexSegment$ of $\play'$.
  We consider a suffix $\suffix{\play'}{\indexPosition'}$ of $\play'$ starting from an occurrence of $\vertex_\indexPosition$ in $\segment'_\indexSegment$.
  The desired inequality follows from the relations
  \[\costReach{\target_\playerIndex}{\weight}(\suffix{\play}{\indexPosition})\leq\costReach{\target_\playerIndex}{\weight}(\suffix{\play'}{\indexPosition'})\leq \valI(\vertex_\indexPosition).\]
  
  We prove that the first inequality holds by contradiction.
  Assume that $\costReach{\target_\playerIndex}{\weight}(\suffix{\play}{\indexPosition})>\costReach{\target_\playerIndex}{\weight}(\suffix{\play'}{\indexPosition'})$.
  It must be the case that either a suffix of $\segment'_\indexSegment$ starting in $\vertex_\indexPosition$ has weight strictly less than the suffix $\vertex_\indexPosition\ldots\vertex_{\indexPosition_\indexSegment}$ of $\segment_\indexSegment$, or that $\weight(\segment_{\indexSegment'}) > \weight(\segment_{\indexSegment'}')$ for some $\indexSegment < \indexSegment'\leq \indexSegment_\playerIndex$. Both possibilities contradict the choice of the elements of $\decomp$, therefore we have $\costReach{\target_\playerIndex}{\weight}(\suffix{\play}{\indexPosition})\leq\costReach{\target_\playerIndex}{\weight}(\suffix{\play'}{\indexPosition'})$.
  The second inequality holds by Theorem~\ref{thm:charac:short} as $\play'$ is an NE outcome.
  We remark (for condition~\ref{item:spath:simple decomposition:four}) that in the special case $\indexPosition=0$, the first inequality implies that $\costReach{\target_\playerIndex}{\weight}(\play)\leq \costReach{\target_\playerIndex}{\weight}(\play')$ as we can choose $\indexPosition'=0$.
  We have shown that $\play$ is an NE outcome.
  
  We conclude by proving that~\ref{item:spath:simple decomposition:two} and \ref{item:spath:simple decomposition:four} hold.
  Condition~\ref{item:spath:simple decomposition:two} follows immediately by construction. 
  Let $\indexSegment\in\integerInterval{\numSegments}$.
  For condition~\ref{item:spath:simple decomposition:four}, due to the above, we need only consider players who do not see their target.
  For these players, the condition follows from the equality $\satpl{\play}=\satpl{\play'}$ implying that players have an infinite cost in $\play$ if and only if they have an infinite cost in $\play'$.
  This concludes the proof. 
\end{proof}

Condition~\ref{item:spath:simple decomposition:one} of the statement of Lemma~\ref{lem:spath:simple decomposition} suggests that we could merge the last two segments of the decomposition provided by the lemma to obtain simple decompositions with one less element.
In reachability games, we can build NEs from these smaller decompositions, and thus obtain smaller Mealy machines than by using the larger decomposition.
However, this is not true in shortest-path games.
In our upcoming NE construction (see the examples of Section~\ref{section:fm:example}), players do not react to deviations that remain within the current segment of the intended outcome.
In the following example, we illustrate that this behaviour can be exploited by players in shortest-path games if we merge these last two segments.

\begin{example}\label{ex:bonus segment}
  We consider the shortest-path game on the weighted arena depicted in Figure~\ref{fig:ex:bonus segment} where the target of $\playerOne$ is $\target_1=\{t_1, t_{12}\}$ and the target of $\playerTwo$ is $\target_2=\{t_{12}\}$.
  The play $\play =\vertex_0 t_{12}(\vertex_1t_1)^\omega$ is a simple lasso that is an NE outcome by Theorem~\ref{thm:charac:short}.
  We claim that there are no NEs in which $\playerTwo$ moves from $\vertex_1$ to $\target_1$ so long as the set of vertices of $\play$ is not left, i.e., in which $\playerTwo$ plays a strategy that is based on the trivial decomposition $(\play)$ of $\play$.

  Let $\stratProfile = (\stratOne, \stratTwo)$ be a strategy profile such that $\outcome{\stratProfile}{\vertex_0} = \play$ and $\stratTwo(\hist) = t_1$ for all histories $\hist$ such that $\last{\hist}=\vertex_1$ and $\vertex_2$ does not occur in $\hist$.
  The play $\vertex_0(\vertex_1t_1)^\omega$ is consistent with $\stratTwo$.
  This implies that $\playerOne$ has a profitable deviation, and thus $\stratProfile$ is not an NE from $\vertex_0$.
  We have shown that the segments $\vertex_0t_{12}$ and $t_{12}(\vertex_1t_1)^\omega$ should be distinguished to implement our relaxation of the punishment mechanism.
\hfill$\lhd$
\end{example}

\subsection{Strategies and play decompositions}\label{section:fm:general}

As suggested by the examples of Section~\ref{section:fm:example}, our construction of finite-memory NE in reachability and shortest-path games is similar: we build an NE from an NE outcome that admits a well-structured simple decomposition (obtained in practice via Lemma~\ref{lem:spath:simple decomposition}).
In this section, we provide partially-defined Mealy machines that we extend in the sequel to construct NEs in reachability and shortest-path games, and later in safety games.

We first introduce the general idea of a strategy based on a finite simple decomposition of a play in Section~\ref{section:fm:general:dec}.
We then define Mealy machines that induce strategies based on finite simple decompositions in Section~\ref{section:fm:general:strat}.

\subsubsection{Strategies based on a play decomposition}\label{section:fm:general:dec}

In Examples~\ref{ex:reach ne} and~\ref{ex:short ne}, we defined Nash equilibria such that some deviations from the intended outcome would not result in the deviating players being punished: if a deviation does not exit the segment of the play that the players are currently progressing along, the players would disregard the deviation and would keep following this segment.
For instance, in Example~\ref{ex:short ne}, if $\player{3}$ deviates, resulting in the history $\vertex_0\vertex_1\vertex_3\vertex_0$, the other players do not try to prevent $\player{3}$ from reaching a target (despite it being possible from $\vertex_0$).
Instead, they attempt to follow the moves suggested by the segment $\vertex_0\vertex_1\vertex_3 t$, i.e., they maintain their initial behaviour.
We say that strategies that exhibit such a behaviour are \textit{based on a simple decomposition}.
In the following, we formalise this idea.
We fix a play $\play$ that admits a finite simple decomposition $\decomp=(\segment_1, \ldots, \segment_\numSegments)$ for the remainder of the section.
We assume that there are \textit{no trivial segments} in $\decomp$, i.e., $\decomp$ does not contain segments of the form $\vertex$ for $\vertex\in\vertexSet$.

First, we define the histories for which players continue following the current segment of the play.
We call such histories coherent with $\decomp$.
\begin{definition}[History coherent with a decomposition]
  A history is \textit{coherent with $\decomp$} if there is some $\indexSegment\in\integerInterval{\numSegments}$ such that it is $\indexSegment$-coherent with $\decomp$.
  We define $\indexSegment$-coherence inductively.
  The only one-vertex history that is $1$-coherent with $\decomp$ is the history $\hist=\vertex_0$.
  For all $\history\in\historySet{\arena}$ and $\vertex\in\succSet{\last{\history}}$, we define that
  \begin{itemize}
  \item if $\hist$ is $\indexSegment$-coherent for some $\indexSegment <\numSegments$ and $\vertex=\last{\segment_\indexSegment}$, then $\hist\vertex$ is $(\indexSegment+1)$-coherent;
  \item otherwise, if $\hist$ is $\indexSegment$-coherent for some $\indexSegment\in\integerInterval{\numSegments}$ and $\vertex$ occurs in $\segment_\indexSegment$, then $\hist\vertex$ is $\indexSegment$-coherent;
  \item in any other case, $\hist\vertex$ is not coherent with $\decomp$.
  \end{itemize}
\end{definition}
As the cases in the above definition are disjoint, given a history $\hist$ coherent with $\decomp$, there is a unique $\indexSegment\in\integerInterval{\numSegments}$ such that $\hist$ is $\indexSegment$-coherent with $\decomp$.

We now define strategies based on $\decomp$, i.e., strategies that, given a coherent history, attempt to complete the segment in progress.
We use the following notion in this definition.
\begin{definition}
  Let $\history$ be a history that is $\indexSegment$-coherent with $\decomp$.
  The \textit{$\decomp$-next vertex} of $\hist$ to be the vertex that follows $\last{\history}$ in $\segment_\indexSegment$.
\end{definition}

Given a $\indexSegment$-coherent history $\hist$, its $\decomp$-next vertex is well-defined only if $\last{\hist}$ differs from $\last{\segment_\indexSegment}$. (note that uniqueness follows from $\decomp$ being simple).
We show that this is always the case below.

\begin{restatable}{lemma}{lemmaDecompositionNextVertex}\label{lemma:decomposition:next:vertex}
  Let $\history$ be a history that is coherent with the simple decomposition $\decomp=(\segment_1, \ldots, \segment_\numSegments)$.
  The $\decomp$-next vertex of $\history$ is well-defined.
\end{restatable}
\begin{proof}
  We assume that $\history$ is $\indexSegment$-coherent.
  We establish existence and uniqueness of this vertex.
  Uniqueness follows from the simplicity of the decomposition.
  Existence is clear if $\indexSegment=\numSegments$: $\segment_\numSegments$ does not have a final vertex.
  
  We therefore assume that $\indexSegment<\numSegments$ and show the existence of a $\decomp$-next vertex by contradiction.
  Assume there is no suitable vertex. 
  It implies that $\last{\history}=\last{\segment_\indexSegment}$.
  By simplicity and the absence of trivial histories in $\decomp$, we have $\first{\segment_\indexSegment}\neq\last{\segment_\indexSegment}$.
  Therefore, there must be a prefix of $\history$ that is $\indexSegment$-coherent by definition of $\indexSegment$-coherence.
  We obtain that $\history$ should either be $(\indexSegment+1)$-coherent or not coherent, a contradiction.
\end{proof}

We can now formalise the notion of a strategy based on $\decomp$.
\begin{definition}
  A strategy of $\playerI$ is \textit{based on $\decomp$} if to any history $\history\in\historySetI{\arena}$ that is coherent with $\decomp$, it assigns the $\decomp$-next vertex of $\history$.
\end{definition}

We now formally show below that if all players use strategies that are based on $\decomp$ from $\first{\play}$, then the resulting outcome is $\play$.
\begin{lemma}
  Let $\stratProfile = (\stratI)_{\playerIndex\in\playerSet}$ be a strategy profile such that $\stratI$ is based on $\decomp$ for all $\playerIndex\in\playerSet$.
  Then $\outcome{\stratProfile}{\first{\play}}=\play$.
\end{lemma}
\begin{proof}
  We consider two cases.
  First, assume that $\decomp = (\play)$.
  It follows that all prefixes of $\play$ are $1$-coherent with $\decomp$, and the claim follows from the definition of $\decomp$-next vertex.

  Assume now that $\numSegments> 1$.
  We let $\play =\vertex_0\vertex_1\ldots$ and $\indexPosition_1<\ldots<\indexPosition_{\numSegments-1}$ be indices such that, for all $\indexSegment\in\integerInterval{\numSegments-1}$, $\prefix{\play}{\indexPosition_\indexSegment} = \concat{\segment_1}{\concat{\ldots}{\segment_\indexSegment}}$.
  To end the proof, due to the definition of a $\decomp$-next vertex, it suffices to show that for all $\indexPosition\in\IN$:
  \begin{itemize}
  \item if $\indexPosition <\indexPosition_1$, then $\prefix{\play}{\indexPosition}$ is $1$-coherent with $\decomp$,
  \item if $\indexPosition_{\indexSegment-1}\leq\indexPosition<\indexPosition_{\indexSegment}$ for some $1< \indexSegment<\numSegments$, then $\prefix{\play}{\indexPosition}$ is $\indexSegment$-coherent with $\decomp$,
  \item if $\indexPosition \geq \indexPosition_{\numSegments-1}$, then $\prefix{\play}{\indexPosition}$ is $\numSegments$-coherent with $\decomp$.
  \end{itemize}
  We prove the above points by induction.
  For $\indexPosition=0$, the $1$-coherence of $\prefix{\play}{0}=\vertex_0$ is due to the base case of the definition of coherence.
  Now assume by induction that the claim holds for $\indexPosition\in\IN$ and let us show it for $\indexPosition+1$.

  First, let us assume that $\indexPosition+1 < \indexPosition_1$.
  The $1$-coherence of $\prefix{\play}{\indexPosition+1}$ follows from the $1$-coherence of $\prefix{\play}{\indexPosition}$ and the fact that $\vertex_{\indexPosition+1}$ occurs in $\segment_1$ and that $\vertex_{\indexPosition+1}\neq\last{\segment_1}= \vertex_{\indexPosition_1}$ by simplicity of $\segment_1$.

  We now assume that there exists some $1<\indexSegment<\numSegments$ such that $\indexPosition_{\indexSegment-1}\leq\indexPosition+1<\indexPosition_{\indexSegment}$.
  We distinguish two cases.
  If $\indexPosition +1 = \indexPosition_{\indexSegment-1}$, then by induction, $\prefix{\play}{\indexPosition}$ is $(\indexSegment-1)$-coherent and the $\indexSegment$-coherence of $\prefix{\play}{\indexPosition}$ follows from $\vertex_{\indexPosition+1}=\last{\segment_{\indexSegment-1}}$.
  If $\indexPosition +1 > \indexPosition_{\indexSegment-1}$, then by induction, $\prefix{\play}{\indexPosition}$ is $\indexSegment$-coherent; the $\indexSegment$-coherence of $\prefix{\play}{\indexPosition}$ is implied by $\vertex_{\indexPosition+1}$ occurring in $\segment_\indexSegment$ and $\vertex_{\indexPosition+1}\neq\last{\segment_{\indexSegment}}=\vertex_{\indexPosition_\indexSegment}$ (which follows from $\decomp$ being simple).

  The case $\indexPosition\geq\indexPosition_{\numSegments-1}$ can be handled in the same way as the previous one, and is omitted to limit repetition.
\end{proof}

\subsubsection{Finite-memory decomposition-based strategies}\label{section:fm:general:strat}
We generalise the common components of the Mealy machines used in Examples~\ref{ex:reach ne} and~\ref{ex:short ne}.
For the remainder of the section, we let $\play= \vertex_0\vertex_1\ldots\in\playSet{\arena}$ be a play that admits a finite simple segment decomposition $\decomp=(\segment_1, \ldots, \segment_\numSegments)$ that has no trivial segments.

We consider Mealy machines that, along coherent histories, keep track of two pieces of information: a player and a segment index.
This is achieved using memory states of the form $(\playerI, \indexSegment)$.
We do not require all such pairs, e.g., it is not necessary in Example~\ref{ex:reach ne}, and instead we parameterise our definition by a non-empty set of players $\playerSubset\subseteq\playerSet$.
Our construction is based on Mealy machines whose memory state space includes $\mealyStateID = \{\playerI\mid \playerIndex\in\playerSubset\}\times\integerInterval{\numSegments}$.
There may be additional memory states, as is the case in Example~\ref{ex:short ne}.
We require that the initial state of the Mealy machine be of the form $(\playerI, 1)$ for some $\playerIndex\in\playerSubset$.

Intuitively, the memory of these Mealy machines keeps track of the index of the current segment of the play (incrementing it when the last vertex of the current segment is reached) and the last player in $\playerSubset$ to have moved.
We define the next-move function so that it proposes the $\decomp$-next vertex after a history coherent with $\decomp$.
If some vertex outside of the current segment is processed (which only happens when dealing with histories that are not coherent with $\decomp$), we do not impose any constraint on the behaviour of the update and next-move functions.
We generalise this aspect differently depending on whether we consider reachability games or shortest-path games in the sequel.

We formalise the information explanation above as follows.
\begin{definition}[$(\playerSubset, \decomp)$-compatible Mealy machines]\label{def:compatible mealy}
  Let $\playerSubset\subseteq\playerSet$.
  We define $\mealyStateID = \{\playerI\mid \playerIndex\in\playerSubset\}\times\integerInterval{\numSegments}$.
  We let $\mealyUpdateID\colon\mealyStateID\times\vertexSet\to\mealyStateSpace$ denote the (partial) function such that for all $(\playerI, \indexSegment)\in\mealyStateID$ and all vertices $\vertex$ occurring in $\segment_\indexSegment$, $\mealyUpdateID((\playerI, \indexSegment), \vertex)=(\playerIAlt, \indexSegment')$ such that
  \begin{enumerate}[(a)]
  \item $\playerIndex'$ is such that $\vertex\in\vertexSet_{\playerIndex'}$ if $\vertex\in\bigcup_{\playerIndex''\in\playerSubset}\vertexSet_{\playerIndex''}$ and otherwise $\playerIndex'= \playerIndex$,
  \item  $\indexSegment'=\indexSegment+1$ if $\indexSegment<\numSegments$ and $\vertex = \last{\segment_\indexSegment}$ and $\indexSegment'=\indexSegment$ otherwise.
  \end{enumerate}
  For a given $\playerIndex\in\playerSet$, we let $\mealyNextIDpI\colon\mealyStateID\times\vertexSet\to\vertexSet$ denote the (partial) function such that for all $(\playerIAlt, \indexSegment)\in\mealyStateID$ and vertices $\vertex\in\vertexSetI$ that occur in $\segment_\indexSegment$, $\mealyNextIDpI((\playerIAlt, \indexSegment), \vertex)=\vertex'$ such that
  \begin{enumerate}[(a)]
  \item if $\indexSegment<\numSegments$ and $\vertex=\last{\segment_\indexSegment}$, then $\vertex'$ is the vertex occurring after $\vertex$ in $\segment_{\indexSegment+1}$,
  \item otherwise, we let $\vertex'$ be the vertex occurring after $\vertex$ in $\segment_\indexSegment$.
  \end{enumerate}
  
  A Mealy machine $\mealyMachine = \mealyTuplePure$ of $\playerI$ is \textit{$(\playerSubset, \decomp)$-compatible} if $\mealyStateID\subseteq\mealyStateSpace$, $\mealyStateInit$ is of the form $(\playerI, 1)$ for some $\playerIndex\in\playerSubset$, and $\mealyUpdate$ and $\mealyNextI$ coincide with $\mealyUpdateID$ and $\mealyNextIDpI$ respectively on the domain of the latter functions.
\end{definition}

Let $\playerSubset\subseteq\playerSet$ be a non-empty set.
We prove that any finite-memory strategy induced by an $(\playerSubset, \decomp)$-compatible Mealy machine is based on $\decomp$.
To this end, we establish that if a history $\history$ is $\indexSegment$-coherent with $\decomp$, then the memory state after the Mealy machine reads $\history$ is of the form $(\playerI, \indexSegment)$.
\begin{restatable}{lemma}{lemmaCoherenceTemplateExtensions}\label{lemma:coherence:template:extensions}
  Let $\indexPlayer\in\playerSet$.
  Let $\mealyMachine=\mealyTuplePure$ be a Mealy machine of $\playerI$ that is $(\playerSubset, \decomp)$-compatible.
  The strategy $\stratI$ induced by $\mealyMachine$ is based on $\decomp$ and for all $\history\in\historySet{\arena}$, if $\history$ is $\indexSegment$-coherent with $\decomp$, then $\widehat{\mealyUpdate}(\history) = (\playerIAlt, \indexSegment)$ for some $\playerIndex'\in\playerSubset$.
\end{restatable}

\begin{proof}
  We first show the second claim of the lemma.
  We proceed by induction on the number of vertices in $\history$.
  The only coherent history with a single vertex is $\vertex_0$.
  If $\segment_1=\play$ (i.e., $\decomp$ is a trivial decomposition), then $\mealyUpdate(\mealyStateInit, \vertex_0)$ is of the form $(\playerIAlt, 1)$ (the definition of $\mealyUpdateID$ in Definition~\ref{def:compatible mealy} has only one case).
  Assume now that $\decomp$ is a non-trivial decomposition.
  It follows that $\segment_1$ is a non-trivial simple history, and thus $\vertex_0\neq\last{\segment_1}$,
  By definition of $\mealyUpdateID$, $\mealyUpdate(\mealyStateInit, \vertex_0)$ is again of the form $(\playerIAlt, 1)$.

  We now consider a $\indexSegment$-coherent history $\history$ and assume by induction that $\widehat{\mealyUpdate}(\history) = (\playerIAlt, \indexSegment)$. 
  Let $\vertex\in\succSet{\last{\history}}$ such that $\history\vertex$ is coherent with $\decomp$.
  It follows that $\vertex$ occurs in $\segment_\indexSegment$.
  Therefore, $\widehat{\mealyUpdate}(\history\vertex) = \mealyUpdate((\playerIAlt, \indexSegment), \vertex) = \mealyUpdateID((\playerIAlt, \indexSegment), \vertex)$.
  We distinguish two cases.
  If $\indexSegment < \numSegments$ and $\vertex=\last{\segment_\indexSegment}$, then $\history\vertex$ is $(\indexSegment+1)$-coherent and by definition of $\mealyUpdateID$, $\mealyUpdateID((\playerIAlt, \indexSegment), \vertex)$ is of the form $(\playerIAltAlt, \indexSegment+1)$.
  Otherwise, $\history\vertex$ is $\indexSegment$-coherent and by definition of $\mealyUpdateID$, $\mealyUpdateID((\playerIAlt, \indexSegment), \vertex)$ is of the form $(\playerIAltAlt, \indexSegment)$.

  It remains to prove that $\stratI$ is based on $\decomp$.
  Let $\history\in\historySetI{\arena}$ be a coherent history.
  If $\history$ contains only one vertex, then by coherence $\history=\vertex_0$.
  The definition of $\mealyNextIDpI$ ensures that $\stratI(\vertex_0)$ is the next vertex of the history $\vertex_0$ with respect to $\decomp$.
  If $\history$ contains more than one vertex, let $\history=\history'\vertex$ and assume that $\history'$ is $\indexSegment$-coherent.
  By the previous point, it holds that $\widehat{\mealyUpdate}(\history') = (\playerIAlt, \indexSegment)$ for some $\playerIndex'\in\playerSubset$.
  Therefore, $\stratI(\history)=\mealyNextIDpI((\playerIAlt, \indexSegment), \vertex)$.
  It follows from the definition of $\mealyNextIDpI$ that $\stratI$ maps $\history$ to its next vertex with respect to $\decomp$.
\end{proof}

\subsection{Nash equilibria in reachability games}\label{section:fm:reach}

We show that from any NE in a reachability game, we can derive a finite-memory NE from the same initial vertex in which the same targets are visited.
We build on the NE outcomes with simple decompositions provided by Lemma~\ref{lem:spath:simple decomposition} and the Mealy machines of Definition~\ref{def:compatible mealy} to obtain NEs with a memory size of at most $\nPlayer^2$.
We fix a reachability game $\game=(\arena, (\reach{\target_\playerIndex})_{\playerIndex\in\playerSet})$.

Let $\playerSubset\subseteq \playerSet$ is the set of players who do not see their targets if it is non-empty, or a single arbitrary player if all players see their target.
The general idea of the construction is to use, for each $\playerIndex\in\playerSet$, an $(\playerSubset, \decomp)$-compatible Mealy machine $\mealyMachine = \mealyTuplePure$ with the memory state space $\mealyStateSpace = \mealyStateID$.
Let $\playerIndex'\in\playerSubset$ and $\indexSegment\in\integerInterval{\numSegments}$.
We define $\mealyUpdate$ such that it leaves the memory state unchanged from a state $(\playerIAlt, \indexSegment)$ whenever an update is made with a vertex that does not occur in $\segment_\indexSegment$.
In this same situation, if $\playerIndex\neq\playerIndex'$, $\mealyNextI$ assign moves from a memoryless uniformly winning strategy of the second player in the coalition game $\game_{\playerIndex'} = (\arena_{\playerIndex'}, \reach{\target_{\playerIndex'}})$ (which exists by Theorem~\ref{thm:qualitative:ML strat}), and is left arbitrary if $\playerIndex=\playerIndex'$.
The equilibrium's stability is a consequence of Theorem~\ref{thm:charac:reach:buchi} and Lemma~\ref{lemma:qualitative:safety}: if the target of $\playerI$ is not visited in the intended outcome, all vertices along this play are not winning for the first player of the coalition game $\game_\playerIndex$.
We formalise the explanation above in the proof of the following theorem.

\begin{restatable}{theorem}{thmReachNE}\label{thm:reach:ne}
  Let $\stratProfile'$ be an NE from a vertex $\vertex_0$ in $\game$.
  There exists a finite-memory NE $\stratProfile$ from $\vertex_0$ such that $\satpl{\outcome{\stratProfile}{\vertex_0}}=\satpl{\outcome{\stratProfile'}{\vertex_0}}$ where each strategy of $\stratProfile$ has a memory size of at most
  \[\max\{1, \nPlayer - |\satpl{\outcome{\stratProfile'}{\vertex_0}}|\}\cdot\max\{1, |\vispos{\outcome{\stratProfile'}{\vertex_0}}\setminus\{0\}|\}\leq \nPlayer^2.\]
\end{restatable}
\begin{proof}
  Let $\numSegments = \max\{1, |\vispos{\outcome{\stratProfile'}{\vertex_0}}\setminus\{0\}|\}$.
  By Lemma~\ref{lem:spath:simple decomposition} (and its property~\ref{item:spath:simple decomposition:one}), there exists an NE outcome $\play$ from $\vertex_0$ that admits a simple segment decomposition $\decomp = (\segment_1, \ldots, \segment_\numSegments)$ and such that $\satpl{\play}=\satpl{\outcome{\stratProfile'}{\vertex_0}}$.
  
  Let $\playerSubset = \playerSet\setminus \satpl{\play}$ if it is not empty and otherwise let $\playerSubset = \{1\}$.
  For $\playerIndex\in\playerSubset$, let $\stratAltAdvI$ denote a memoryless strategy of the second player in the coalition game $\game_\playerIndex= (\arena_\playerIndex, \reach{\target_\playerIndex})$ that is uniformly winning on their winning region (Theorem~\ref{thm:qualitative:ML strat}).
  We let $\winningIAdv{\safe{\target_\playerIndex}}$ denote this winning region.

  Let $\playerIndex\in\playerSet$.
  We define an $(\playerSubset, \decomp)$-compatible Mealy machine $\mealyMachine_\playerIndex = (\mealyStateID, (\player{\min\playerSubset}, 1), \mealyUpdate, \mealyNextI)$ for $\playerI$ as follows (cf.~Definition~\ref{def:compatible mealy}).
  The functions $\mealyUpdate$ and $\mealyNextI$ respectively extend $\mealyUpdateID$ and $\mealyNextIDpI$ as follows.
  Let $(\playerIAlt, \indexSegment)\in\mealyStateID$ and $\vertex\in\vertexSet$ that does not occur in $\segment_\indexSegment$.
  We let $\mealyUpdate((\playerIAlt, \indexSegment), \vertex) = (\playerIAlt, \indexSegment)$.
  If $\vertex\in\vertexSetI$, we let $\mealyNextI((\playerIAlt, \indexSegment), \vertex) = \stratAltAdvIb(\vertex)$ if $\playerIndex'\neq\playerIndex$ and, let $\mealyNextI((\playerIAlt, \indexSegment), \vertex)$ be arbitrary, otherwise.

  We let $\stratI$ denote the strategy induced by $\mealyMachine_\playerIndex$.
  It is easy to see that the memory size of $\stratI$ is at most $n^2$.
  It follows from Lemma~\ref{lemma:coherence:template:extensions} that the outcome of $\stratProfile = (\stratI)_{\playerIndex\in\playerSet}$ from $\vertex_0$ is $\play$.

  We now prove that $\stratProfile$ is an NE from $\vertex_0$.
  It suffices to show that for all $\playerIndex\notin\satpl{\play}$, $\playerI$ does not have a profitable deviation.
  Fix $\playerIndex\notin\satpl{\play}$.
  We show that all histories that are consistent with the strategy profile $\stratIAdv = (\strat{\playerIndex'})_{\playerIndex'\neq\playerIndex}$ do not leave $\winningIAdv{\safe{\target_\playerIndex}}$.
  We proceed by induction on the length of histories that start in $\vertex_0$ and are consistent with $\stratIAdv$.
  To establish the property above, we show in parallel that for any such history $\history\vertex$, letting $(\playerIAlt, \indexSegment) = \widehat{\mealyUpdate}(\history)$, we have the implication: if $\vertex$ does not occur in $\segment_\indexSegment$, then $\playerIndex' = \playerIndex$.

  Before starting the induction, we remark that, by Theorem~\ref{thm:charac:reach:buchi}, all vertices occurring in $\play$ are in $\winningIAdv{\safe{\target_\playerIndex}}$.    
  The base case of the induction is the history $\vertex_0$.
  Both claims hold without issue.
  We now assume that the claim holds for a history $\history\vertex$ starting in $\vertex_0$ consistent with $\stratIAdv$ by induction, and show they hold for $\history\vertex\vertex'$, assumed consistent with $\stratIAdv$.
  Let $(\playerIAlt, j) = \widehat{\mealyUpdate}(\history)$.

  We first show that $\vertex'\in\winningIAdv{\safe{\target_\playerIndex}}$.
  Assume that $\vertex\in\vertexSetI$.
  The induction hypothesis implies that $\vertex\in\winningIAdv{\safe{\target_\playerIndex}}$.
  Therefore, Lemma~\ref{lemma:qualitative:safety} implies that all successors of $\vertex$ are in $\winningIAdv{\safe{\target_\playerIndex}}$.
  We now assume that $\vertex\notin\vertexSetI$.
  In this case, $\vertex' = \strat{\playerIndex''}(\history\vertex) = \mealyNextIAltAlt((\playerIAlt, \indexSegment), \vertex)$ for some $\playerIndex''\neq\playerIndex$.
  We consider two cases.
  If $\vertex$ occurs in $\segment_j$, then $\vertex'$ occurs in $\play$ by definition of $\mealyNextIAltAlt$.
  This implies that $\vertex\in\winningIAdv{\safe{\target_\playerIndex}}$.
  Otherwise, by the induction hypothesis, we have $\playerIndex'=\playerIndex$, which implies $\vertex' =\stratAltAdvI(\vertex)\in \winningIAdv{\safe{\target_\playerIndex}}$ (Lemma~\ref{lemma:qualitative:safety}).

  We now move on to the second half of the induction.
  Let $(\playerIAltAlt, \indexSegment') = \widehat{\mealyUpdate}(\history\vertex)$.
  By definition of $\mealyUpdate$, we have $\indexSegment'=\indexSegment+1$ if $\indexSegment<\numSegments$ and $\vertex = \last{\segment_\indexSegment}$ and $\indexSegment'=\indexSegment$ otherwise.
  It follows that $\vertex$ occurs in $\segment_\indexSegment$ if and only if it occurs in $\segment_{\indexSegment'}$.
  Assume that $\vertex'$ does not occur in $\segment_{\indexSegment'}$.
  We consider two cases.
  First, assume that $\vertex$ occurs in $\segment_{\indexSegment'}$.
  We must have $\vertex\in\vertexSetI$ by definition of $\stratIAdv$.
  The definition of $\mealyUpdateID$ ensures that $\playerIndex''=\playerIndex$.
  Second, assume that $\vertex$ does not occur in $\segment_{\indexSegment'}$.
  On the one hand, we have $\playerIndex'=\playerIndex$ by the induction hypothesis.
  On the other hand, the definition of $\mealyUpdate$ implies that $(\playerIAltAlt, \indexSegment')=(\playerIAlt, \indexSegment)$, implying $\playerIndex''=\playerIndex$.
  This ends the proof by induction.
  
  We have shown that $\playerI$ does not have a profitable deviation.
  This shows that $\stratProfile$ is an NE from $\vertex_0$.
\end{proof}

We remark that Theorem~\ref{thm:reach:ne} provides a memory bound that is linear in the number of players when no players see their target and when all players see their target.

\begin{corollary}\label{cor:reach:ne:lin}
  If there exists an NE from a vertex $\vertex_0$ such that no (resp.~all) players see their target in its outcome, then there is a finite-memory NE from $\vertex_0$ such that no (resp.~all) players  see their target in its outcome such that all strategies have a memory size of at most $\nPlayer$. 
\end{corollary}

An NE from $\vertex_0\in\vertexSet$ such that no targets are visited in its outcome in the reachability game $\game=(\arena, (\reach{\target_\playerIndex})_{\playerIndex\in\playerSet})$  is an NE in the shortest-path game $(\arena, (\costReach{\target_\playerIndex}{\weight})_{\playerIndex\in\playerSet})$ for all weight functions $\weight$.
Therefore, Theorem~\ref{thm:reach:ne} can also be applied to obtain upper bounds on the memory needed to implement an NE in a shortest-path game such that that no player visits their target.

\begin{corollary}\label{cor:short:ne:lin}
  Let $\game'=(\arena, (\costReach{\target_\playerIndex}{\weight})_{\playerIndex\in\playerSet})$ be a shortest-path game.
  If there exists an NE from a vertex $\vertex_0$ such that no players see their target in its outcome, then there is a finite-memory NE from $\vertex_0$ such that no players see their target in its outcome such that all strategies have a memory size of at most $\nPlayer$. 
\end{corollary}
\begin{proof}
  Let $\game=(\arena, (\reach{\target_\playerIndex})_{\playerIndex\in\playerSet})$ be the reachability game obtained by ignoring weights in $\game'$.
  Let $\stratProfile$ be a strategy profile and $\vertex_0\in\vertexSet$.
  Assume that $\outcome{\stratProfile}{\vertex_0}\notin\reach{\target_\playerIndex}$ for all $\playerIndex\in\playerSet$.
  To end the proof, it suffices to show that $\stratProfile$ is an NE from $\vertex_0$ in $\game$ if and only if  $\stratProfile$ is an NE from $\vertex_0$ in $\game'$ (and the result then follows from Corollary~\ref{cor:reach:ne:lin}).
  This follows from the fact that a deviation of a player with respect to $\stratProfile$ is profitable in $\game$ if and only if it is profitable in $\game'$: the cost of a player in $\game'$ is finite if and only if their target is visited.
\end{proof}
\begin{remark}
  The upcoming Theorem~\ref{thm:short:ne} provides an upper bound of $2\nPlayer$ on the memory required to implement an NE in a shortest-path game with an outcome that does not visit the target of any player.
  By using Theorem~\ref{thm:reach:ne} instead, we obtain the more precise result stated above.
\end{remark}
\subsection{Nash equilibria in shortest-path games}\label{section:fm:short}
We now consider shortest-path games.
We show that from any NE, we can derive a finite-memory NE from the same initial vertex in which the same targets are visited and the cost profile of the derived NE is preferable to that of the original one.
We construct NEs by instantiating Mealy machines from Definition~\ref{def:compatible mealy} for outcomes and decompositions given by Lemma~\ref{lem:spath:simple decomposition} in another way than Section~\ref{section:fm:reach}.
We obtain NEs with strategies of memory size at most $\nPlayer^2+2\nPlayer$.
We now fix a shortest-path game $\game=(\arena, (\costReach{\target_\playerIndex}{\weight})_{\playerIndex\in\playerSet})$ for the remainder of the section.

We generalise the strategies that we used in Example~\ref{ex:short ne}.
This example shows that only altering the construction of Theorem~\ref{thm:reach:ne} to also monitor (and punish) players whose targets are visited is not sufficient.
In other words, the approach used in the proof of Theorem~\ref{thm:reach:ne} to encode punishing strategies (i.e., freezing memory updates outside of the intended current segment) is no longer satisfactory in shortest-path games.

We modify the construction of Theorem~\ref{thm:reach:ne} as follows.
To overcome the issue mentioned above, we change the approach so players commit to punishing any player who exits the current segment of the intended outcome.
Instead of freezing memory updates if the current segment is left when the memory state is of the form $(\playerI, \indexSegment)$, the memory switches to a new memory state $\playerI$ that is never left.
This switch can only occur if $\playerI$ deviates. 
The next-move function, for this memory state, assigns moves from a punishing strategy obtained from the coalition game $\game_\playerIndex=(\arena_\playerIndex, \costReach{\target_\playerIndex}{\weight})$ by Theorem~\ref{thm:short:ML strat}, chosen to hinder $\playerI$, ensuring that in case of a deviation, $\playerI$'s cost is at least that of the original outcome.

The conditions imposed on outcomes of Lemma~\ref{lem:spath:simple decomposition} and the characterisation of Theorem~\ref{thm:charac:short} imply the correctness of this construction.
Condition~\ref{item:spath:simple decomposition:three} of Lemma~\ref{lem:spath:simple decomposition} ensures that a player cannot reach their target with a lesser cost by traversing the vertices within a segment in another order, whereas the characterisation of Theorem~\ref{thm:charac:short} guarantees that the punishing strategies sabotage deviating players sufficiently.

\begin{restatable}{theorem}{thmShortNE}\label{thm:short:ne}
  Let $\stratProfile'$ be an NE from a vertex $\vertex_0$.
  There exists a finite-memory NE $\stratProfile$ from $\vertex_0$ such that $\satpl{\outcome{\stratProfile}{\vertex_0}}=\satpl{\outcome{\stratProfile'}{\vertex_0}}$ and, for all $\playerIndex\in\playerSet$, $\costReach{\target_\playerIndex}{\weight}(\outcome{\stratProfile}{\vertex_0})\leq \costReach{\target_\playerIndex}{\weight}(\outcome{\stratProfile'}{\vertex_0})$
  where each strategy of $\stratProfile$ has a memory size of at most
  \[\nPlayer\cdot (|\vispos{\outcome{\stratProfile'}{\vertex_0}}\setminus\{0\}| + 2)\leq \nPlayer^2 + 2\nPlayer\]
\end{restatable}
\begin{proof}
  Let $\numSegments = |\vispos{\outcome{\stratProfile'}{\vertex_0}}\setminus\{0\}|$.
  By Lemma~\ref{lem:spath:simple decomposition}, there exists an NE outcome $\play$ from $\vertex_0$ which admits a simple segment decomposition $\decomp = (\segment_1, \ldots, \segment_{\numSegments+1})$ satisfying conditions~\ref{item:spath:simple decomposition:two}-\ref{item:spath:simple decomposition:four} of Lemma~\ref{lem:spath:simple decomposition}.

  For $\playerIndex\in\playerSet$, let $\stratAltAdvI$ denote a memoryless strategy of the second player in the coalition game $\game_\playerIndex= (\arena_\playerIndex, \costReach{\target_\playerIndex}{\weight})$ such that $\stratAltAdvI$ is uniformly winning on their winning region $\winningIAdv{\safe{\target_\playerIndex}}$ in the reachability game $(\arena_\playerIndex, \reach{\target_\playerIndex})$ and such that $\stratAltAdvI$ ensures a cost of at least $\min\{\valI(\vertex), \costReach{\target_\playerIndex}{\weight}(\play)\}$ from any $\vertex\in\vertexSet$, where $\valI(\vertex)$ denotes the value of $\vertex$ in $\game_\playerIndex$ (Theorem~\ref{thm:short:ML strat}).
  
  We work with $\playerSubset = \playerSet$ in the following and drop $\playerSubset$ from the notation of Definition~\ref{def:compatible mealy} to lighten it.
  Let $\playerIndex\in\playerSet$.
  We define an $(\playerSet, \decomp)$-compatible Mealy machine $\mealyMachine_\playerIndex = (\mealyStateSpace, (\player{1}, 1), \mealyUpdate, \mealyNextI)$ for $\playerI$ as follows.
  We let $\mealyStateSpace = \mealyStateD\cup\{\playerI\mid\playerIndex\in\playerSet\}$.
  The functions $\mealyUpdate$ and $\mealyNextI$ extend $\mealyUpdateD$ and $\mealyNextDpI$ as follows.
  For all $(\playerIAlt, \indexSegment)\in\mealyStateD$ and $\vertex\in\vertexSet$ that does not occur in $\segment_\indexSegment$, we let $\mealyUpdate((\playerIAlt, \indexSegment), \vertex) = \playerIAlt$ and, if $\vertex\in\vertexSetI$, we let $\mealyNextI((\playerIAlt, \indexSegment), \vertex) = \stratAltAdvIb(\vertex)$ if $\playerIndex'\neq\playerIndex$ and $\mealyNextI((\playerI, \indexSegment), \vertex)$ is left arbitrary otherwise.
  For all $\playerIndex'\in\playerSet$ and $\vertex\in\vertexSet$, we let $\mealyUpdate(\playerIAlt, \vertex) = \playerIAlt$ and, if $\vertex\in\vertexSetI$, we let $\mealyNextI(\playerIAlt, \vertex) = \stratAltAdvIb(\vertex)$ if $\playerIndex'\neq\playerIndex$ and $\mealyNextI(\playerI, \vertex)$ is left arbitrary.
  
  We let $\stratI$ denote the strategy induced by $\mealyMachine_\playerIndex$.
  We have $|\mealyStateSpace| = \nPlayer\cdot(\numSegments + 2)$, therefore the memory size of $\stratI$ satisfies the upper bound of the theorem.
  Furthermore, it follows from Lemma~\ref{lemma:coherence:template:extensions} that the outcome of $\stratProfile = (\stratI)_{\playerIndex\in\playerSet}$ from $\vertex_0$ is $\play$ and $\stratI$ is based on $\decomp$ for all $\playerIndex\in\playerSet$.
  
  We now establish that $\stratProfile$ is an NE from $\vertex_0$.
  Let $\playerIndex\in\playerSet$.
  Let $\play'$ be a play starting in $\vertex_0$ consistent with $\stratIAdv = (\strat{\playerIndex'})_{\playerIndex'\neq\playerIndex}$.
  To end the proof, it suffices to show that we have $\costReach{\target_\playerIndex}{\weight}(\play') \geq \costReach{\target_\playerIndex}{\weight}(\play)$.

  We first show the following claim.
  If some prefix of $\play'$ is not coherent with $\decomp$, then there exists $\indexPosition\in\IN$ such that $\prefix{\play'}{\indexPosition}$ is the longest prefix of $\play'$ coherent with $\decomp$ and $\suffix{\play'}{\indexPosition}$ is a play that is consistent with $\stratAltAdvI$.
  Assume that some prefix of $\play'$ is not coherent with $\decomp$.
  Let $\indexPosition\in\IN$ such that $\prefix{\play'}{\indexPosition}$ is the longest prefix of $\play'$ coherent with $\decomp$, and assume it is $\indexSegment$-coherent.
  As the strategies of $\stratIAdv$ are based on $\decomp$, we must have $\first{\suffix{\play'}{\indexPosition}}\in\vertexSetI$.
  Lemma~\ref{lemma:coherence:template:extensions} and the definition of $\mealyUpdate$ ensure that $\widehat{\mealyUpdate}(\prefix{\play'}{\indexPosition})=(\playerI, \indexSegment)$.
  It follows from $\prefix{\play'}{\indexPosition+1}$ being inconsistent with $\decomp$ that its last vertex does not occur in $\segment_\indexSegment$.
  The definitions of $\mealyUpdate$ and $\mealyNextIAlt$ for $\playerIndex'\neq\playerIndex$ combined with the above ensure that $\suffix{\play'}{\indexPosition}$ is consistent with $\stratAltAdvI$.

  We now show that $\costReach{\target_\playerIndex}{\weight}(\play') \geq \costReach{\target_\playerIndex}{\weight}(\play)$.
  We first assume that $\playerIndex\notin\satpl{\play}$.
  We establish that $\play'\notin\reach{\target_\playerIndex}$.
  By Theorem~\ref{thm:charac:short}, all vertices occurring in $\play$ belong to $\winningIAdv{\safe{\target_\playerIndex}}$.
  Therefore, if all prefixes of $\play'$ are coherent with $\decomp$, as all vertices of $\play'$ occur in $\play$, it holds that $\target_\playerIndex$ is not visited in $\play'$.
  Otherwise, let $\indexPosition\in\IN$ such that $\prefix{\play'}{\indexPosition}$ is the longest prefix of $\play'$ that is coherent with $\decomp$ and $\suffix{\play'}{\indexPosition}$ is consistent with $\stratAltAdvI$.
  No vertices of $\target_\playerIndex$ occur in $\prefix{\play'}{\indexPosition}$ by coherence with $\decomp$.
  It follows from the coherence of $\prefix{\play'}{\indexPosition}$ with $\decomp$ that $\first{\suffix{\play'}{\indexPosition}}=\last{\prefix{\play'}{\indexPosition}}$ occurs in $\play$.
  We obtain that $\first{\suffix{\play'}{\indexPosition}}\in \winningIAdv{\safe{\target_\playerIndex}}$, therefore $\suffix{\play}{\indexPosition}\notin\reach{\target_\playerIndex}$.
  We have shown that for all $\playerIndex\notin\satpl{\play}$, we have $\costReach{\target_\playerIndex}{\weight}(\play') \geq \costReach{\target_\playerIndex}{\weight}(\play)$.

  We now assume that $\playerIndex\in\satpl{\play}$.
  The desired inequality is immediate if $\target_\playerIndex$ is not visited in $\play'$.
  Similarly, it holds directly if $\vertex_0\in\target_\playerIndex$.
  We therefore assume that we are in neither of the previous two cases.
  We write the shortest prefix of $\play'$ ending in $\target_\playerIndex$ (the weight of which is $\costReach{\target_\playerIndex}{\weight}(\play')$) as a combination $\concat{\history}{\history'}$ where $\history$ is its longest prefix that is coherent with $\decomp$.
  We note that $\history'$ is consistent with $\stratAltAdvI$ because $\history$ is a prefix of the longest prefix of $\play'$ coherent with $\decomp$: if $\history$ is a strict prefix, then $\history'$ is a trivial history, and otherwise it follows from the above.
  
  We provide lower bounds on the weights of $\history$ and $\history'$.
  Assume that $\history$ is $\indexSegment$-coherent.
  By definition of coherence, we can write $\history$ as a history combination $\concat{\history_1}{\concat{\ldots}{\history_\indexSegment}}$ where, for all $\indexSegment' <\indexSegment$, $\history_{\indexSegment'}$ shares its first and last vertices with $\segment_{\indexSegment'}$ and contains only vertices of $\segment_{\indexSegment'}$, and $\history_\indexSegment$ shares its first vertex with $\segment_\indexSegment$ and contains only vertices of $\segment_\indexSegment$.
  Let $\segment_\indexSegment'$ be the prefix of $\segment_\indexSegment$ up to $\last{\history_\indexSegment}$.
  
  On the one hand, we have $\sum_{\indexSegment'< \indexSegment}\weight(\history_{\indexSegment'}) \geq \sum_{\indexSegment'< \indexSegment} \weight(\segment_{\indexSegment'})$ and $\weight(\history_\indexSegment) \geq \weight(\segment_\indexSegment')$.
  This follows from $\play$ satisfying property~\ref{item:spath:simple decomposition:three} of Lemma~\ref{lem:spath:simple decomposition} (for $\history_\indexSegment$, having $\weight(\history_\indexSegment) < \weight(\segment_\indexSegment')$ would contradict property~\ref{item:spath:simple decomposition:three} with respect to $\segment_\indexSegment$).
  On the other hand, by choice of $\stratAltAdvI$, we obtain that $\weight(\history')\geq\min\{\val_\playerIndex(\first{\history'}), \costReach{\target_\playerIndex}{\weight}(\play)\}$.
  From the characterisation of Theorem~\ref{thm:charac:short}, we obtain $\val_\playerIndex(\first{\history'})\geq \costReach{\target_\playerIndex}{\weight}(\play) - \weight(\concat{\concat{\segment_1}{\concat{\ldots}{\segment_{j-1}}}}{\segment_j'})$.
  By combining these inequalities, we obtain $\costReach{\target_\playerIndex}{\weight}(\play')\geq \costReach{\target_\playerIndex}{\weight}(\play)$, ending the proof.
\end{proof}

\section{Finite-memory Nash equilibria in safety games}\label{section:safety}
We now move on to safety games.
We show that from any NE, we can derive a finite-memory NE from the same initial vertex such that any safety objective satisfied by the outcome of the first NE also is satisfied by the outcome of the second one.
In contrast to reachability games, we do not necessarily preserve the set of winning players and may obtain a superset thereof.
We use the same construction as for shortest-path games and obtain the same arena-independent memory upper bounds.

This section is structured as follows.
In Section~\ref{section:safety:example}, we provide examples that illustrate some properties of the decomposition-based approach for NEs in safety games.
We then show how to simplify NE outcomes of safety games such that we can derive finite-memory NEs from the simplified outcomes in Section~\ref{section:safety:outcomes}.
Finally, we show that the finite-memory NE construction previously used for shortest-path games also yields NEs from these outcomes in Section~\ref{section:safety:equilibria}.

We fix an arena $\arena=\arenaTuple$, sets $\target_1$, \ldots, $\target_\nPlayer\subseteq\vertexSet$ (to be avoided) and a safety game $\game =(\arena, (\safe{\target_\indexPlayer})_{\indexPlayer\in\playerSet})$ for the remainder of this section.

\subsection{Examples}\label{section:safety:example}
In the following, we will construct finite-memory NEs in safety games by adapting the construction used in shortest-path games.
When simplifying NE outcomes in safety games, the resulting plays are simple plays or simple lassos.
Such plays admit trivial decompositions.
The first example below illustrates that, in spite of this, we may require decompositions with more than one segment to construct NEs.

  \begin{figure}
    \centering
    \begin{subfigure}[b]{0.31\textwidth}
      \centering
      \begin{tikzpicture}[node distance=0.5cm]
        \node[state, circle, align=center] (v0) {$\vertex_0$};
        \node[state, circle, align=center, right = of v0] (v1) {$\vertex_1$};
        \node[state, square, align=center, right = of v1] (v2) {$\vertex_2$};
        \node[state, circle, align=center, right = of v2] (v3) {$\vertex_3$};
        \node[state, circle, align=center, above = of v3] (v4) {$\vertex_4$};
        \node[state, circle, align=center, below = of v3] (v5) {$\vertex_5$};
        \path[->] (v0) edge[bend left] (v2);
        \path[->] (v0) edge (v1);
        \path[->] (v1) edge[bend left] (v2);
        \path[->] (v2) edge[bend left] (v1);
        \path[->] (v2) edge (v3);
        \path[->] (v3) edge (v4);
        \path[->] (v3) edge (v5);
        \path[->] (v4) edge[loop left] (v4);
        \path[->] (v5) edge[loop left] (v5);
      \end{tikzpicture}
      \caption{If the objective of $\playerOne$ is $\safe{\{\vertex_1, \vertex_4\}}$, $\playerOne$ has a profitable deviation from $\vertex_0$ with respect to any strategy profile based on the simple decomposition $(\vertex_0\vertex_1\vertex_2\vertex_3\vertex_4^\omega)$.}
      \label{fig:ex:safe:one}
    \end{subfigure}
    \hfill
    \begin{subfigure}[b]{0.31\textwidth}
      \centering
      \begin{tikzpicture}[node distance=7mm]
        \node[state, square, align=center] (v0) {$\vertex_0$};
        \node[state, circle, align=center, right = of v0] (v1) {$\vertex_1$};
        \node[state, circle, align=center, below = 4mm of v1] (v2) {$\vertex_2$};
        \path[->] (v0) edge[bend left] (v1);
        \path[->] (v0) edge (v2);
        \path[->] (v1) edge[bend left] (v0);
        \path[->] (v1) edge (v2);
        \path[->] (v2) edge[loop below] (v2);
      \end{tikzpicture}
      \caption{If the objective of $\playerOne$ is $\safe{\{\vertex_2\}}$, $\playerOne$ has a profitable deviation from $\vertex_0$ with respect to any strategy profile based on the simple decomposition $(\vertex_0\vertex_1\vertex_2, \vertex_2^\omega)$.}\label{fig:ex:safe:two}
    \end{subfigure}
    \hfill
    \begin{subfigure}[b]{0.31\textwidth}
      \centering
      \begin{tikzpicture}[node distance=0.5cm]
        \node[state, circle, align=center] (v0) {$\vertex_0$};
        \node[state, circle, align=center, right = of v0] (v1) {$\vertex_1$};
        \node[state, circle, align=center, above = 4mm of v1] (v2) {$\vertex_2$};
        \node[state, square, align=center, right = of v1] (v3) {$\vertex_3$};
        \node[state, square, align=center, right = of v3] (v4) {$\vertex_4$};
        \path[->] (v0) edge (v1);
        \path[->] (v0) edge (v2);
        \path[->] (v1) edge (v3);
        \path[->] (v2) edge (v3);
        \path[->] (v3) edge (v4);
        \path[->] (v3) edge[loop below] (v3);
        \path[->] (v4) edge[loop below] (v4);
      \end{tikzpicture}
      \caption{If the objectives of $\playerOne$ and $\playerTwo$ are respectively $\safe{\{\vertex_1,\vertex_4\}}$ and $\safe{\{\vertex_2\}}$, $\playerTwo$ must commit to a punishing strategy if $\playerOne$ deviates from the outcome $\vertex_0\vertex_1\vertex_3^\omega$.}\label{fig:ex:safe:three}
    \end{subfigure}
    \caption{Two-player arenas. Circle and squares are respectively $\playerOne$ and $\playerTwo$ vertices.}
    \label{fig:ex:safe}
  \end{figure}

\begin{example}\label{example:safe:decomposition}
  We consider the safety game played on the arena of Figure~\ref{fig:ex:safe:one} such that the objectives of $\playerOne$ and $\playerTwo$ are respectively $\safe{\{\vertex_1, \vertex_4\}}$ and $\safe{\{\vertex_5\}}$.
  The play $\play = \vertex_0\vertex_1\vertex_2\vertex_3\vertex_4^\omega$ is an NE outcome that satisfies only the objective of $\playerTwo$.
  As this play is a simple lasso, it admits a trivial simple segment decomposition.

  We claim that there are no NEs  in which $\playerTwo$ follows a strategy based on the decomposition $(\play)$.
  Indeed, if $\playerTwo$ follows such a strategy, $\playerOne$ can win with a strategy that moves from $\vertex_0$ to $\vertex_2$ and from $\vertex_3$ to $\vertex_5$ as $\playerTwo$ will move from $\vertex_2$ to $\vertex_3$ in this case.

  We note that there exists an NE with outcome $\play$ based on the simple decomposition $(\vertex_0\vertex_1, \vertex_1\vertex_2\vertex_3\vertex_4^\omega)$.
  On the one hand, if $\playerOne$ moves from $\vertex_0$ to $\vertex_2$, $\playerTwo$ detects this deviation and will punish $\playerOne$ by moving back to $\vertex_1$.
  On the other hand, $\playerOne$ has no incentive to deviate after $\vertex_1$ is visited, as their safety objective can no longer be satisfied.
  \hfill$\lhd$
\end{example}

In reachability and shortest-path games, players have no incentive to stay within a segment of a play decomposition.
This is no longer the case for safety games: stalling the progress of the play in a segment by exploiting the decomposition-based behaviour of the other players can allow a player to win by preventing a visit to undesirable vertices.
We illustrate this property with the following example.

\begin{example}\label{example:safe:in-segment}
  We consider the safety game played on the arena of Figure~\ref{fig:ex:safe:two} such that the objective of $\playerOne$ is $\safe{\{\vertex_2\}}$ and the objective of $\playerTwo$ is $\safe{\emptyset}$ (i.e., $\playerTwo$ wins no matter what).
  The play $\play = \vertex_0\vertex_1\vertex_2^\omega$ is an NE outcome from $\vertex_0$ (see the characterisation of Theorem~\ref{thm:charac:safe}).
  However, there are no NEs where the two players use a strategy based on the decompositions $(\play)$ or $(\vertex_0\vertex_1\vertex_2, \vertex_2^\omega)$ of $\play$.
  Indeed, given such a strategy profile, $\playerOne$ has a profitable deviation: moving from $\vertex_1$ back to $\vertex_0$.
  \hfill $\lhd$
\end{example}

The previous example illustrates that players can have profitable deviations \textit{within segments of a decomposition} in safety games.
However, such deviations are not harmful to the other players: we remain within the set of vertices of the original outcome.
In particular, if such a profitable deviation exists, we can consider the play resulting from this deviation, as it is an NE outcome with more winning players, and work from this play instead.

Finally, we provide an example showing that we cannot use the same approach for punishing strategies as in reachability games: players should commit to punishing anyone who leaves the set of vertices of the ongoing segment of the intended outcome.

\begin{example}\label{example:safety:punishment}
  We consider the safety game played on the arena of Figure~\ref{fig:ex:safe:three} in which the objective of $\playerOne$ is $\safe{\{\vertex_1, \vertex_4\}}$ and the objective of $\playerTwo$ is $\safe{\{\vertex_2\}}$.
  In this game, at most one player can win from $\vertex_0$.
  
  The play $\play = \vertex_0\vertex_1\vertex_3^\omega$ is an NE outcome in this game that satisfies the objective of $\playerTwo$, but not the objective of $\playerOne$.
  There exists an NE from $\vertex_0$ resulting in this outcome in which the two players follow strategies based on the trivial decomposition $(\play)$ of $\play$.
  However, such an NE requires that $\playerTwo$ remembers whether $\vertex_2$ is visited (i.e., whether a vertex not in $\play$ has been visited), to be able to prevent $\playerOne$ from having a profitable deviation by going through $\vertex_2$.
  In other words, if $\playerOne$ leaves the set of vertices of $\play$, $\playerTwo$ must commit to punishing $\playerOne$ to ensure the stability of the equilibrium.
  \hfill $\lhd$
\end{example}

In light of the previous example, we will adapt the technique from the proof of Theorem~\ref{thm:short:ne} that was used to construct finite-memory NEs in shortest-path games.

\subsection{Well-structured Nash equilibrium outcomes}\label{section:safety:outcomes}

In this section, we formulate a variant of Lemma~\ref{lem:spath:simple decomposition} for safety games.
We transform an NE outcome into another NE outcome that is a simple play or a simple lasso.
Let $\play'=\vertex_0\vertex_1\ldots$ be an NE outcome in $\game$.
If it is possible to move from some vertex $\vertex_\indexLast$ to a vertex $\vertex_\indexPosition$ with $\indexPosition\leq\indexLast$, we let $\play$ be a simple lasso of the form $\vertex_0\vertex_1\ldots(\vertex_{\indexPosition}\ldots\vertex_{\indexLast})^\omega$ that exploits the earliest such available edge along $\play'$.
Otherwise, we define $\play = \play'$, as the absence of such edges implies that $\play'$ is already a simple play.
The resulting play $\play$ is an NE outcome by Theorem~\ref{thm:charac:safe}: target sets only occur in the common prefix of $\play$ and $\play'$.
By construction, there may be fewer vertices occurring in $\play$ than in $\play'$ and thus $\satpl{\play'}\subseteq\satpl{\play}$ .
While this construction may enlarge the set of winning players, Example~\ref{example:safe:in-segment} illustrates that decomposition-based strategies are not sufficient in general to preserve a set of winning players exactly.
To construct a decomposition-based NEs from the simplified play $\play$, we will also require a well-chosen decomposition of $\play$ (see Example~\ref{example:safe:decomposition}).
We use a decomposition analogous to those used in Lemma~\ref{lem:spath:simple decomposition}: we decompose $\play$ into segments between two consecutive visits to targets.

We formalise the above simplification process and the properties imposed on the simplified play in the following lemma.

\begin{lemma}\label{lem:safety:simple decomposition}
  Let $\play'$ be the outcome of an NE from $\vertex_0\in\vertexSet$ in the safety game $\game$.
  There exist $\numSegments \leq |\vispos{\play'}\setminus\{0\}|$ and an NE outcome $\play$ from $\vertex_0$ with $\vispos{\play}\setminus\{0\} = \{\indexPosition_1 < \ldots < \indexPosition_\numSegments\}$ that admits a finite simple segment decomposition $(\segment_1, \ldots, \segment_{\numSegments+1})$ such that
  \begin{enumerate}[(i)]
  \item $\play$ is a simple play or a simple lasso;\label{item:safe:decomp:one}
  \item if $\numSegments>0$, then for all $\indexSegment\in\integerInterval{\numSegments}$, $\concat{\concat{\segment_1}{\ldots}}{\segment_\indexSegment}=\prefix{\play}{\indexPosition_\indexSegment}$;\label{item:safe:decomp:two}
  \item if $\numSegments>0$, then for all $\indexSegment\in\integerInterval{\numSegments}$, all cycles whose vertices all occur in $\concat{\concat{\segment_1}{\ldots}}{\segment_\indexSegment}$ must contain $\last{\segment_\indexSegment}$; and\label{item:safe:decomp:three}
  \item $\satpl{\play'}\subseteq\satpl{\play}$.\label{item:safe:decomp:four}
  \end{enumerate}
\end{lemma}
\begin{proof}
  We let $\play' = \vertex_0\vertex_1\ldots$ and distinguish two cases.

  First, we assume that there exist some indices $\indexPosition\leq\indexLast$ such that $(\vertex_\indexLast, \vertex_\indexPosition)\in\edgeSet$.
  We let $\indexLast^\star$ be the smallest index for which there exists such an $\indexPosition$, and let $\indexPosition^\star$ be the smallest among these.
  We define $\play = \vertex_0\vertex_1\ldots(\vertex_{\indexPosition^\star}\ldots\vertex_{\indexLast^\star})^\omega$, i.e., we take the first edge from a vertex in $\play'$ that leads to a previously visited vertex and loop in the resulting simple cycle.
  The play $\play$ is a simple lasso and thus satisfies~\ref{item:safe:decomp:one}.
  Its construction implies that~\ref{item:safe:decomp:four} also holds.
  By the characterisation of Theorem~\ref{thm:charac:safe}, $\play$ is an NE outcome (as this characterisation applies to $\play'$).

  We now define $\numSegments$ and the required simple decomposition $(\segment_1, \ldots, \segment_{\numSegments+1})$ of $\play$.
  By construction, we have $\vispos{\play} \subseteq \vispos{\play'}$.
  We let $\numSegments =  |\vispos{\play}\setminus\{0\}| \leq |\vispos{\play'}\setminus\{0\}|$.
  If $\numSegments = 0$, we consider the trivial decomposition $(\play)$ of $\play$.
  In this case, properties~\ref{item:safe:decomp:two} and~\ref{item:safe:decomp:three} hold trivially.
  We now assume that $\numSegments > 0$ and write $\vispos{\play}\cup \{0\} = \{\indexPosition_0 < \indexPosition_1 < \ldots < \indexPosition_\numSegments\}$.
  For all $\indexSegment\in\integerInterval{\numSegments}$, we define $\segment_\indexSegment = \vertex_{\indexPosition_{\indexSegment-1}}\ldots\vertex_{\indexPosition_{\indexSegment}}$ and we let $\segment_{\numSegments+1} = \suffix{\play}{\indexPosition_{\numSegments}}$.
  Property~\ref{item:safe:decomp:two} follows by definition of $(\segment_1, \ldots, \segment_{\numSegments+1})$ and~\ref{item:safe:decomp:three} follows from the choice of $\indexLast^\star$: if there were a cycle witnessing the violation of~\ref{item:safe:decomp:three}, then the choice of $\indexLast^\star$ would not have been minimal.
  This closes the first case.
  
  Second, we assume that for all indices $\indexPosition<\indexLast$, we do not have $(\vertex_\indexLast, \vertex_\indexPosition)\in\edgeSet$.
  In this case, $\play'$ is a simple play: if there were a cycle in $\play'$, then we would fall in the first case.
  We let $\play = \play'$, and obtain~\ref{item:safe:decomp:one} and~\ref{item:safe:decomp:four} immediately.
  We define $\numSegments = |\vispos{\play}\setminus\{0\}|$ and a decomposition $(\segment_1, \ldots, \segment_{\numSegments+1})$ of $\play$ as in the previous case (we omit the construction to avoid redundancy), and~\ref{item:safe:decomp:two} follows in the same way.
  Finally, we can show that~\ref{item:safe:decomp:three} holds with the same argument as above.
\end{proof}
\subsection{Nash equilibria in safety games}\label{section:safety:equilibria}
We build on the NE outcomes given by Lemma~\ref{lem:safety:simple decomposition} to show that we can derive, from any NE, another NE from the same initial vertex induced by Mealy machines of size quadratic in the number of players and whose outcome satisfies all of the safety objectives satisfied by the outcome of the original NE.
We adapt the construction used for shortest-path games in the proof of Theorem~\ref{thm:short:ne}: the only difference is that we need only monitor players whose objectives are not satisfied (as in reachability games), as the others cannot have a profitable deviation.
To prove that the strategy profile that we construct is an NE, we rely on property~\ref{item:safe:decomp:three} of outcomes obtained by Lemma~\ref{lem:safety:simple decomposition}.
Intuitively, this property guarantees that players cannot halt the progress of the play in a segment of the decomposition on which we base the strategy to obtain a profitable deviation.
Furthermore, players cannot have profitable deviations by leaving the current segment of the play: the other players switch to a punishing strategy of the player and, if the safety objective of the deviating player has not been falsified yet, then the NE outcome characterisation of Theorem~\ref{thm:charac:safe} guarantees the effectiveness of the punishment by the other players.

\begin{theorem}\label{thm:safe:ne}
  Let $\stratProfile'$ be an NE from a vertex $\vertex_0$ in $\game$.
  There exists a finite-memory NE $\stratProfile$ from $\vertex_0$ such that $\satpl{\outcome{\stratProfile'}{\vertex_0}}\subseteq\satpl{\outcome{\stratProfile}{\vertex_0}}$
  where each strategy of $\stratProfile$ has a memory size of at most
  \[\max\{1, \nPlayer - |\satpl{\outcome{\stratProfile'}{\vertex_0}}|\} \cdot (|\vispos{\outcome{\stratProfile'}{\vertex_0}}\setminus\{0\}| + 2)\leq \nPlayer^2 + 2\nPlayer.\]
\end{theorem}
\begin{proof}
  By Lemma~\ref{lem:safety:simple decomposition} applied to $\outcome{\stratProfile'}{\vertex_0}$, there exists an NE outcome $\play$ from $\vertex_0$ with a simple segment decomposition $\decomp = (\segment_1, \ldots, \segment_{\numSegments+1})$ where $\numSegments \leq |\vispos{\outcome{\stratProfile'}{\vertex_0}}\setminus\{0\}|$ satisfying conditions~\ref{item:safe:decomp:one}-\ref{item:safe:decomp:four} of Lemma~\ref{lem:safety:simple decomposition}.
  To prove the statement of the theorem, it suffices to build a finite-memory NE with outcome $\play$ from $\vertex_0$.
  For $\playerIndex\in\playerSet$, let $\stratAltAdvI$ denote a memoryless strategy of the second player in the coalition game $\game_\playerIndex= (\arena_\playerIndex, \safe{\target_\playerIndex})$ such that $\stratAltAdvI$ is uniformly winning on their winning region $\winningIAdv{\reach{\target_\playerIndex}}$ (cf.~Theorem~\ref{thm:qualitative:ML strat}).

  We let $\playerSubset = \playerSet\setminus\satpl{\play}$ if this set is non-empty, and $\playerSubset =\{1\}$ otherwise.
  Let $\playerIndex\in\playerSet$.
    We define an $(\playerSubset, \decomp)$-compatible Mealy machine $\mealyMachine_\playerIndex = (\mealyStateSpace, (\player{\min\playerSubset}, 1), \mealyUpdate, \mealyNextI)$ for $\playerI$ as follows (cf.~Definition~\ref{def:compatible mealy}).
  We let $\mealyStateSpace = \mealyStateID\cup\{\playerI\mid\playerIndex\in\playerSubset\}$.
  The functions $\mealyUpdate$ and $\mealyNextI$ extend $\mealyUpdateID$ and $\mealyNextIDpI$ as follows.
  For all $(\playerIAlt, \indexSegment)\in\mealyStateID$ and $\vertex\in\vertexSet$ that does not occur in $\segment_\indexSegment$, we let $\mealyUpdate((\playerIAlt, \indexSegment), \vertex) = \playerIAlt$ and, if $\vertex\in\vertexSetI$, we let $\mealyNextI((\playerIAlt, \indexSegment), \vertex) = \stratAltAdvIb(\vertex)$ if $\playerIndex'\neq\playerIndex$ and $\mealyNextI((\playerI, \indexSegment), \vertex)$ is left arbitrary otherwise.
  For all $\playerIndex'\in\playerSet$ and $\vertex\in\vertexSet$, we let $\mealyUpdate(\playerIAlt, \vertex) = \playerIAlt$ and, if $\vertex\in\vertexSetI$, we let $\mealyNextI(\playerIAlt, \vertex) = \stratAltAdvIb(\vertex)$ if $\playerIndex'\neq\playerIndex$ and $\mealyNextI(\playerI, \vertex)$ is left arbitrary otherwise.
  
  We let $\stratI$ denote the strategy induced by $\mealyMachine_\playerIndex$.
  We have $|\mealyStateSpace| = \max\{1, \nPlayer - |\satpl{\play}|\}\cdot(\numSegments + 2)$, therefore the memory size of $\stratI$ satisfies the announced upper bound.
  Furthermore, it follows from Lemma~\ref{lemma:coherence:template:extensions} that the outcome of $\stratProfile = (\stratI)_{\playerIndex\in\playerSet}$ from $\vertex_0$ is $\play$ and $\stratI$ is based on $\decomp$ for all $\playerIndex\in\playerSet$.
  
  We now establish that $\stratProfile$ is an NE from $\vertex_0$.
  Players in $\satpl{\play}$ cannot have profitable deviations as their objective is satisfied.
  We thus let $\playerIndex\in\playerSet\setminus\satpl{\play}$ and show that $\playerI$ does not have a profitable deviation to end the proof.
  Let $\play'$ be a play starting in $\vertex_0$ consistent with $\stratIAdv = (\strat{\playerIndex'})_{\playerIndex'\neq\playerIndex}$.
  We distinguish two cases depending on whether all prefixes of $\play'$ are coherent with $\decomp$.

  First, assume that all prefixes of $\play'$ are coherent with $\decomp$.
  By property~\ref{item:safe:decomp:three} of Lemma~\ref{lem:safety:simple decomposition} for $\play$, in $\play$, there are no edges leading to earlier vertices from any vertex in the segments $\segment_\indexSegment$ for $\indexSegment \leq \numSegments$.
  It follows that all targets that are visited along $\play$ also are along $\play'$, by property~\ref{item:safe:decomp:two} of Lemma~\ref{lem:safety:simple decomposition} for $\play$.
  In particular, we obtain that $\play'\notin\safe{\target_\indexPlayer}$.

  We now assume that there is a prefix of $\play'$ that is not coherent with $\decomp$.
  Let $\indexPosition\in\IN$ such that $\prefix{\play'}{\indexPosition}$ is the longest prefix of $\play'$ coherent with $\decomp$.
  We show that $\suffix{\play'}{\indexPosition}$ is a play that is consistent with $\stratAltAdvI$.
  Assume that $\prefix{\play'}{\indexPosition}$ is $\indexSegment$-coherent with $\decomp$.
  As the strategies of $\stratIAdv$ are based on $\decomp$, we must have $\first{\suffix{\play'}{\indexPosition}}\in\vertexSetI$.
  Lemma~\ref{lemma:coherence:template:extensions} and the definition of $\mealyUpdate$ ensure that $\widehat{\mealyUpdate}(\prefix{\play'}{\indexPosition})=(\playerI, \indexSegment)$.
  It follows from $\prefix{\play'}{\indexPosition+1}$ not being coherent with $\decomp$ that its last vertex does not occur in $\segment_\indexSegment$.
  The definitions of $\mealyUpdate$ and $\mealyNextIAlt$ for $\playerIndex'\neq\playerIndex$ combined with the above ensure that $\suffix{\play'}{\indexPosition}$ is consistent with $\stratAltAdvI$.

  We now distinguish two sub-cases.
  First, assume that a target of $\playerI$ occurs in $\prefix{\play}{\indexPosition}$, i.e., $\last{\segment_{\indexSegment'}}\in\target_\playerIndex$ for some $\indexSegment' < \indexSegment$.
  In this case, $\last{\segment_{\indexSegment'}}$ occurs in $\play'$ by $\indexSegment$-coherence of $\prefix{\play'}{\indexPosition}$, and thus $\play'\notin\safe{\target_\indexPlayer}$.
  We now assume that no targets of $\playerI$ occur in $\prefix{\play}{\indexPosition}$.
  It follows from Theorem~\ref{thm:charac:safe} that $\first{\suffix{\play'}{\indexPosition}}\in\winningIAdv{\reach{\target_\playerIndex}}$, and thus consistency of $\suffix{\play'}{\indexPosition}$ with $\stratIAdv$ yields $\suffix{\play'}{\indexPosition}\notin\safe{\target_\indexPlayer}$, which implies $\play'\notin\safe{\target_\indexPlayer}$.
\end{proof}

\section{Finite-memory Nash equilibria in Büchi and co-Büchi games}\label{section:buchi}

This section presents our results for Büchi and co-Büchi games.
For reachability, safety and shortest-path games, we have constructed finite-memory Nash equilibria from NE outcomes admitting a well-structured segment decomposition, and the resulting NEs have a memory size at most quadratic in the number of players (cf. Theorems~\ref{thm:reach:ne},~\ref{thm:short:ne} and~\ref{thm:safe:ne}).
In Büchi and co-Büchi games, we cannot directly use the same approach to construct finite-memory NEs.
Instead, we combine the classical idea of encoding an outcome to be enforced in the memory structure and decomposition-based strategies to show that from any NE in a Büchi or co-Büchi game, we can derive a finite-memory NE from the same initial state with a larger set of winning players with respect to set inclusion.

In Section~\ref{section:buchi:ex}, we present several examples, including one illustrating that, unlike reachability, safety and shortest-path games, we cannot obtain arena-independent upper memory bounds for constrained NEs in Büchi and co-Büchi games.
We accompany these examples with intuition regarding the construction presented in the following sections.
We then provide our results for Büchi games in Section~\ref{section:buchi:buchi-ne} and for co-Büchi games in Section~\ref{section:buchi:cobuchi-ne}.

We fix an arena $\arena = \arenaTuple$ and targets $\target_1, \ldots, \target_\nPlayer\subseteq\vertexSet$ for this entire section.

\subsection{Examples}\label{section:buchi:ex}

In the games considered in the previous sections, we constructed NEs based on simple segment decompositions of NE outcomes such that each finite segment ends in the first occurrence of a target in the outcome.
In these NEs, players do not react to \textit{in-segment deviations}, i.e., to players that deviate without leaving the current segment of the intended outcome.
In reachability and shortest-path games, there cannot be profitable in-segment deviations as the set of vertices of the intended outcome is never left.
In safety games, we imposed conditions on the NE outcomes of Lemma~\ref{lem:safety:simple decomposition} that prevent the existence of profitable in-segment deviations.
Nonetheless, if there existed a profitable in-segment deviation in a safety game, it yields another NE outcome with a strict superset of satisfied objectives, and we can reason on this preferable outcome.

In Büchi and co-Büchi games, strategy profiles based on similarly defined simple decompositions  are no longer sufficient to build NEs.
Furthermore, in contrast to safety games, profitable in-segment deviations can be harmful to players whose objectives were satisfied by the intended outcome.
We illustrate this below.

\begin{figure}
  \centering
  \begin{subfigure}[b]{0.48\textwidth}
    \centering
    \begin{tikzpicture}[node distance=0.8cm]
      \node[state, square, align=center] (v0) {$\vertex_0$};
      \node[state, circle, align=center, right = of v0] (v1) {$\vertex_1$};
      \node[state, square, align=center, right = of v1] (v2) {$\vertex_2$};
      \node[state, square, align=center, left = of v0] (v3) {$\vertex_3$};
      
      \path[->] (v0) edge[bend left] (v1);
      \path[->] (v1) edge[bend left] (v0);
      \path[->] (v1) edge (v2);
      \path[->] (v0) edge (v3);
      \path[->] (v3) edge[loop below] (v3);
      \path[->] (v2) edge[loop below] (v2);
    \end{tikzpicture}
    \caption{A naive decomposition-based approach fails to obtain an NE from $\vertex_0$ with outcome $\vertex_0\vertex_1\vertex_2^\omega$ in the game on the above arena where the objectives of $\playerOne$ and $\playerTwo$ are respectively $\buchi{\{\vertex_0, \vertex_1\}} = \cobuchi{\{\vertex_2, \vertex_3\}}$ and $\buchi{\{\vertex_2\}} = \cobuchi{\{v_0,\vertex_1, \vertex_3\}}$.}
    \label{fig:ex:buchi:one}
  \end{subfigure}
  \hfill
  \begin{subfigure}[b]{0.48\textwidth}
    \centering
    \begin{tikzpicture}[node distance=12mm]
      \node[state, square, align=center] (v0) {$\vertex_0$};
      \node[emptynode, right = of v0] (mid) {};
      \node[state, circle, align = center, above = 5mm of mid] (v1) {$\vertex_1$};
      \node[state, circle, align = center, below = 5mm of mid] (v2) {$\vertex_2$};
      \node[state, circle, align=center, right = of mid] (v3) {$\vertex_3$};
      \node[state, circle, align=center, right = of v3] (v4) {$\vertex_4$};
      
      \path[->] (v0) edge (v1);
      \path[->] (v0) edge (v2);
      \path[->] (v1) edge (v3);
      \path[->] (v2) edge (v3);
      \path[->] (v3) edge (v0);
      \path[->] (v3) edge (v4);
      \path[->] (v4) edge[loop below] (v4);
    \end{tikzpicture}
    \caption{To define an NE based on the trivial decomposition $((\vertex_0\vertex_1\vertex_3)^\omega)$ in the co-Büchi game on the above arena with objectives $\cobuchi{\{\vertex_2, \vertex_4\}}$ for $\playerOne$ and $\cobuchi{\{\vertex_1, \vertex_4\}}$ for $\playerTwo$, $\playerOne$ must (eventually) punish $\playerTwo$ in $\vertex_3$ if $\vertex_2$ is visited.}\label{fig:ex:buchi:two}
  \end{subfigure}
  \caption{Two arenas. Circle and squares respectively denote $\playerOne$ and $\playerTwo$ vertices.}
  \label{fig:ex:buchi}
\end{figure}

\begin{example}\label{example:buchi:one}
  Consider the game $\game$ on the arena depicted in Figure~\ref{fig:ex:buchi:one} where the objectives of $\playerOne$ and $\playerTwo$ are $\buchi{\{\vertex_0, \vertex_1\}} = \cobuchi{\{\vertex_2, \vertex_3\}}$ and $\buchi{\{\vertex_2\}} = \cobuchi{\{\vertex_0, \vertex_1, \vertex_3\}}$ respectively.
  The play $\play = \vertex_0\vertex_1\vertex_2^\omega$ is the outcome of an NE by Theorem~\ref{thm:charac:reach:buchi}.

  First, let us view $\game$ as a Büchi game.
  We consider the (infinite ultimately periodic) decomposition $\decomp = (\vertex_0\vertex_1\vertex_2, \vertex_2\vertex_2, \vertex_2\vertex_2, \ldots)$ of $\play$.
  The decomposition $\decomp$ is a natural candidate from which we can build a finite-memory NE from $\vertex_0$ with outcome $\play$ with the approaches used in Theorems~\ref{thm:reach:ne},~\ref{thm:short:ne} and~\ref{thm:safe:ne}: each segment connects successive visits to the target of $\playerTwo$, whose objective is satisfied by $\play$.
  However, there are no suitable NEs in which $\playerTwo$ follows a strategy based on $\decomp$.
  Indeed, if $\playerTwo$ uses such a strategy, then $\playerOne$ can enforce their objective with an in-segment deviation from $\vertex_1$ resulting in the outcome $(\vertex_0\vertex_1)^\omega\in\buchi{\{\vertex_1, \vertex_0\}}$, as $\playerTwo$ would not punish the deviation.

  Let us now view $\game$ as a co-Büchi game.
  In this case, a natural (with respect to our argument for safety games) decomposition candidate to construct an NE with outcome $\play$ is $(\vertex_0\vertex_1\vertex_2, \vertex_2^\omega)$, as $\playerOne$ wants to avoid $\vertex_2$ in the limit.
  The same issue as above arises.
  \hfill $\lhd$
\end{example}

The previous example highlights a shared limitation of decomposition-based strategies in Büchi and co-Büchi games.
To implement finite-memory NEs in Büchi and co-Büchi games in the following sections, we propose a two-phase mechanism.
First, we enforce a prefix of the intended outcome (i.e., we fully implement the punishment mechanism).
Second, we follow a decomposition-based approach for a suffix of the play.

We select the prefix and suffix such that no profitable in-segment deviations exist in the second phase.
For Büchi games, we start the second phase once there are no more targets of the \textit{losing players} appearing in the remainder of the intended outcome.
In co-Büchi games, we start the second phase once there are no more targets of the \textit{winning players} appearing in the suffix.
As in safety games, this construction for co-Büchi games does not prohibit the existence of profitable in-segment deviations by itself.
However, these deviations are not harmful to the players whose objectives are satisfied.
In practice, we will work with outcomes for which in-segment deviations are not profitable in the second phase (see Lemma~\ref{lem:cobuchi:outcomes}, the counterpart for co-Büchi games of Lemma~\ref{lem:safety:simple decomposition}).

In reachability games, the finite-memory NEs constructed in the proof of Theorem~\ref{thm:reach:ne} (see also Example~\ref{ex:reach ne}) are such that players stop punishing deviating players whenever the current segment of the intended outcome is reentered.
In particular, these strategies do not require additional memory states to remember that a player has left the current segment.
We illustrate, through the following example, that these additional memory states are needed to construct NEs in Büchi and co-Büchi games, i.e., players must commit to punishing whoever deviates in the first phase or leaves the current segment in the second phase of our approach.

\begin{example}
  Consider the game $\game$ on the arena depicted in Figure~\ref{fig:ex:buchi:two} where the objectives of $\playerOne$ and $\playerTwo$ are $\cobuchi{\{\vertex_2, \vertex_4\}}$ and $\cobuchi{\{\vertex_1, \vertex_4\}}$ respectively.
  The play $\play = (\vertex_0\vertex_1\vertex_3)^\omega$ is the outcome of an NE by Theorem~\ref{thm:charac:reach:buchi}.
  It satisfies the objective of $\playerOne$ but not the objective of $\playerTwo$.

  Let $\stratOne$ be the memoryless strategy $\stratOne$ of $\playerOne$ such that $\stratOne(\vertex_3) = \vertex_0$.
  The strategy $\stratOne$ is the only strategy based on the trivial decomposition $(\play)$ that tries to progress along $\play$ whenever the current vertex is in one of its vertices (i.e., $\stratOne$ is the only strategy implementing the punishment approach used in the proof of Theorem~\ref{thm:reach:ne}).
  We claim that there are no NEs from $\vertex_0$ whose outcome is $(\vertex_0\vertex_1\vertex_3)^\omega$ where the strategy of $\playerOne$ is $\stratOne$.
  Indeed, $\playerTwo$ would have a profitable deviation by moving to $\vertex_2$ whenever $\vertex_0$ is reached.

  We can apply the same reasoning on the Büchi game on the arena depicted in Figure~\ref{fig:ex:buchi:two} with the objectives $\buchi{\{\vertex_1\}}$ for $\playerOne$ and $\buchi{\{\vertex_2\}}$ for $\playerTwo$ and the same outcome to show that players must also remember whether a player has left the current segment of the intended outcome for our construction to work.
  \hfill $\lhd$
\end{example}

Finally, we observe that the two-phase approach described above yields finite-memory strategies for which we cannot provide arena-independent memory bounds as we encode a simple history in the memory.
The following result implies that we cannot provide arena-independent memory bounds for constrained NEs in Büchi and co-Büchi games in general.

\newcommand{\indexLocal}{p}
\newcommand{\indexVertex}{q}
\newcommand{\indexPower}{a}
\newcommand{\indexMod}{b}

\begin{figure}
  \centering
  \begin{tikzpicture}[node distance=7mm]
    \node[state, square,align=center] (v1) {$v_1$};
    \node[state, circle, align=center, right = of v1] (w1) {$w_1$};
    \node[state, square, align=center, right = of w1] (v2) {$v_2$};
    \node[state, circle, align=center, right = of v2] (w2) {$w_2$};
    \node[state, square, align=center, right = of w2] (v3) {$v_3$};
    \node[state, circle, align=center, right = of v3] (w3) {$w_3$};
    \node[state, square, align=center, right = of w3] (t) {$v_4$};
    \node[state, square, align=center, below = of v2] (bot) {$v_5$};
    \path[->] (v1) edge (w1);
    \path[->] (v2) edge (w2);
    \path[->] (v3) edge (w3);
    \path[->] (w1) edge[bend right] (v1);
    \path[->] (w1) edge (v2);
    \path[->] (w2) edge[bend right] (v1);
    \path[->] (w2) edge[bend right] (v2);
    \path[->] (w2) edge (v3);
    \path[->] (w3) edge[bend right] (v1);
    \path[->] (w3) edge[bend right] (v2);
    \path[->] (w3) edge[bend right] (v3);
    \path[->] (w3) edge (t);
    \path[->] (v1) edge[bend right] (bot);
    \path[->] (v2) edge (bot);
    \path[->] (v3) edge[bend left] (bot);
    \path[->] (t) edge[loop below] (t);
    \path[->] (bot) edge[loop below] (bot);
  \end{tikzpicture}
  \caption{The arena used in the proof of Proposition~\ref{prop:buchi:dep} for $\indexLocal = 3$.
    In the game where the objective of $\playerOne$ and $\playerTwo$ are respectively $\buchi{\vertexSet\setminus\{\vertex_4, \vertex_5\}} = \cobuchi{\{\vertex_4, \vertex_5\}}$ and $\buchi{\{\vertex_4\}} = \cobuchi{\vertexSet\setminus\{\vertex_4\}}$ (where $\vertexSet$ is the set of vertices), $\playerTwo$ needs a memory of size at least $3$ in any NE from $\vertex_1$ ensuring the objective of $\playerTwo$.
    Circles and squares are respectively $\playerOne$ and $\playerTwo$ vertices.}
  \label{fig:prop:buchi:dep}
\end{figure}

\begin{proposition}\label{prop:buchi:dep}
  Let $\indexLocal\in\IN\setminus\{0\}$.
  There exist a two-player game $\game=(\arena, (\objective_1, \objective_2))$ where $\arena$ is an arena with $2\indexLocal+2$ vertices and $\objective_1$ and $\objective_2$ are both Büchi and co-Büchi objectives, and an initial vertex $\vertex$ in $\game$ such that $\playerTwo$ requires a memory of size at least $\indexLocal$ in any NE from $\vertex$ whose outcome satisfies $\objective_2$.
\end{proposition}
\begin{proof}
  We define $\arena = ((\vertexSetOne, \vertexSetTwo), \edgeSet)$ as follows.
  We let $\vertexSetOne = \{\vertexAlt_1, \ldots, \vertexAlt_\indexLocal\}$ and $\vertexSetTwo = \{\vertex_1, \ldots, \vertex_{\indexLocal+2}\}$.
  The set of edges $\edgeSet$ is defined by the three following rules.
  Let $\indexVertex\in\integerInterval{\indexLocal}$.
  From $\vertex_\indexVertex$, there is an edge to $\vertex_{\indexLocal+2}$ and an edge to $\vertexAlt_{\indexVertex}$.
  From $\vertexAlt_\indexVertex$, there is an edge to $\vertex_{\indexVertex'}$ for all $\indexVertex'\in\integerInterval{\indexVertex+1}$.
  Finally, there are self-loops in $\vertex_{\indexLocal+1}$ and $\vertex_{\indexLocal+2}$.
  We illustrate $\arena$ for $\indexLocal = 3$ in Figure~\ref{fig:prop:buchi:dep}, and remark that the arena for $\indexLocal = 1$ matches the arena of Figure~\ref{fig:ex:buchi:one} up to a renaming of the vertices.
  We consider the objectives $\objective_1 = \buchi{\{\vertexSet\setminus\{\vertex_{\indexLast+1}, \vertex_{\indexLast+2}\}} = \cobuchi{\{\vertex_{\indexLast+1}, \vertex_{\indexLast+2}\}}$ for $\playerOne$ and $\objective_2 = \buchi{\{\vertex_{\indexLast+1}\}} = \cobuchi{\{\vertexSet\setminus\{\vertex_{\indexLast+2}\}}$ for $\playerTwo$, and let $\game = (\arena, (\objective_1, \objective_2))$.

  We show the following two claims in the Büchi and co-Büchi game $\game$:
  \begin{enumerate}[(i)]
  \item there exists a finite-memory NE from $\vertex_1$ such that the objective of $\playerTwo$ is satisfied where $\playerOne$ has a memoryless strategy and $\playerTwo$ has a memory of size $\indexLocal+1$; and\label{item:example:buchi:dependence:one}
  \item any strategy profile in which $\playerTwo$ has a finite-memory strategy of size at most $\indexLocal$ and whose outcome from $\vertex_1$ satisfies $\objective_2$ is not an NE from $\vertex_1$.\label{item:example:buchi:dependence:two}
  \end{enumerate}

  We first prove~\ref{item:example:buchi:dependence:one}.
  For $\playerOne$, we consider the memoryless strategy $\stratOne$ such that for all $\indexVertex\in\integerInterval{\indexLocal}$, $\stratOne(\vertexAlt_\indexVertex) = \vertex_{\indexVertex+1}$.
  For $\playerTwo$, we consider the strategy $\stratTwo$ induced by the Mealy machine $\mealyMachine = \mealyTuplePure$ defined as follows.
  We let $\mealyStateSpace = \integerInterval{\indexLocal+1}$ and $\mealyStateInit = 1$.
  For all $\indexVertex\in\mealyStateSpace$, we let $\mealyUpdate(\indexVertex, \vertex_{\indexVertex}) = \indexVertex+1$ and, for all $\vertex\in\vertexSet\setminus\{\vertex_\indexVertex\}$, $\mealyUpdate(\indexVertex, \vertex) = \indexVertex$.
  For all $\indexVertex\in\mealyStateSpace$, we let $\mealyNext(\indexVertex, \vertex_{\indexVertex}) = \vertexAlt_{\indexVertex}$, $\mealyNext(\indexVertex, \vertex_{\indexLocal+1}) = \vertex_{\indexLocal+1}$, $\mealyNext(\indexVertex, \vertex_{\indexLocal+2}) = \vertex_{\indexLocal+2}$ and, for all $\vertex\in\vertexSetTwo\setminus\{\vertex_\indexVertex, \vertex_{\indexLocal+1}, \vertex_{\indexLocal+2}\}$, we let $\mealyNext(\indexVertex, \vertex) = \vertex_{\indexLocal+2}$.
  Intuitively, $\stratTwo$ moves rightward (with respect to Figure~\ref{fig:prop:buchi:dep}) at each step as long as $\playerOne$ does so, and moves to $\vertex_{\indexLocal+2}$ as soon as $\playerOne$ moves to the left.
  
  By construction, $\outcome{(\stratOne,\stratTwo)}{\vertex_1} = \vertex_1\vertexAlt_1\ldots \vertex_\indexLocal\vertexAlt_\indexLocal \vertex_{\indexLocal+1}^\omega$ satisfies $\objective_2$.
  To show that $(\stratOne, \stratTwo)$ is an NE, it suffices to prove that $\playerOne$ does not have a profitable deviation.
  Let $\stratAltOne$ be an arbitrary strategy of $\playerOne$.
  We consider two cases.
  First, assume that for all $\indexVertex\in\integerInterval{\indexLocal}$, for $\history_\indexVertex = \vertex_1\vertexAlt_1\ldots\vertex_\indexVertex\vertexAlt_\indexVertex$, we have $\stratAltOne(\history_\indexVertex) = \vertex_{\indexVertex+1}$.
  In this case, the outcome from $\vertex_1$ of $(\stratAltOne, \stratTwo)$ matches that of $(\stratOne, \stratTwo)$, hence $\stratAltOne$ is not a profitable deviation of $\playerOne$.
  Second, assume that for some $\indexVertex\in\integerInterval{\indexLocal}$, we have $\stratAltOne(\history_\indexVertex) = \vertex_{\indexVertex'}$ with $\indexVertex'\neq\indexVertex+1$.
  We consider the smallest such $\indexVertex$.
  It follows from a straightforward induction that $\widehat{\mealyUpdate}(\history_\indexVertex) = \indexVertex+1$.
  By definition of $\mealyMachine$, we obtain that the outcome of $\stratAltOne$ and $\stratTwo$ from $\vertex_1$ is the play $\history_\indexVertex\vertex_{\indexVertex'}\vertex_{\indexLocal+2}^\omega$, hence $\stratAltOne$ is not a profitable deviation in this case either.
  This shows that $(\stratOne, \stratTwo)$ is an NE from $\vertex_1$.
  This ends the proof of~\ref{item:example:buchi:dependence:one}.

  We now show~\ref{item:example:buchi:dependence:two}.
  Let $\stratProfile = (\stratOne, \stratTwo)$ be a strategy profile such that its outcome from $\vertex_1$ is winning for $\playerTwo$ and such that $\stratTwo$ is given by a Mealy machine $\mealyMachine = \mealyTuplePure$ with memory size at most $\indexLocal$.
  We prove that $\playerOne$ has a profitable deviation from $\vertex_1$.
  
  The structure of $\arena$ ensures that all vertices of $\arena$ besides $\vertex_{\indexLocal+2}$ occur in $\outcome{\stratProfile}{\vertex_1}$.
  For all $\indexVertex\in\integerInterval{\indexLocal}$, we let $\history_\indexVertex$ be the shortest prefix of $\outcome{\stratProfile}{\vertex_1}$ that ends in $\vertexAlt_\indexVertex$, and let $\history_0$ denote the empty word.
  It follows from our assumption on the memory size of $\mealyMachine$ that there exist $\indexVertex<\indexVertex'\in\integerInterval{\indexLocal}\cup\{0\}$ such that $\widehat{\mealyUpdate}(\history_\indexVertex) = \widehat{\mealyUpdate}(\history_{\indexVertex'})$.
  Let $\history = \vertexAltAlt_0\ldots\vertexAltAlt_\indexLast$ be the non-empty history such that $\history_{\indexVertex'} = \history_\indexVertex\history$.
  We prove that $\history_{\indexVertex}\history^\omega$ is consistent with $\stratTwo$, which implies the existence of a profitable deviation for $\playerOne$.

  The consistency argument is by induction.
  We show that for all $\indexPosition\in\IN$, letting $\indexPower=\left\lfloor\frac{\indexPosition}{\indexLast+1}\right\rfloor$ and $\indexMod=\indexPosition\bmod(\indexLast+1)$, the history $\hist_\indexVertex\hist^{\indexPower}\vertexAltAlt_0\ldots\vertexAltAlt_{\indexMod}$ is consistent with $\stratTwo$ and $\widehat{\mealyUpdate}(\hist_\indexVertex\hist^{\indexPower}\vertexAltAlt_0\ldots\vertexAltAlt_{\indexMod}) = \widehat{\mealyUpdate}(\hist_\indexVertex\vertexAltAlt_0\ldots\vertexAltAlt_{\indexMod})$.
  For $0\leq\indexPosition\leq\indexLast$, the former property follows from $\history_{\indexVertex}\history$ being consistent with $\stratTwo$ (it is a prefix of $\outcome{\stratProfile}{\vertex_1}$) and the latter property is trivially true (because $\indexPower=0$).
  We now assume that the previous property holds for $\indexPosition\in\IN$ and show that it holds for $\indexPosition+1$.
  We let $\indexPower=\left\lfloor\frac{\indexPosition}{\indexLast+1}\right\rfloor$ and $\indexMod=\indexPosition\bmod(\indexLast+1)$.
  To avoid having to distinguish the case in which $\indexPosition+1$ is divisible by $\indexLast+1$ from the other case, we (abusively) let $\vertexAltAlt_{\indexLast+1}$ denote $\vertexAltAlt_0$.

  Regarding memory states, we have that $\widehat{\mealyUpdate}(\history_\indexVertex\history^\indexPower\vertexAltAlt_0\ldots\vertexAltAlt_\indexMod\vertexAltAlt_{\indexMod+1}) = \widehat{\mealyUpdate}(\history_\indexVertex\vertexAltAlt_0\ldots\vertexAltAlt_\indexMod\vertexAltAlt_{\indexMod+1})$ directly from the induction hypothesis $\widehat{\mealyUpdate}(\history_\indexVertex\history^\indexPower\vertexAltAlt_0\ldots\vertexAltAlt_\indexMod) = \widehat{\mealyUpdate}(\history_\indexVertex\vertexAltAlt_0\ldots\vertexAltAlt_\indexMod)$ and the definition of $\widehat{\mealyUpdate}$.
  We consider two cases for the consistency.
  First, we assume that $\vertexAltAlt_{\indexMod}\in\vertexSetOne$.
  The consistency of $\hist_\indexVertex\hist^{\indexPower}\vertexAltAlt_0\ldots\vertexAltAlt_{\indexMod}\vertexAltAlt_{\indexMod+1}$ with $\stratTwo$ follows directly by induction as $\playerTwo$ does not select the last transition.
  Second, we assume that $\vertexAltAlt_{\indexMod}\in\vertexSetTwo$.
  By induction, $\history_\indexVertex\history^\indexPower\vertexAltAlt_0\ldots\vertexAltAlt_{\indexMod}$ is consistent with $\stratTwo$ and thus it remains to prove that $\stratTwo(\history_\indexVertex\history^\indexPower\vertexAltAlt_0\ldots\vertexAltAlt_{\indexMod}) =\vertexAltAlt_{\indexMod+1}$.
  The induction hypothesis on memory updates and the consistency of $\hist_\indexVertex\hist$ with $\stratTwo$ imply that $\stratTwo(\history_\indexVertex\hist^\indexPower\vertexAltAlt_0\ldots\vertexAltAlt_{\indexMod})= \stratTwo(\history_\indexVertex\vertexAltAlt_0\ldots\vertexAltAlt_{\indexMod}) =\vertexAltAlt_{\indexMod+1}$.

  We have shown that in the game $\game$, any NE from $\vertex_1$ with an outcome that is winning for $\playerTwo$ requires $\playerTwo$ to have a memory size of at least $\indexLocal = \frac{|\vertexSet|-2}{2}$.
  This shows that arena-dependent memory is required for general constrained NEs in Büchi and co-Büchi games.
\end{proof}

\let\indexLocal\undefined
\let\indexPower\undefined
\let\indexVertex\undefined

\subsection{Nash equilibria in Büchi games}\label{section:buchi:buchi-ne}

In this section, we show that from any NE in a Büchi game, we can derive a finite-memory NE from the same initial vertex with the same set of satisfied objectives.
As before, we build finite-memory NEs from well-chosen decompositions of NE outcomes obtained with a simplification process.
We devise two variants of the two-phase approach described in the previous section.
In Section~\ref{section:buchi:ne:infinitely occuring}, we provide a construction for NE outcomes in which some vertex occurs infinitely often.
We provide a variant for other NE outcomes in Section~\ref{section:buchi:ne:finitely occuring}.

We let $\game = (\arena, (\buchi{\target_\indexPlayer})_{\indexPlayer\in\playerSet})$ for the remainder of the section.
We assume without loss of generality that for any considered NE outcome $\play$ in the following, $\satpl{\play}$ is not empty.
This can be ensured by adding a new player for whom all vertices are targets if necessary.

\subsubsection{Outcomes with an infinitely occurring vertex}\label{section:buchi:ne:infinitely occuring}

We study NE outcomes in which a vertex occurs infinitely often.
This case is a generalisation of the finite-arena case; in a finite arena, all plays contain some infinitely occurring vertex.
We first show that from any NE outcome, we can derive an NE outcome satisfying the same objectives of $\game$ that has an ultimately periodic segment decomposition containing only simple cycles and simple histories.
Second, we build finite-memory NEs from these outcomes and their decompositions.

\subparagraph*{Well-structured NE outcomes.}
Let $\play'$ be an NE outcome in $\game$ such that some vertex occurs infinitely often in $\play'$.
In contrast to reachability, shortest-path and safety games, we do not derive an NE outcome that has a finite simple segment decomposition from $\play'$.
This would be too restrictive: plays that admit finite simple segment decompositions can satisfy the objectives of two players in a Büchi game only if their targets lie on the same simple cycle.
Instead, we derive a lasso NE outcome $\play$ from $\play'$ with the same initial vertex such that $\satpl{\play}=\satpl{\play'}$.
We construct $\play$ such that it admits an infinite decomposition $(\segment_0, \segment_1, \ldots)$ that satisfies three key properties that are used to implement the two-phase mechanism described in the previous section.

The prefix of $\play$ that the players enforce in the first phase will be $\segment_0$.
To ensure that no profitable deviations exist within a segment in the second phase, we ensure that for all $\indexSegment\geq 1$, no targets of players whose objective is not satisfied by $\play$ occur in $\segment_\indexSegment$.
For the second phase, the players use a strategy that is based on the decomposition $(\segment_1, \segment_2, \ldots)$.
To construct a finite-memory strategy, we require that $(\segment_1, \segment_2, \ldots)$ is a periodic sequence with a period $\numSegments\in\integerInterval{\nPlayer}$ (thus $(\segment_0, \segment_1, \ldots)$ is ultimately periodic).
Furthermore, all segments must be either simple histories or simple cycles (a decomposition-based strategy need not be well-defined otherwise).
Therefore, we require that either the period is one and $\segment_1$ is a simple cycle, or its period is greater and all segments $\segment_1$, $\segment_2$, \ldots are non-trivial simple histories.

The general idea for the derivation of $\play$ is as follows.
Let $\indexPosition\in\IN$ such that no targets of losing players occur in $\suffix{\play'}{\indexPosition}$ and  fix an infinitely occurring vertex $\vertex^\star$ of $\play'$.
For all $\indexPlayer\in \satpl{\play'}$, fix $t_\indexPlayer\in\target_\indexPlayer$ occurring after $\vertex^\star$ in $\suffix{\play'}{\indexPosition}$.
For all $\indexPlayer\in\satpl{\play'}$, there are histories from $\vertex^\star$ to $t_\indexPlayer$ and from $t_\indexPlayer$ to $\vertex^\star$ in which only vertices of $\suffix{\play'}{\indexPosition}$ appear.
It follows that, for all $\indexPlayer, \indexPlayer'\in\satpl{\play'}$, the vertices $t_\indexPlayer$ and $t_{\indexPlayer'}$ are connected by simple histories or simple cycles consisting only of vertices from $\suffix{\play'}{\indexPosition}$.
We construct the sought decomposition by letting $\segment_0$ be a simple history from $\vertex_0$  to some $t_\indexPlayer$ using vertices occurring in $\play'$, and then selecting the other segments among the aforementioned simple histories or cycles to obtain a periodic sequence (with period of at most $\nPlayer$) such that all $t_\indexPlayer$ occur in the combination of its histories.
We formalise this construction below.

\begin{restatable}{lemma}{lemBuchiSimpleLasso}\label{lem:buchi:simplelasso}
  Let $\play'$ be the outcome of an NE from $\vertex_0\in\vertexSet$ in the Büchi game $\game$ such that some vertex occurs infinitely often in $\play'$.
  There exists an NE outcome $\play$ from $\vertex_0$ with $\satpl{\play}=\satpl{\play'}$ such that
  $\play$ admits an infinite simple segment decomposition $(\segment_0, \segment_1, \ldots)$ satisfying
  \begin{enumerate}[(i)]
  \item for all $\indexSegment\geq 1$ and all $\indexPlayer\in\playerSet\setminus\satpl{\play}$, no vertex of $\target_\indexPlayer$ occurs in $\segment_\indexSegment$,\label{item:buchi:simplelasso:one}
  \item there exists $\numSegments\in\integerInterval{|\satpl{\play'}|}$ such that for all $\indexSegment\geq 1$, $\segment_\indexSegment=\segment_{\indexSegment+\numSegments}$; and\label{item:buchi:simplelasso:two}
  \item either $\numSegments = 1$ and $\segment_1$ is a simple cycle, or $\numSegments > 1$ and $\segment_\indexSegment$ is a non-trivial simple history for all $\indexSegment\in\integerInterval{\numSegments}$.\label{item:buchi:simplelasso:three}
  \end{enumerate}
\end{restatable}
\begin{proof}
  We let $\indexPosition\in\IN$ such that for all $\indexPlayer\in\playerSet\setminus\satpl{\play}$, no vertices of $\target_\indexPlayer$ occur in $\suffix{\play}{\indexPosition}$.
  For all $\indexPlayer\in\satpl{\play}$, we fix $\vertexAlt_\indexPlayer\in\target_\indexPlayer$ which appears in $\suffix{\play}{\indexPosition}$ later than an infinitely occurring vertex of $\suffix{\play}{\indexPosition}$.
  We let $t_1,\ldots, t_\numSegments$ denote the distinct elements of $\{\vertexAlt_\indexPlayer\mid\indexPlayer\in\satpl{\play'}\}$ (recall our assumption that $\satpl{\play}$ is non-empty).
  We distinguish two cases.

  First, assume that $\numSegments = 1$.
  We let $\segment_0$ be a simple history from $\vertex_0$ to $t_1$ in which only states of $\play$ occur ($\segment_0$ is a trivial history if $\vertex_0 = t_1$), and let $\segment_1$ be a simple cycle from $t_1$ to itself using only vertices occurring in $\suffix{\play}{\indexPosition}$.
  Let $\play$ be the play admitting the decomposition $(\segment_0, \segment_1, \segment_1, \ldots)$.
  By construction, we have $\satpl{\play} = \satpl{\play'}$, and Theorem~\ref{thm:charac:reach:buchi} implies that $\play$ is an NE outcome.
  Furthermore, $\play$ and its decomposition $(\segment_0, \segment_1, \segment_1, \ldots)$ satisfy~\ref{item:buchi:simplelasso:one},~\ref{item:buchi:simplelasso:two} and~\ref{item:buchi:simplelasso:three}.

  We now assume that $\numSegments\geq 2$.
  Once more, we define an NE outcome from a decomposition.
  We let $\segment_0$ be a (possibly trivial) simple history from $\vertex_0$ to $t_1$ using only vertices from $\play$.
  For all $\indexSegment\in\integerInterval{\numSegments-1}$, we let $\segment_\indexSegment$ be a simple history from $t_{\indexSegment}$ to $t_{\indexSegment+1}$.
  Similarly, we let $\segment_\numSegments$ be a simple non-trivial history from $t_{\numSegments}$ to $t_1$.
  Other segments are defined so~\ref{item:buchi:simplelasso:two} holds.
  Let $\decomp =  (\segment_0, \segment_1, \segment_2, \ldots)$.
  The play $\play$ admitting the decomposition $\decomp$ and $\decomp$ satisfy~\ref{item:buchi:simplelasso:one} and~\ref{item:buchi:simplelasso:three}.
  Furthermore, Theorem~\ref{thm:charac:reach:buchi} ensures that $\play$ is the outcome of an NE.  
\end{proof}

\subparagraph*{Finite-memory Nash equilibria.}
We now construct finite-memory NEs from the NE outcomes given by Lemma~\ref{lem:buchi:simplelasso}.
Assume that Lemma~\ref{lem:buchi:simplelasso} provides an NE outcome $\play$ from $\vertex_0\in\vertexSet$ and a decomposition $\decomp = (\segment_0, \segment_1, \ldots)$ of $\play$.
We implement the two-phase mechanism described in Section~\ref{section:buchi:ex}.
For the first phase, the players enforce the history $\segment_0$.
For the second phase, we switch to a strategy that is based on the periodic decomposition $\decomp' = (\segment_1, \segment_2, \ldots)$.
Finite memory suffices for the second phase thanks to the periodicity of $\decomp'$: we adapt the updates given in Definition~\ref{def:compatible mealy} such that when reading $\last{\segment_\numSegments}$ in memory states of the form $(\playerI, \numSegments)$, we update the memory to an appropriate memory state of the form $(\playerIAlt, 1)$.

By completing this behaviour such that players switch to memoryless punishing strategies if $\segment_0$ is not accurately simulated or if a player exits the current segment, we obtain a finite-memory NE.
The stability of the NE follows from Theorem~\ref{thm:charac:reach:buchi} for deviations that induce the use of punishing strategies and the property that no targets of losing players occur in segments $\segment_\indexSegment$, $\indexSegment\geq 1$, for in-segment deviations.
We formally present these finite-memory strategies in the proof of the following theorem.

\begin{restatable}{theorem}{thmBuchiNELasso}\label{thm:buchi:fmne:lasso}
  Let $\stratProfile'$ be an NE from a vertex $\vertex_0$ such that a vertex occurs infinitely often in its outcome.
  There exists a finite-memory NE $\stratProfile$ from $\vertex_0$ such that $\satpl{\outcome{\stratProfile}{\vertex_0}}=\satpl{\outcome{\stratProfile'}{\vertex_0}}$.
  If $\arena$ is finite, a memory size of at most $|\vertexSet| + \nPlayer^2+\nPlayer$ suffices.
\end{restatable}
\begin{proof}
  Let $\play$ be an NE outcome obtained via  Lemma~\ref{lem:buchi:simplelasso} from $\outcome{\stratProfile'}{\vertex_0}$.
  We let $(\segment_0, \segment_1, \ldots)$ be the decomposition provided by the lemma, and let $\numSegments\leq|\satpl{\play}|$ be the least period of $(\segment_1, \segment_2, \segment_3, \ldots)$.
  We show that $\play$ is the outcome of a finite-memory NE.
  We first introduce some notation.
  Let $\playerSubset\subseteq\playerSet$ be $\playerSet\setminus\satpl{\play}$ if this set is not empty, or $\{1\}$ otherwise.
  For all $\indexPlayer\in\playerSubset$, we let $\stratAltAdvI$ be a memoryless uniformly winning strategy for the second player of the coalition game $\game_\indexPlayer = (\arena_\indexPlayer, \buchi{\target_\indexPlayer})$ (it exists by Theorem~\ref{thm:qualitative:ML strat}) and let $\winningIAdv{\cobuchi{\target_\indexPlayer}}$ denote the winning region of this player in $\game_\indexPlayer$.
  We write $\segment_0 = \vertex_0\ldots\vertex_\indexLast$.
  Finally, we define $\vertexSet_\playerSubset = \bigcup_{\indexPlayer\in\playerSubset}\vertexSetI$.

  For each $\indexPlayer\in\playerSet$, we define a Mealy machine $\mealyMachine_\indexPlayer= (\mealyStateSpace, \mealyStateInit, \mealyUpdate, \mealyNextI)$ as follows.
  We define $\mealyStateSpace$ as the set $\{\vertex_\indexPosition\in\vertexSet\mid 0\leq\indexPosition \leq \indexLast\}\cup (\{\playerI\mid\indexPlayer\in\playerSubset\}\times \integerInterval{\numSegments})\cup \{\playerI\mid\indexPlayer\in\playerSubset\}$ and let $\mealyStateInit=\vertex_0$.
  The memory states that are vertices correspond to the first phase of the informal explanation preceding the theorem, and the others to the second phase.
  We note that the memory bounds claimed for finite arenas follow from the simplicity of $\segment_0$ and $\numSegments\leq\nPlayer$.

  The update function $\mealyUpdate$ is defined as follows.
  \begin{enumerate}[(a)]
  \item For all $\indexPosition < \indexLast$, we let $\mealyUpdate(\vertex_\indexPosition, \vertex_\indexPosition) = \vertex_{\indexPosition+1}$.\label{enum:buchi-update:a}
  \item  We let $\mealyUpdate(\vertex_\indexLast, \vertex_\indexLast) = (\playerI, 1)$ where $\indexPlayer\in\playerSubset$ is such that $\playerI$ controls $\vertex_\indexLast$ if $\vertex_\indexLast\in\vertexSet_\playerSubset$ and $\indexPlayer$ is arbitrary otherwise.\label{enum:buchi-update:b}
  \item For all $\indexPosition \leq \indexLast$ and $\vertex\neq\vertex_\indexPosition$,\label{enum:buchi-update:c}
    \begin{itemize}
    \item if $\indexPosition\geq 1$ and there exists $\indexPlayer\in\playerSubset$ such that $\vertex_{\indexPosition-1}\in\vertexSet_{\indexPlayer}$ (in particular, $\vertex\in\vertexSet_\playerSubset$), we let $\mealyUpdate(\vertex_\indexPosition, \vertex) = \playerI$;
    \item otherwise, $\mealyUpdate(\vertex_\indexPosition, \vertex)$ is left arbitrary.
    \end{itemize}
  \item For all states of the form $(\playerI, \indexSegment)\in\mealyStateSpace$ and all $\vertex\in\vertexSet$ occurring in $\segment_\indexSegment$, we let $\mealyUpdate((\playerI, \indexSegment), \vertex) = (\playerIAlt, \indexSegment')$ where
    \begin{itemize}
    \item $\playerIAlt$ is the player controlling $\vertex$ if $\vertex\in\vertexSet_\playerSubset$ and otherwise, $\indexPlayer'=\indexPlayer$, and 
    \item $\indexSegment'=\indexSegment$ if $\vertex\neq\last{\segment_\indexSegment}$, and otherwise, if $\vertex=\last{\segment_\indexSegment}$, we set $\indexSegment'=\indexSegment+1$ if $\indexSegment < \numSegments$ and $\indexSegment'=1$ otherwise.
    \end{itemize}
  \item For all states of the form $(\playerI, \indexSegment)\in\mealyStateSpace$ and all $\vertex\in\vertexSet$ that do not occur in $\segment_\indexSegment$, we let $\mealyUpdate((\playerI, \indexSegment), \vertex)=\playerI$.
  \item Finally, for all $\indexPlayer\in\playerSubset$ and $\vertex\in\vertexSet$, we let $\mealyUpdate(\playerI, \vertex) = \playerI$.
  \end{enumerate}
  Points \ref{enum:buchi-update:a},~\ref{enum:buchi-update:b} and~\ref{enum:buchi-update:c} relate to the first phase of our construction: the Mealy machine checks that the current vertex matches the one it should be while following $\segment_0$ and switches to a special punishment state if a deviation is detected.

  Let $\indexPlayer\in\playerSubset$. We now define $\mealyNextI$.
  Let $\vertex\in\vertexSetI$.
  \begin{enumerate}[(a)]
  \item We first consider memory states of the form $\vertex_\indexPosition$.
    Fix $\indexPosition\leq\indexLast$.
    \begin{itemize}
    \item If $\vertex = \vertex_\indexPosition$ and $\indexPosition\neq\indexLast$, we let $\mealyNextI(\vertex_\indexPosition, \vertex)=\vertex_{\indexPosition+1}$.
    \item If $\vertex=\vertex_\indexPosition$ and $\indexPosition=\indexLast$, we let $\mealyNextI(\vertex_\indexLast, \vertex_\indexLast)$ be the second vertex of $\segment_1$ (i.e., the only vertex that follows $\vertex_\indexLast$ in $\segment_1$).
    \item We distinguish two cases whenever $\vertex\neq\vertex_\indexPosition$:
      \begin{itemize}
      \item if $\indexPosition\geq 1$ and there exists $\indexPlayer'\in\playerSubset$ such that $\vertex_{\indexPosition-1}\in\vertexSet_{\indexPlayer'}$ and $\indexPlayer'\neq\indexPlayer$, we let $\mealyNextI(\vertex_\indexPosition, \vertex)=\stratAltAdvIb(\vertex)$;
      \item we let $\mealyNextI(\vertex_\indexPosition, \vertex)$ be arbitrary in all other cases.
      \end{itemize}
    \end{itemize}
  \item We now deal with memory states of the form $(\playerIAlt, \indexSegment)$.
    Fix $\indexPlayer'\in\playerSubset$ and $\indexSegment\in\integerInterval{\numSegments}$.
    \begin{itemize}
    \item If $\vertex$ occurs in $\segment_\indexSegment$ and $\vertex\neq\last{\segment_\indexSegment}$, we let $\mealyNextI((\playerIAlt, \indexSegment), \vertex)$ be the vertex following $\vertex$ in $\segment_\indexSegment$.
    \item If $\vertex=\last{\segment_\indexSegment}$ and $\indexSegment < \numSegments$ (resp.~$\indexSegment=\numSegments$), we let $\mealyNextI((\playerIAlt, \indexSegment), \vertex)$ be the second vertex of $\segment_{\indexSegment+1}$ (resp.~$\segment_1$).
    \item If $\vertex$ does not occur in $\segment_\indexSegment$ and $\indexPlayer\neq\indexPlayer'$, we let $\mealyNextI((\playerIAlt, \indexSegment), \vertex) = \stratAltAdvIb(\vertex)$.
    \item If $\vertex$ does not occur in $\segment_\indexSegment$ and $\indexPlayer=\indexPlayer'$, we let $\mealyNextI((\playerIAlt, \indexSegment), \vertex)$ be arbitrary.
    \end{itemize}
    \item Finally, we let $\mealyNextI(\playerIAlt, \vertex) = \stratAltAdvIb(\vertex)$ if $\indexPlayer'\neq\indexPlayer$, and otherwise we let $\mealyNextI(\playerIAlt, \vertex)$ be arbitrary.
\end{enumerate}

  We let $\stratI$ be the strategy induced by $\mealyMachine_\indexPlayer$.
  It can be shown that the outcome of $\stratProfile=(\stratI)_{\indexPlayer\in\playerSet}$ is $\play$.
  We observe that $\widehat{\mealyUpdate}(\segment_0)$ is of the form $(\playerI, 1)$.  
  If $\numSegments > 1$, all segments in $(\segment_1, \segment_2, \ldots)$ are non-trivial histories, and we can adapt the argument for coherence from the proof of Lemma~\ref{lemma:coherence:template:extensions} to show that $\outcome{\stratProfile}{\vertex_0} = \play$.
  If $\numSegments = 1$, we can adapt this argument similarly, as it amounts to studying a strategy based on a trivial segment decomposition of a simple lasso.
  
  It remains to show that $\stratProfile$ is an NE from $\vertex_0$.
  It suffices to show that for all $\indexPlayer\in\playerSet\setminus\satpl{\play}$, $\playerI$ does not have a profitable deviation.
  We fix one such $\indexPlayer$.
  We recall that by Theorem~\ref{thm:charac:reach:buchi}, all vertices in $\play$ are elements of $\winningIAdv{\cobuchi{\target_\indexPlayer}}$.
  
  Let $\play'=\vertex_0'\vertex_1'\ldots$ be a play starting in $\vertex_0$ that is consistent with the strategy profile $\stratIAdv$.
  We consider three cases.
  First, assume that $\play'$ does not have $\segment_0$ as a prefix.
  Let $\indexPosition < \indexLast$ be such that $\prefix{\play'}{\indexPosition}$ is the longest common prefix of $\play'$ and $\segment_0$.
  We have that $\widehat{\mealyUpdate}(\prefix{\play'}{\indexPosition}) = \vertex_{\indexPosition+1}$.
  The definition of $\stratProfile$ and the relation $\vertex_{\indexPosition+1}'\neq \vertex_{\indexPosition+1}$ imply that $\vertex_\indexPosition\in\vertexSetI$.
  It follows that $\suffix{\play'}{\indexPosition}$ is a play consistent with $\stratAltAdvI$ starting in $\vertex_\indexPosition\in\winningIAdv{\cobuchi{\target_\indexPlayer}}$, thus $\suffix{\play'}{\indexPosition}\in\cobuchi{\target_\indexPlayer}$.
  We obtain that $\play'\in\cobuchi{\target_\indexPlayer}$, ending this first case.

  Second, assume that $\play'$ has $\segment_0$ as a prefix, and that for all vertices $\vertex$ occurring in $\suffix{\play'}{\indexLast}$, there is some $\indexSegment\geq 1$ such that $\vertex$ appears in some $\segment_\indexSegment$.
  Because there are no elements of $\target_\indexPlayer$ in these segments, it follows that $\play'\in\cobuchi{\target_\indexPlayer}$.

  Finally, assume that $\play'$ has $\segment_0$ as a prefix and some vertex appearing in $\suffix{\play'}{\indexLast}$ does not occur in any of the segments $\segment_\indexSegment$ with $\indexSegment\geq 1$.
  It follows that the memory state of the players relying on $\mealyMachine_\indexPlayer$ eventually becomes of the form $\playerIAlt$.
  Let $\indexPosition\in\IN$ be the largest number such that $\widehat{\mealyUpdate}(\prefix{\play'}{\indexPosition})$ is of the form $(\playerIAlt, \indexSegment)$.
  It holds that $\vertex_\indexPosition'$ occurs in $\segment_\indexSegment$ and $\vertex_{\indexPosition+1}'$ does not occur in $\segment_\indexSegment$ by choice of $\indexPosition$.
  The definition of $\stratProfile$ implies that $\vertex_\indexPosition'\in\vertexSetI$ (otherwise, $\vertex_{\indexPosition+1}'$ would occur in $\segment_\indexSegment$).
  We obtain that $\indexPlayer'=\indexPlayer$ and that $\suffix{\play'}{\indexPosition}$ is a play consistent with $\stratAltAdvI$ starting in $\vertex_\indexPosition'$.
  Because $\vertex_\indexPosition'$ occurs in $\segment_\indexSegment$, we have $\vertex_\indexPosition'\in\winningIAdv{\cobuchi{\target_\indexPlayer}}$.
  As in the first case, we obtain $\play'\in\cobuchi{\target_\indexPlayer}$, ending the proof.
\end{proof}

\begin{remark}\label{remark:buchi:ne-size}
  The classical approach to derive NEs from lasso NE outcomes is to encode the lasso into the memory (similarly to the first phase of the above finite-memory strategies).
  This approach thus also applies to the NE outcomes given by Lemma~\ref{lem:buchi:simplelasso}.
  If $|\vertexSet|$ is finite, the strategies result from this approach have a memory size of at most $(|\vertexSet|+2)\nPlayer$.
  Our construction thus yields smaller Mealy machines than the classical approach when there are (much) fewer players than vertices.
  \hfill$\lhd$
\end{remark}

\subsubsection{Outcomes without an infinitely occurring vertex}\label{section:buchi:ne:finitely occuring}
We now study NE outcomes in which no vertices occur infinitely often.
These outcomes are specific to infinite arenas.
We first present how to NE outcomes with a decomposition compatible with our approach from general NE outcomes.
We obtain NE outcomes with an aperiodic infinite decomposition.
We then present a variant of the previously used two-phase mechanism for these outcomes.
While we leave the first phase is unchanged, finite memory no longer suffices to implement a decomposition-based approach for the second phase.
Instead, we use a looser (but sufficient) variant of decomposition-based strategies to construct NEs.

\subparagraph*{Well-structured NE outcomes.}
We first provide a counterpart to Lemma~\ref{lem:buchi:simplelasso} for this case.
Let $\play'\in\playSet{\arena}$ be the outcome of an NE from $\vertex_0\in\vertexSet$ such that no vertex occurs infinitely often in $\play'$.
We observe that $\play'$ need not be a simple play: the objectives satisfied by $\play$ may require visiting targets that are not connected by a simple play or a lasso.
Therefore, simple plays do not suffice for our purpose, and we need memory to construct the outcome (even without taking profitable deviations into account).

We construct NE outcomes that can be obtained with two memory states.
More precisely, we derive an NE outcome $\play$ from $\play'$ that has an infinite simple segment decomposition $(\segment_0, \segment_1, \ldots)$ such that $\satpl{\play} = \satpl{\play'}$, for all $\indexSegment\neq\indexSegment'$ with the same parity, no vertex of $\segment_\indexSegment$ occurs in $\segment_{\indexSegment'}$ and, for all $\indexSegment\geq 1$, no targets of players whose objective is not satisfied by $\play'$ occurs in $\segment_\indexSegment$.
The condition on segments with the same parity implies that $\play$ is the outcome of a strategy profile with memory $\{1, 2\}$: intuitively, memory state $1$ is used when following an odd segment and $2$ is used when following an even segment.

We construct the decomposition $(\segment_0, \segment_1, \ldots)$ as follows.
The segment $\segment_0$ is obtained similarly to the previous section: there exists a position $\indexPosition\in\IN$ such that no vertices of $\target_\indexPlayer$ occur in $\suffix{\play'}{\indexPosition}$ for any $\indexPlayer\notin\satpl{\play'}$.
We let $\indexPosition_0\geq\indexPosition$ be such that $\vertex_{\indexPosition_0}\in\target_\indexPlayer$ for some $\indexPlayer\in\satpl{\play'}$, and choose $\segment_0$ to be a simple history that shares its first and last vertices with $\prefix{\play'}{\indexPosition_0}$ and that uses only vertices occurring in this prefix. 

The other segments are constructed by induction.
We explain how $\segment_1$ is defined from $\indexPosition_0$ to illustrate the idea of the general construction.
As no vertices appear infinitely often in $\play$, there exists some position $\indexPosition_1 > \indexPosition_0$ such that no vertex of $\prefix{\play}{\indexPosition_0}$ occurs in $\suffix{\play}{\indexPosition_1}$ and $\vertex_{\indexPosition_1}\in\target_\indexPlayer$ for some $\indexPlayer\in\satpl{\play}$.
We let $\segment_1$ be a simple history that starts in $\last{\segment_0}$, ends in $\last{\prefix{\play}{\indexPosition_1}}$ and uses only vertices that occur in the segment of $\play$ between positions $\indexPosition_0$ and $\indexPosition_1$.
If we construct $\segment_2$ similarly from $\indexPosition_1$ (by induction), then it shares no vertices with $\segment_0$ by choice of $\indexPosition_1$.
Proceeding with this inductive construction while ensuring that vertices of $\target_\indexPlayer$ occur infinitely often for all $\indexPlayer\in\satpl{\play}$, we obtain the desired decomposition.
Furthermore, the play described by this decomposition is an NE outcome by Theorem~\ref{thm:charac:reach:buchi}.

\begin{restatable}{lemma}{lemBuchiSimpleInfinite}\label{lem:buchi:simpleinfinite}
  Let $\play'$ be the outcome of an NE from $\vertex_0\in\vertexSet$ in the Büchi game $\game$ such that no vertex occurs infinitely often in $\play'$.
  Then there exists an NE outcome $\play$ from $\vertex_0$ with $\satpl{\play}=\satpl{\play'}$ such that
  $\play$ admits an infinite simple segment decomposition $(\segment_0, \segment_1, \ldots)$ such that
  \begin{enumerate}[(i)]
  \item for all $\indexSegment\geq 1$ and all $\indexPlayer\in\playerSet\setminus\satpl{\play}$, no vertex of $\target_\indexPlayer$ occurs in $\segment_\indexSegment$; \label{item:buchi:simpleinfinite:one}
  \item for all $\indexSegment\neq\indexSegment'$, $\segment_\indexSegment$ and $\segment_{\indexSegment'}$ have no vertices in common if $\indexSegment$ and $\indexSegment'$ have the same parity.\label{item:buchi:simpleinfinite:two}
  \end{enumerate}
\end{restatable}
\begin{proof}
  We write $\play'=\vertex_0\vertex_1\ldots$ 
  For convenience of notation, we assume that $\satpl{\play'} = \integerInterval{\numSegments}$ where $\numSegments = |\satpl{\play'}|\geq 1$.
  We define the sought outcome $\play$ via an infinite segment decomposition.

  We let $\indexPosition_0$ be the minimum $\indexPosition\in\IN$ such that $\vertex_\indexPosition\in\target_1$ and, for all $\indexPlayer\in\playerSet\setminus\satpl{\play}$, no vertices of $\target_\indexPlayer$ occur in $\suffix{\play}{\indexPosition}$.
  We let $\segment_0$ be a simple history from $\vertex_0$ to $\vertex_{\ell_0}$ that uses only vertices occurring in $\prefix{\play'}{\ell_0}$.

  We now assume that segments $\segment_0$, \ldots, $\segment_\indexSegment$ and positions $\ell_0$, \ldots, $\ell_\indexSegment$ are defined.
  We assume by induction that (a) for all $\indexSegment'\leq \indexSegment$, $\last{\segment_{\indexSegment'}}\in\target_{\indexSegment'\bmod\numSegments + 1}$, (b) for all $\indexSegment'\leq \indexSegment$, $\segment_{\indexSegment'}$ contains only vertices occurring in $\prefix{\play'}{\ell_{\indexSegment'}}$, (c) for all $\indexSegment' < \indexSegment$, no vertices of $\prefix{\play'}{\ell_{\indexSegment'}}$ occur in $\suffix{\play'}{\ell_{\indexSegment'+1}}$, and (d) for all $\indexSegment''\leq\indexSegment'-2$, $\segment_{\indexSegment'}$ and $\segment_{\indexSegment''}$ have no vertices in common.
  
  We define $\ell_{\indexSegment+1}$ and $\segment_{\indexSegment+1}$ as follows.
  There exists some $\ell$ such that no vertex of $\prefix{\play'}{\ell_\indexSegment}$ occurs in $\suffix{\play'}{\ell}$. We choose $\ell_{\indexSegment+1}>\ell$ such that $\vertex_{\ell_{\indexSegment+1}}\in\target_{(\indexSegment'+1)\bmod\numSegments + 1}$.
  We let $\segment_{\indexSegment+1}$ be a simple history from $\vertex_{\ell_\indexSegment}$ to $\vertex_{\ell_{\indexSegment+1}}$ that uses only vertices occurring in the segment $\vertex_{\ell_\indexSegment}\ldots\vertex_{\ell_{\indexSegment+1}}$ of $\play'$.
  
  We show that the induction hypothesis is preserved by this choice.
  Properties (a), (b) and (c) hold by definition.
  We show that (d) holds, i.e., that for all $\indexSegment' \leq \indexSegment - 1$, $\segment_{\indexSegment+1}$ and $\segment_{\indexSegment'}$ have no vertices in common.
  Let $\indexSegment' \leq \indexSegment - 1$.
  It holds that the vertices occurring in $\segment_{\indexSegment+1}$ all appear in $\suffix{\play'}{\ell_\indexSegment}$.
  By induction, all vertices of $\segment_{\indexSegment'}$ occur in $\prefix{\play'}{\ell_{\indexSegment'}}$, and none of these vertices occur in $\suffix{\play'}{\ell_\indexSegment}$.
  This ends the inductive construction.

  To end the proof, we note that the play $\play = \concat{\segment_0}{\concat{\segment_1}{\ldots}}$ is the outcome of an NE by Theorem~\ref{thm:charac:reach:buchi}: all vertices appearing in $\play$ appear in $\play'$ and by construction, $\satpl{\play}=\satpl{\play'}$.
\end{proof}

\subparagraph*{Finite-memory Nash equilibria.}
We now show that the plays given by Lemma~\ref{lem:buchi:simpleinfinite} are outcomes of finite-memory NEs.
Let $\stratProfile'$ be an NE from a vertex $\vertex_0\in\vertexSet$.
We apply Lemma~\ref{lem:buchi:simpleinfinite} to $\outcome{\stratProfile'}{\vertex_0}$ to obtain a play $\play$ and a decomposition $\decomp = (\segment_0, \segment_1, \ldots)$ of $\play$ satisfying~\ref{item:buchi:simpleinfinite:one} and~\ref{item:buchi:simpleinfinite:two} of Lemma~\ref{lem:buchi:simpleinfinite}.

We define a finite-memory NE from $\vertex_0$ whose outcome is $\play$ via the two-phase mechanism previously used in the proof of Theorem~\ref{thm:buchi:fmne:lasso}.
The first phase is defined identically to the previous section (see the proof of Theorem~\ref{thm:buchi:fmne:lasso}).
For the second phase, we use a variation of the Mealy machines described in Definition~\ref{def:compatible mealy}.
In said definition, for each segment index $\indexSegment$ and each memory state of the form $(\playerI, \indexSegment)$ , we refer to the simple segment $\segment_\indexSegment$ to define memory updates and next moves.
This approach no longer yields a finite-memory strategy with the infinite decomposition $\decomp$.
Instead of having a group of memory states for each segment, we use a group of memory states for each parity.
More precisely, we use memory states of the form $(\playerI, 1)$ for all odd segments and memory states of the form $(\playerI, 2)$  for all even segments besides $\segment_0$ -- the second component of the memory state indicates the parity of the current segment.
When the end of an odd (resp.~even) segment is reached in a memory state of the form $(\playerI, 1)$  (resp.~$(\playerI, 2)$), the memory is updated to a state of the form $(\playerIAlt, 2)$ (resp.~$(\playerI, 1)$).
We obtain a well-defined next-move function because segments of $\decomp$ with the same parity traverse pairwise disjoint sets of vertices.

If at some point in the second phase, a vertex that does not occur in an odd segment is read in a memory state $(\playerI, 1)$, the memory is updated to a punishing state $\playerI$, such that players attempt to punish $\playerI$ with a memoryless strategy from a coalition game.
We proceed similarly for the even case.
The resulting finite-memory strategy profile is an NE from $\vertex_0$.
On the one hand, any deviation such that the memory never updates to a punishing state must only have vertices that occur segments of $\decomp$ with a non-zero index in the limit.
By choice of $\decomp$, this deviation cannot be profitable.
Otherwise, it can be shown with Theorem~\ref{thm:charac:reach:buchi} that the punishing strategy successfully sabotages the deviating player whenever their objective is not satisfied in $\play$.
We formally describe the construction above and establish its correctness below.

\begin{restatable}{theorem}{thmBuchiNENoLasso}\label{thm:buchi:fmne:nolasso}
  Let $\stratProfile'$ be an NE from a vertex $\vertex_0$ such that no vertex occurs infinitely often in $\outcome{\stratProfile'}{\vertex_0}$.
  There exists a finite-memory NE $\stratProfile$ from $\vertex_0$ such that $\satpl{\outcome{\stratProfile}{\vertex_0}}=\satpl{\outcome{\stratProfile'}{\vertex_0}}$.
\end{restatable}
\begin{proof}
  Let $\play$ be an NE outcome obtained via Lemma~\ref{lem:buchi:simpleinfinite} from $\outcome{\stratProfile'}{\vertex_0}$.
  We let $\decomp = (\segment_0, \segment_1, \ldots)$ be the decomposition provided by the lemma.
  We note that all segments in $\decomp$ besides $\segment_0$ are non-trivial simple histories; if $\segment_\indexSegment$ is trivial or a cycle for some $\indexSegment\geq 1$, then $\last{\segment_{\indexSegment-1}} = \first{\segment_{\indexSegment+1}}$, which contradicts the fact that segments with the same parity in $\decomp$ do not have vertices in common.
  We prove that $\play$ is the outcome of a finite-memory NE.
  The construction below is an adaptation of the proof of Theorem~\ref{thm:buchi:fmne:lasso}.
  
  We introduce some notation first.
  Let $\playerSubset\subseteq\playerSet$ be $\playerSet\setminus\satpl{\play}$ if this set is not empty, or $\{1\}$ otherwise.
  For all $\indexPlayer\in\playerSubset$, we let $\stratAltAdvI$ be a memoryless uniformly winning strategy for the second player of the coalition game $\game_\indexPlayer = (\arena_\indexPlayer, \buchi{\target_\indexPlayer})$ (it exists by Theorem~\ref{thm:qualitative:ML strat}) and let $\winningIAdv{\cobuchi{\target_\indexPlayer}}$ denote the winning region of this player in $\game_\indexPlayer$.
  We write $\segment_0 = \vertex_0\ldots\vertex_\indexLast$.
  Finally, we define $\vertexSet_\playerSubset = \bigcup_{\indexPlayer\in\playerSubset}\vertexSetI$, $S_1$ (resp.~$S_2$) to be the vertices occurring in segments $\segment_\indexSegment$ with odd $\indexSegment$ (resp.~even $\indexSegment\geq 2$), $L_1 = \{\last{\segment_\indexSegment}\mid \indexSegment\in 2\IN+1\}$ and $L_2 = \{\last{\segment_\indexSegment}\mid \indexSegment\in 2\IN+2\}$ be the set of last vertices of odd and positive even segments respectively.

  For each $\indexPlayer\in\playerSet$, we define a Mealy machine $\mealyMachine_\indexPlayer= (\mealyStateSpace, \mealyStateInit, \mealyUpdate, \mealyNextI)$ as follows.
  We define $\mealyStateSpace$ as the set $\{\vertex_\indexPosition\mid 0\leq\indexPosition \leq \indexLast\}\cup (\{\playerI\mid\indexPlayer\in\playerSubset\}\times \integerInterval{2})\cup \{\playerI\mid\indexPlayer\in\playerSubset\}$ and $\mealyStateInit=\vertex_0$.

  The update function $\mealyUpdate$ is defined as follows.
  \begin{enumerate}[(a)]
  \item For all $\indexPosition < \indexLast$, we let $\mealyUpdate(\vertex_\indexPosition, \vertex_\indexPosition) = \vertex_{\indexPosition+1}$.
\item For all $\indexPosition \leq \indexLast$ and $\vertex\neq\vertex_\indexPosition$,
    \begin{itemize}
    \item if $\indexPosition\geq 1$ and there exists $\indexPlayer\in\playerSubset$ such that $\vertex_{\indexPosition-1}\in\vertexSet_{\indexPlayer}$ (i.e., $\vertex\in\vertexSet_\playerSubset$), we let $\mealyUpdate(\vertex_\indexPosition, \vertex) = \playerI$;
    \item otherwise, $\mealyUpdate(\vertex_\indexPosition, \vertex)$ is left arbitrary.
    \end{itemize}
  \item For all states of the form $(\playerI, \indexParity)\in\mealyStateSpace$ and all $\vertex\in S_\indexParity$, we let $\mealyUpdate((\playerI, \indexParity), \vertex) = (\playerIAlt, \indexParity')$ where
    \begin{itemize}
      \item $\playerIAlt$ is the player controlling $\vertex$ if $\vertex\in\vertexSet_\playerSubset$ and otherwise, $\indexPlayer'=\indexPlayer$, and 
      \item $\indexParity'=\indexParity$ if $\vertex\notin L_\indexParity$, and otherwise we set $\indexParity'=3 - \indexParity$ (i.e., if $\indexParity=1$, it becomes $2$ and vice-versa).
      \end{itemize}
    \item For all states of the form $(\playerI, \indexParity)\in\mealyStateSpace$ and all $\vertex\in\vertexSet\setminus S_\indexParity$, we let $\mealyUpdate((\playerI, \indexParity), \vertex)=\playerI$.
    \item We let $\mealyUpdate(\vertex_\indexLast, \vertex_\indexLast) = \mealyUpdate((\playerI, 1), \vertex_\indexLast)$ for any $\indexPlayer\in\playerSubset$ (the choice of $\indexPlayer$ here does not impact the following argument because $\vertex_\indexLast\in S_1$).
    \item Finally, for all $\indexPlayer\in\playerSubset$ and $\vertex\in\vertexSet$, we let $\mealyUpdate(\playerI, \vertex) = \playerI$.
\end{enumerate}

  Let $\indexPlayer\in\playerSubset$. We now define $\mealyNextI$.
  Let $\vertex\in\vertexSetI$.
  \begin{enumerate}[(a)]
  \item We first consider memory states of the form $\vertex_\indexPosition$.
    Fix $\indexPosition\leq\indexLast$.
    \begin{itemize}
    \item If $\vertex = \vertex_\indexPosition$ and $\indexPosition\neq\indexLast$, we let $\mealyNextI(\vertex_\indexPosition, \vertex)=\vertex_{\indexPosition+1}$.
    \item If $\vertex = \vertex_\indexPosition$ and $\indexPosition=\indexLast$, we let $\mealyNextI(\vertex_\indexLast, \vertex_\indexLast)$ be the second vertex of $\segment_1$ (which exists because $\segment_1$ is non-trivial).
    \item We distinguish two cases whenever $\vertex\neq\vertex_\indexPosition$:
      \begin{itemize}
      \item if $\indexPosition\geq 1$ and there exists $\indexPlayer'\in\playerSubset$ such that $\vertex_{\indexPosition-1}\in\vertexSet_{\indexPlayer'}$ and $\indexPlayer'\neq\indexPlayer$, we let $\mealyNextI(\vertex_\indexPosition, \vertex)=\stratAltAdvIb(\vertex)$;
      \item we let $\mealyNextI(\vertex_\indexPosition, \vertex)$ be arbitrary in all other cases.
      \end{itemize}
    \end{itemize}
  \item  We now deal with memory states of the form $(\playerIAlt, \indexParity)$.
    Fix $\indexPlayer'\in\playerSubset$ and $\indexParity\in\integerInterval{2}$.
    \begin{itemize}
    \item If $\vertex\in S_\indexParity\setminus L_\indexParity$, we set $\mealyNextI((\playerIAlt, \indexParity), \vertex)$ to the vertex following $\vertex$ in $\segment_\indexSegment$, where $\indexSegment\in 2\IN+p$ is such that $\vertex$ occurs in $\segment_\indexSegment$ (such a $\indexSegment$ is unique because all segments with the same parity traverse disjoint sets of vertices).
  \item If $\vertex\in L_\indexParity$, we let $\mealyNextI((\playerIAlt, \indexParity), \vertex)$ be the second vertex of $\segment_{\indexSegment}$ where $\indexSegment\in 2\IN+ (3-p)$ is such that $\vertex=\first{\segment_\indexSegment}$ ( such a $\indexSegment$ is unique and $\segment_\indexSegment$ has at least two vertices because it is non-trivial).
  \item  If $\vertex\notin S_\indexParity$ and $\indexPlayer'\neq\indexPlayer$, we let $\mealyNextI((\playerIAlt, \indexParity), \vertex) = \stratAltAdvIb(\vertex)$.
  \item If $\vertex\notin S_\indexParity$ and $\indexPlayer'=\indexPlayer$, we let $\mealyNextI((\playerIAlt, \indexParity), \vertex)$ be arbitrary.
  \end{itemize}
\item Finally, we let $\mealyNextI(\playerIAlt, \vertex) = \stratAltAdvIb(\vertex)$ if $\indexPlayer'\neq\indexPlayer$, and otherwise we let $\mealyNextI(\playerIAlt, \vertex)$ be arbitrary.
\end{enumerate}
  We let $\stratI$ be the strategy induced by $\mealyMachine_\indexPlayer$.
  It can be shown by induction that the outcome of $\stratProfile=(\stratI)_{\indexPlayer\in\playerSet}$ is $\play$.
  We omit the proof here; it is very close to the argument for coherence in the proof of Lemma~\ref{lemma:coherence:template:extensions}.
  
  We now show that $\stratProfile$ is an NE from $\vertex_0$ by proving that for all $\indexPlayer\in\playerSet\setminus\satpl{\play}$, $\playerI$ does not have a profitable deviation.
  We fix $\indexPlayer\in\playerSet\setminus\satpl{\play}$.
  By Theorem~\ref{thm:charac:reach:buchi}, all vertices in $\play$ are in $\winningIAdv{\cobuchi{\target_\indexPlayer}}$.
  
  Let $\play'=\vertex_0'\vertex_1'\ldots$ be a play starting in $\vertex_0$ that is consistent with the strategy profile $\stratIAdv$.
  We consider three cases.
  First, assume that $\play'$ does not have $\segment_0$ as a prefix.
  We conclude that $\play'\in\cobuchi{\target_\indexPlayer}$ in the same way as in the proof of Theorem~\ref{thm:buchi:fmne:lasso}.
  Second, assume that $\play'$ has $\segment_0$ as a prefix, and all vertices $\vertex$ occurring in $\suffix{\play'}{\indexLast}$ are elements of $S_1\cup S_2$. 
  Because $S_1$ and $S_2$ do not intersect $\target_\indexPlayer$ (\ref{item:buchi:simpleinfinite:one} in Lemma~\ref{lem:buchi:simpleinfinite}), it follows that $\play'\in\cobuchi{\target_\indexPlayer}$.

  Finally, assume that $\play'$ has $\segment_0$ as a prefix and some vertex appearing in $\suffix{\play'}{\indexLast}$ is not an element of $S_1\cup S_2$.
  It follows that the memory state of the players relying on $\mealyMachine_\indexPlayer$ eventually becomes of the form $\playerIAlt$.
  Let $\indexPosition\in\IN$ be the largest number such that $\widehat{\mealyUpdate}(\prefix{\play'}{\indexPosition})$ is of the form $(\playerIAlt, \indexParity)$.
  It holds that $\vertex_\indexPosition'\in S_\indexParity$ and $\vertex_{\indexPosition+1}'\notin S_\indexParity$ by choice of $\indexPosition$.
  It follows that $\vertex_\indexPosition'\in\vertexSetI$ by definition of $\stratProfile$ (otherwise, $\vertex_{\indexPosition+1}'$ would be an element of $S_\indexParity$).
  We obtain that $\indexPlayer'=\indexPlayer$ and that $\suffix{\play'}{\indexPosition}$ is a play consistent with $\stratAltAdvI$ starting in $\vertex_\indexPosition'$.
  Because $\vertex_\indexPosition'\in S_\indexParity$, we have $\vertex_\indexPosition'\in\winningIAdv{\cobuchi{\target_\indexPlayer}}$.
  Therefore, $\play'\in\cobuchi{\target_\indexPlayer}$.
  This ends the proof.
\end{proof}

\subsection{Nash equilibria in co-Büchi games}\label{section:buchi:cobuchi-ne}

In this section, we show that from any NE in a co-Büchi game, we can derive a finite-memory NE from the same initial vertex such that the objectives satisfied in the outcome of the first NE are also satisfied by the outcome of the second one.
We build finite-memory NEs from well-structured NE outcomes with a variant of the two-phase approach used in the previous section.
We use a simplification process to construct well-structured NE outcomes.
In Section~\ref{section:co-buchi:well-structured}, we present how we simplify NE outcomes in co-Büchi games.
We build finite-memory NEs from these outcomes in Section~\ref{section:co-buchi:ne}.

We let $\game = (\arena, (\cobuchi{\target_\indexPlayer})_{\indexPlayer\in\playerSet})$ for the remainder of the section.

\subsubsection{Well-structured Nash equilibrium outcomes}\label{section:co-buchi:well-structured}

In safety games, players can exploit the decomposition-based behaviour of players to obtain profitable deviations (see Example~\ref{example:safe:in-segment}) if no additional conditions are imposed on NE outcomes (e.g., those of Lemma~\ref{lem:safety:simple decomposition}).
These deviations are not harmful to players whose objectives are already satisfied.

In co-Büchi games, these deviations can be harmful: if we use a fully decomposition-based approach to define strategies from NE outcomes, plays resulting from profitable in-segment deviations may violate some objectives satisfied by the original NE outcome (Example~\ref{example:buchi:one}).
This motivates the use of a two-phase approach in co-Büchi games: we adopt a decomposition-based behaviour once in-segment deviations can no longer be harmful to winning players.
Formally, we show that from any NE outcome $\play'$ in $\game$, we can derive an NE outcome $\play$ from $\first{\play'}$ that is a simple play or a simple lasso such that $\satpl{\play'}\subseteq\satpl{\play}$, and there are no plays $\play''$ from $\first{\suffix{\play}{\indexLast^\star}}$ such that $\satpl{\play}\subsetneq\satpl{\play''}$ and $\play''$ uses only vertices appearing in $\suffix{\play}{\indexLast^\star}$ where $\indexLast^\star$ is the smallest $\indexLast\in\IN$ such that no targets of players in $\satpl{\play}$ appear in $\suffix{\play}{\indexLast}$.
The prefix $\prefix{\play}{\indexLast^\star}$ corresponds to the first phase and the suffix $\suffix{\play}{\indexLast^\star}$ to the second phase of our approach, and the condition imposed on $\play$ guarantees the absence of profitable deviations.
To construct $\play$, we proceed by induction: we repeatedly simplify $\play'$ into well-chosen simple lasso or simple plays until we obtain a suitable play.
The proof of the following lemma provides a formal construction. 

\begin{lemma}\label{lem:cobuchi:outcomes}
  Let $\play'$ be the outcome of an NE from $\vertex_0\in\vertexSet$ in the co-Büchi game $\game$.
  There exists an NE outcome $\play$ from $\vertex_0$ such that
  \begin{enumerate}[(i)]
  \item $\play$ is a simple play or a simple lasso;\label{item:cobuchi:decomp:one}
\item there are no plays of the form $\concat{\prefix{\play}{\indexLast^\star}}{\play''}$ such that only vertices of $\suffix{\play}{\indexLast^\star}$ occur in $\play''$ and $\satpl{\play}\subsetneq\satpl{\play''}$ where $\indexLast^\star=\min\{\indexLast\in\IN\mid \forall\indexPosition\geq\indexLast,\,\forall\indexPlayer\in\satpl{\play},\,\vertex_\indexPosition\notin\target_\indexPlayer\}$; and\label{item:cobuchi:decomp:three}
  \item $\satpl{\play'}\subseteq\satpl{\play}$.\label{item:cobuchi:decomp:four}
  \end{enumerate}
\end{lemma}
\begin{proof}
  We construct two finite sequences $\playB^{(0)}$, \ldots, $\playB^{(\numSegments)}$ and $\play^{(0)}$, \ldots, $\play^{(\numSegments)}$ of NE outcomes from $\vertex_0$ such that $\playB^{(0)} = \play'$, the plays $\play^{(0)}$, \ldots, $\play^{(\numSegments)}$ are simple plays or simple lassos, for all $\indexSegment\in \integerInterval{\numSegments}$, $\playB^{\indexSegment}$ witnesses that $\play^{\indexSegment-1}$ does not satisfy~\ref{item:cobuchi:decomp:three}, $\play^{\numSegments}$ satisfies~\ref{item:cobuchi:decomp:three} and we have the chain of inclusions
  \[
    \satpl{\playB^{(0)}}\subseteq\satpl{\play^{(0)}}\subsetneq\satpl{\playB^1}\subseteq\satpl{\play^1}\subsetneq\ldots\subsetneq \satpl{\playB^{(\numSegments)}}\subseteq\satpl{\play^{(\numSegments)}}.
  \]
  We construct these sequences by induction and conclude by letting $\play = \play^{(\numSegments)}$.
  Termination is guaranteed:~\ref{item:cobuchi:decomp:three} holds trivially for plays that satisfy all objectives in $\game$.

  Assume that $\playB^{(\indexSegment)} = \vertex_0\vertex_1\ldots$ is well-defined and let $\indexLast_\indexSegment = \min\{\indexLast\in\IN\mid \forall\indexPosition\geq\indexLast,\,\forall\indexPlayer\in\satpl{\playB^{(\indexSegment)}},\,\vertex_\indexPosition\notin\target_\indexPlayer\}$.
  We define $\play^{(\indexSegment)}$ as a combination of a history $\segment_\indexSegment$ and a play $\segment_\indexSegment'$.
  We distinguish two cases for the suffixes.
  If $\suffix{\playB^{(\indexSegment)}}{\indexLast_\indexSegment}$ is a simple play, we let $\segment_\indexSegment' = \suffix{\playB^{(\indexSegment)}}{\indexLast_\indexSegment}$.
  Otherwise, we let $\segment_\indexSegment'$ be a simple lasso from $\first{\suffix{\playB^{(\indexSegment)}}{\indexLast_\indexSegment}}$ such that all states that occur in $\segment_\indexSegment'$ occur in $\suffix{\playB^{(\indexSegment)}}{\indexLast_\indexSegment}$ (such a lasso exists because $\suffix{\playB^{(\indexSegment)}}{\indexLast_\indexSegment}$ is not a simple play).
  For both cases, we choose the (possibly trivial) prefix $\segment_\indexSegment$ as a history starting in $\first{\playB^{(\indexSegment)}}$ that uses only vertices occurring in $\playB^{(\indexSegment)}$ such that $\play^{(\indexSegment)} = \concat{\segment_\indexSegment}{\segment_\indexSegment'}$ is a simple play or a simple lasso.
All states that occur infinitely often in $\play^{(\indexSegment)}$ are taken from $\suffix{\playB^{(\indexSegment)}}{\indexLast_\indexSegment}$.
  Therefore, the definition of $\indexLast_\indexSegment$ ensures that $\satpl{\playB^{(\indexSegment)}}\subseteq\satpl{\play^{(\indexSegment)}}$.
  Furthermore, Theorem~\ref{thm:charac:reach:buchi} ensures that $\play^{(\indexSegment)}$ is an NE outcome (all vertices occurring in $\play^{(\indexSegment)}$ occur in the NE outcome $\playB^{(\indexSegment)}$).
  This shows that $\play^{(\indexSegment)}$ satisfies all required properties of the induction that are unrelated to~\ref{item:cobuchi:decomp:three}. 

  If $\play^{(\indexSegment)}$ satisfies~\ref{item:cobuchi:decomp:three}, then we end the induction.
  Otherwise, we let $\playB^{(\indexSegment+1)}$ witnessing that~\ref{item:cobuchi:decomp:three} does not hold for $\play^{(\indexSegment)}$.
  To end the proof, it remains to show that $\playB^{(\indexSegment+1)}$ is an NE outcome as the strict inclusion $\satpl{\play^{\indexSegment}}\subsetneq\satpl{\playB^{(\indexSegment+1)}}$ follows from the choice of $\playB^{(\indexSegment+1)}$.
  This follows from Theorem~\ref{thm:charac:reach:buchi}, as all vertices occurring in $\playB^{(\indexSegment+1)}$ occur in $\play^{\indexSegment}$ and $\satpl{\play^{\indexSegment}}\subsetneq\satpl{\playB^{(\indexSegment+1)}}$.
\end{proof}

\subsubsection{Finite-memory Nash equilibria in co-Büchi games}\label{section:co-buchi:ne}

We now build on the NE outcomes of Lemma~\ref{lem:cobuchi:outcomes} to construct finite-memory NEs.
Let $\play$ be an NE outcome from $\vertex_0$ given by Lemma~\ref{lem:cobuchi:outcomes}, and let $\indexLast^\star$ be the least index such that no targets (to be avoided) of players for which $\play$ is winning occur in $\suffix{\play}{\indexLast^\star}$.
We construct Mealy machines as follows.
In the first phase, we punish any player which deviates from the history $\prefix{\play}{\indexLast^\star}$.
In the second phase, we follow a decomposition-based strategy with respect to the trivial simple decomposition $(\suffix{\play}{\indexLast^\star})$ and punish any player who exits the set of vertices of $\suffix{\play}{\indexLast^\star}$.
As usual, players are punished with memoryless uniformly winning strategies of adversaries in coalition games.

The resulting strategy profile is an NE from $\vertex_0$.
Let $\indexPlayer\in\playerSet\setminus\satpl{\play}$.
On the one hand, if $\playerI$ deviates from $\prefix{\play}{\indexLast^\star}$ or exits the set of vertices of the segment $\suffix{\play}{\indexLast^\star}$ in the second phase, then the play resulting from the deviation has a suffix that is consistent with a winning strategy of the adversary in the coalition game of $\playerI$ and that starts from a vertex in the winning region of this adversary by Theorem~\ref{thm:charac:reach:buchi}.
On the other hand, if $\playerI$ respects $\prefix{\play}{\indexLast^\star}$ and never exits the set of vertices of $\suffix{\play}{\indexLast^\star}$ in the second phase, then property~\ref{item:cobuchi:decomp:three} of Lemma~\ref{lem:cobuchi:outcomes} for $\play$ guarantees that the deviation is not profitable.
We formalise the above in the following proof.

\begin{theorem}\label{theorem:cobuchi:fm-ne}
  Let $\stratProfile'$ be an NE from a vertex $\vertex_0$ in the co-Büchi game $\game$.
  There exists a finite-memory NE $\stratProfile$ from $\vertex_0$ such that $\satpl{\outcome{\stratProfile'}{\vertex_0}}\subseteq\satpl{\outcome{\stratProfile}{\vertex_0}}$.
  If $\arena$ is finite, a memory size of at most $|\vertexSet| + 2\nPlayer$ suffices.
\end{theorem}
\begin{proof}
  Let $\play=\vertex_0\vertex_1\ldots$ be an NE outcome obtained via  Lemma~\ref{lem:cobuchi:outcomes} from $\outcome{\stratProfile'}{\vertex_0}$.
  We show that $\play$ is the outcome of a finite-memory NE.
  We first introduce some notation.
  Let ${\indexLast^\star} = \min\{\indexLast\in\IN\mid\forall\indexPosition\geq\indexLast,\,\forall\indexPlayer\in\satpl{\play},\,\vertex_\indexPosition\notin\target_\indexPlayer\}$.
  Let $\playerSubset = \playerSet\setminus\satpl{\play}$ if this set is not empty, and let $\playerSubset = \{1\}$ otherwise.
  For all $\indexPlayer\in\playerSubset$, we let $\stratAltAdvI$ be a memoryless uniformly winning strategy for the second player of the coalition game $\game_\indexPlayer = (\arena_\indexPlayer, \cobuchi{\target_\indexPlayer})$ (it exists by Theorem~\ref{thm:qualitative:ML strat}) and let $\winningIAdv{\buchi{\target_\indexPlayer}}$ denote the winning region of this player in $\game_\indexPlayer$.
  Finally, we define $\vertexSet_\playerSubset = \bigcup_{\indexPlayer\in\playerSubset}\vertexSetI$.
  
  For each $\indexPlayer\in\playerSet$, we define a Mealy machine $\mealyMachine_\indexPlayer= (\mealyStateSpace, \mealyStateInit, \mealyUpdate, \mealyNextI)$ as follows.
  We define $\mealyStateSpace$ as the set $\{\vertex_\indexPosition\in\vertexSet\mid 0\leq\indexPosition \leq \indexLast^\star\}\cup \{(\playerI, 1)\mid\indexPlayer\in\playerSubset\} \cup \{\playerI\mid\indexPlayer\in\playerSubset\}$ and $\mealyStateInit=\vertex_0$.
  With respect to the informal explanation preceding the theorem (and as was the case for Büchi games), the memory states that are vertices correspond to the first phase (enforcing the history $\prefix{\play}{{\indexLast^\star}}$), and the others to the second phase (following a strategy based on a simple decomposition).
  The memory bounds claimed for finite arenas hold by simplicity of $\prefix{\play}{{\indexLast^\star}}$.

  The update function $\mealyUpdate$ is defined as follows.
  \begin{enumerate}[(a)]
  \item For all $\indexPosition < {\indexLast^\star}$, we let $\mealyUpdate(\vertex_\indexPosition, \vertex_\indexPosition) = \vertex_{\indexPosition+1}$.
\item For all $\indexPosition \leq {\indexLast^\star}$ and all $\vertex\neq\vertex_\indexPosition$,
    \begin{itemize}
    \item if $\indexPosition\geq 1$ and there exists $\indexPlayer\in\playerSubset$ such that $\vertex_{\indexPosition-1}\in\vertexSet_{\indexPlayer}$, we let $\mealyUpdate(\vertex_\indexPosition, \vertex) = \playerI$;
    \item otherwise the update is arbitrary.
    \end{itemize}
  \item For all $\indexPlayer\in\playerSubset$ and $\vertex\in\vertexSet$ occurring in $\suffix{\play}{{\indexLast^\star}}$, we let $\mealyUpdate((\playerI, 1), \vertex) = (\playerIAlt, 1)$ where $\playerIAlt$ is the player controlling $\vertex$ if $\vertex\in\vertexSet_\playerSubset$ and otherwise, $\indexPlayer'=\indexPlayer$.
  \item For all $\indexPlayer\in\playerSubset$ and all $\vertex\in\vertexSet$ that do not occur in $\suffix{\play}{{\indexLast^\star}}$, we let $\mealyUpdate((\playerI, 1), \vertex)=\playerI$.
  \item We let $\mealyUpdate(\vertex_{\indexLast^\star}, \vertex_{\indexLast^\star}) = \mealyUpdate((\playerI, 1), \vertex_{\indexLast^\star})$ for some $\indexPlayer\in\playerSubset$ (the choice of $\indexPlayer$ here does not impact the following arguments because $\vertex_{\indexLast^\star}$ occurs in $\suffix{\play}{{\indexLast^\star}}$).
  \item Finally, for all $\indexPlayer\in\playerSubset$ and $\vertex\in\vertexSet$, we let $\mealyUpdate(\playerI, \vertex) = \playerI$.
  \end{enumerate}

  Let $\indexPlayer\in\playerSubset$. We now define $\mealyNextI$.
  Let $\vertex\in\vertexSetI$.
  \begin{enumerate}[(a)]
  \item We first consider memory states of the form $\vertex_\indexPosition$.
    Fix $\indexPosition\leq{\indexLast^\star}$.
    \begin{itemize}
    \item If $\vertex = \vertex_\indexPosition$ and $\indexPosition\neq{\indexLast^\star}$, we let $\mealyNextI(\vertex_\indexPosition, \vertex)=\vertex_{\indexPosition+1}$.
    \item If $\vertex = \vertex_\indexPosition$ and $\indexPosition={\indexLast^\star}$, we let $\mealyNextI(\vertex_{\indexLast^\star}, \vertex_{\indexLast^\star})$ be the second vertex of $\suffix{\play}{{\indexLast^\star}}$ (i.e., the only vertex that follows $\vertex_{{\indexLast^\star}}$ in $\suffix{\play}{\indexLast^\star}$) if $\indexPosition={\indexLast^\star}$.
    \item We consider two cases whenever $\vertex\neq\vertex_\indexPosition$.
      \begin{itemize}
      \item If $\indexPosition\geq 1$ and there exists $\indexPlayer'\in\playerSubset$ such that $\vertex\in\vertexSet_{\indexPlayer'}$ and $\indexPlayer'\neq\indexPlayer$, we let  and set $\mealyNextI(\vertex_\indexPosition, \vertex)=\stratAltAdvIb(\vertex)$.
      \item We let $\mealyNextI(\vertex_\indexPosition, \vertex)$ be arbitrary in all other cases.
      \end{itemize}
    \end{itemize}
  \item We now deal with memory states of the form $(\playerIAlt, 1)$.
    Fix $\indexPlayer'\in\playerSubset$.
    \begin{itemize}
    \item If $\vertex$ occurs in $\suffix{\play}{\indexLast^\star}$, we let $\mealyNextI((\playerIAlt, 1), \vertex)$ be the vertex following $\vertex$ in $\suffix{\play}{\indexLast^\star}$ (this vertex is well-defined as $\play$ is a simple play or a simple lasso).
    \item If $\vertex$ does not occur in $\suffix{\play}{\indexLast^\star}$ and $\playerIndex'\neq\playerIndex$, we let $\mealyNextI((\playerIAlt, 1), \vertex) = \stratAltAdvIb(\vertex)$.
    \item If $\vertex$ does not occur in $\suffix{\play}{\indexLast^\star}$ and $\playerIndex'=\playerIndex$, we let $\mealyNextI((\playerIAlt, 1), \vertex)$ be arbitrary.
    \end{itemize}
  \item Finally, we let $\mealyNextI(\playerIAlt, \vertex) = \stratAltAdvIb(\vertex)$ if $\indexPlayer'\neq\indexPlayer$, and otherwise we let it be arbitrary.
\end{enumerate}

  We let $\stratI$ be the strategy induced by $\mealyMachine_\indexPlayer$.
  It can be shown by induction that the outcome of $\stratProfile=(\stratI)_{\indexPlayer\in\playerSet}$ is $\play$ (the argument is analogous to that for coherence from the proof of Lemma~\ref{lemma:coherence:template:extensions}).
  
  It remains to show that $\stratProfile$ is an NE from $\vertex_0$.
  It suffices to show that for all $\indexPlayer\in\playerSet\setminus\satpl{\play}$, $\playerI$ does not have a profitable deviation.
  We fix $\indexPlayer\in\playerSet\setminus\satpl{\play}$.
  The NE outcome characterisation of Theorem~\ref{thm:charac:reach:buchi} implies that all vertices in $\play$ are elements of $\winningIAdv{\buchi{\target_\indexPlayer}}$.

  Let $\play'=\vertex_0'\vertex_1'\ldots$ be a play starting in $\vertex_0$ that is consistent with the strategy profile $\stratIAdv$.
  We consider three cases.
  First, assume that $\play'$ does not have $\prefix{\play}{\indexLast^\star}$ as a prefix.
  Let $\indexPosition < \indexLast^\star$ be such that $\prefix{\play'}{\indexPosition}$ is the longest common prefix of $\play'$ and $\segment_0$.
  We have that $\widehat{\mealyUpdate}(\prefix{\play'}{\indexPosition}) = \vertex_{\indexPosition+1}$.
  The definition of $\stratProfile$ and the relation $\vertex_{\indexPosition+1}'\neq \vertex_{\indexPosition+1}$ imply that $\vertex_\indexPosition\in\vertexSetI$.
  It follows that $\suffix{\play'}{\indexPosition}$ is a play consistent with $\stratAltAdvI$ starting in $\vertex_\indexPosition\in\winningIAdv{\buchi{\target_\indexPlayer}}$, thus $\suffix{\play'}{\indexPosition}\in\buchi{\target_\indexPlayer}$.
  We obtain that $\play'\in\buchi{\target_\indexPlayer}$, ending this first case.
  
  Second, we assume that we can write $\play'$ as a combination $\concat{\prefix{\play}{\indexLast^\star}}{\play''}$ and that all states that occur in $\play''$ also occur in $\suffix{\play}{\indexLast^\star}$.
  By definition of the co-Büchi objective, $\satpl{\play}\subseteq\satpl{\play''}$.
  Furthermore, as $\play$ was obtained through Lemma~\ref{lem:cobuchi:outcomes},~\ref{item:cobuchi:decomp:three} implies that $\satpl{\play}=\satpl{\play''}$, and, in particular, $\play'\in\buchi{\target_\indexPlayer}$.

  Finally, we assume that we can write $\play'$ as a combination $\concat{\prefix{\play}{\indexLast^\star}}{\play''}$ and that there is a state occurring in $\play''$ but not in $\suffix{\play}{\indexLast^\star}$.  
  It follows that the memory state of the Mealy machines of the players in $\playerIAdv$ eventually becomes of the form $\playerIAlt$.
  Let $\indexPosition\in\IN$ be the largest number such that $\widehat{\mealyUpdate}(\prefix{\play'}{\indexPosition})$ is of the form $(\playerIAlt, 1)$.
  By choice of $\indexPosition$, it holds that $\vertex_\indexPosition'$ occurs in $\suffix{\play}{\indexLast^\star}$ and $\vertex_{\indexPosition+1}'$ does not.
  It follows that $\vertex_\indexPosition'\in\vertexSetI$ by definition of $\stratProfile$.
  We obtain that $\indexPlayer'=\indexPlayer$ and that $\suffix{\play'}{\indexPosition}$ is a play consistent with $\stratAltAdvI$ starting in $\vertex_\indexPosition'$.
  Because $\vertex_\indexPosition'$ occurs in $\suffix{\play}{\indexLast^\star}$, we have $\vertex_\indexPosition'\in\winningIAdv{\buchi{\target_\indexPlayer}}$.
  As in the first case, we obtain $\play'\in\buchi{\target_\indexPlayer}$, ending the proof.
\end{proof}

\begin{remark}\label{remark:cobuchi:ne-size}
  In a finite arena, NE outcomes provided by Lemma~\ref{lem:cobuchi:outcomes} are guaranteed to be simple lassos.
  Constructing NE with the classical construction (encoding the entire outcome in the memory state space with additional punishment states) yields finite-memory NEs with memory of size at most $|\vertexSet|+\nPlayer$.
  This bound is smaller than that provided by Theorem~\ref{theorem:cobuchi:fm-ne} for finite arenas.
  This contrasts with Remark~\ref{remark:buchi:ne-size} regarding finite-memory NEs in Büchi games on finite arenas result from the proof of Theorem~\ref{thm:buchi:fmne:lasso}.
  We note that the value of Theorem~\ref{theorem:cobuchi:fm-ne} lies in its applicability for infinite arenas.
\end{remark}
 
\section{Conclusion}
In this paper, we have provided constructions for finite-memory NEs that are based on a relaxation of the classical punishment mechanism: players punish some deviations while others are left unpunished.
We obtain finite-memory strategies as follows: instead of encoding an entire play in the memory of the players (as is done in classical approaches), we encode information regarding a decomposition of a play along with information on the latest players to have moved, to discern who to punish when necessary.
We have used variants of this idea to show that in reachability, shortest-path and safety games, a memory quadratic in the number of players is sufficient to implement an NE with an outcome whose cost profile is componentwise smaller than or equal to that of some given NE, and to prove that finite-memory is always sufficient for the same problem in Büchi and co-Büchi games.
The limits of the approach can be seen in Büchi and co-Büchi games with Proposition~\ref{prop:buchi:dep}: it is no longer possible to obtain arena-independent memory upper bounds.

We discuss two natural directions for future work.
A first direction would be to study the usefulness of decomposition-based approaches to refine memory bounds for (constrained) Nash equilibria in games with other types of objectives.
A natural candidate objective would the parity objective, which generalises both Büchi and co-Büchi objectives.
Any construction showing that finite memory suffices in parity games would yield a unified construction for Büchi games and co-Büchi games.
We remark that the two-phase approach (enforce a prefix then use a decomposition-based strategy) appears to not directly extend to parity games.
Consider the arena of Figure~\ref{fig:conclusion} and the (parity) game in which the objective of $\playerOne$ is $\buchi{\{\vertex_1\}}$ and the objective of $\playerTwo$ is $\cobuchi{\{\vertex_1, \vertex_3\}}$.
The play $(\vertex_0\vertex_1\vertex_2)^\omega$ is arguably the simplest NE outcome from $\vertex_0$ that satisfies the objective of $\playerOne$.
A natural candidate decomposition for the second phase of (a variant of) the two-phase mechanism would be some trivial decomposition such as $((\vertex_0\vertex_1\vertex_2)^\omega)$ or a periodic decomposition such as $(\vertex_0\vertex_1\vertex_2\vertex_0,\vertex_0\vertex_1\vertex_2\vertex_0,\ldots)$.
In either case, once the second phase starts and $\playerOne$ uses a decomposition-based strategy, $\playerTwo$ has a profitable deviation by using the edge from $\vertex_0$ to $\vertex_2$ to avoid $\vertex_1$.
This suggests that a different approach may be necessary to study the more general parity games.
\begin{figure}
  \centering
  \begin{tikzpicture}
    \node[state,square] (v0) {$\vertex_0$};
    \node[state,right = of v0] (v1) {$\vertex_1$};
    \node[state,right = of v1] (v2) {$\vertex_2$};
    \node[state,square,right = of v2] (v3) {$\vertex_3$};

    \path[->] (v0) edge (v1);
    \path[->] (v0) edge[bend right] (v2);
    \path[->] (v1) edge (v2);
    \path[->] (v2) edge[bend right] (v0);
    \path[->] (v2) edge (v3);
  \end{tikzpicture}
  \caption{A two-player arena (circles and squares are respectively vertices of $\playerOne$ and $\playerTwo$) illustrating that the two-phase approach used in Büchi and co-Büchi games does not directly generalise to games with both Büchi and co-Büchi objectives: if the objective of $\playerOne$ is $\buchi{\{\vertex_2\}}$ and the objective of $\playerTwo$ is $\cobuchi{\{\vertex_2,\vertex_3\}}$, then there is no NE in which $\playerOne$ wins that eventually conforms to a decomposition-based strategy with respect to $(\vertex_0\vertex_1\vertex_2)^\omega$.}\label{fig:conclusion}
\end{figure}

A second direction would be to investigate the use of decomposition-based strategies to obtain results on memory requirements for equilibria besides Nash equilibria, such as secure equilibria and subgame perfect equilibria.

\section*{Acknowledgements}
The author expresses their thanks to Thomas Brihaye, Aline Goeminne and Mickael Randour for fruitful discussions and their comments on a preliminary version of this paper.

\bibliography{master_references}

\newpage

\appendix
\section{Proof of Theorem~\ref{thm:short:pOneOpt}}\label{appendix:thm:short:pOneOpt}

In this section, we provide a proof of Theorem~\ref{thm:short:pOneOpt}.
We fix a two-player arena $\arena = ((\vertexSetOne, \vertexSetTwo), E)$ and a target $\target\subseteq\vertexSet$, a weight function $\weight\colon\edgeSet\to\IN$ and the shortest-path game $\game = (\arena, \costReach{\target}{\weight})$.

We establish the existence of memoryless uniformly optimal strategies of $\playerOne$ in $\game$ by exploiting the existence of memoryless uniformly winning strategies in zero-sum reachability games (Theorem~\ref{thm:qualitative:ML strat}).
However, we cannot directly apply this result to the reachability game $(\arena, \reach{\target})$ as memoryless uniformly winning strategies in this game need not be optimal in $\game$.
Instead, we modify $\arena$ by removing edges of $\playerOne$ that are not optimal in one step to obtain an arena $\arena'$.
We identify these edges by studying the value of vertices.
We then conclude by showing that memoryless uniformly winning strategies in $(\arena', \reach{\target})$ are uniformly optimal strategies in $\game$.

\thmShortPOneOpt*
\begin{proof}
  We recall that there exists optimal strategies of $\playerOne$ in $\game$ from each vertex by Lemma~\ref{lem:short:determinacy}.

  First, we prove that for all $\vertex\in\vertexSetOne\setminus\target$, there exists $\vertex'\in\succSet{\vertex}$ such that $\val(\vertex) = \val(\vertex') + \weight(\vertex, \vertex')$.
  Let $\vertex\in\vertexSetOne\setminus\target$.
  If $\playerOne$ moves from $\vertex$ to $\vertex'\in\succSet{\vertex}$, $\playerOne$ can ensure $\val(\vertex') + \weight(\vertex, \vertex')$ at best.
  It follows that $\val(\vertex) = \min\{\val(\vertex') + \weight(\vertex, \vertex')\mid \vertex'\in\succSet{\vertex}\}$ (this minimum is well-defined because values are in the well-ordered set $\INbar$).
  This implies the claim.

  Second, we claim that for all $\vertex\in\vertexSetTwo\setminus \target$, we have $\val(\vertex) \geq \val(\vertex') + \weight(\vertex, \vertex')$ for all $\vertex'\in\succSet{\vertex}$.
  Let $\vertex\in\vertexSetTwo\setminus \target$ and $\vertex'\in\succSet{\vertex}$.
  Assume first that $\val(\vertex')$ is finite, in which case $\playerTwo$ has an optimal strategy from $\vertex'$ by Lemma~\ref{lem:short:determinacy}.
  It follows that $\playerTwo$ can ensure $\val(\vertex') + \weight(\vertex, \vertex')$ from $\vertex$ by moving to $\vertex'$ from $\vertex$ and playing optimally from there, which implies the desired inequality.
  Assume now that $\val(\vertex')$ is infinite.
  Then for all $\alpha\in\IN$, $\playerTwo$ has a strategy ensuring $\alpha$ from $\vertex'$.
  We conclude, similarly to the previous case, that $\val(\vertex)$ is infinite and therefore satisfies the inequality.

  Third, we remove edges of $\arena$ to derive a game in which values are unchanged with respect to $\game$.
  Let $\edgeSet'$ be obtained by removing from $\edgeSet$ all edges $(\vertex,\vertex')$ with $\vertex\in\vertexSetOne\setminus\target$ and $\val(\vertex) \neq \val(\vertex') + \weight(\vertex, \vertex')$.
  We let $\arena' = ((\vertexSetOne, \vertexSetTwo), \edgeSet')$ and $\game' = (\arena', \costReach{\target}{\weight})$.
  We remark that there are no deadlocks in $\arena'$ by the first point above.
  For any $\vertex\in\vertexSet$, let $\val'(\vertex)$ denote its value in $\game'$.
  We claim that
  \begin{enumerate}[(i)]
  \item for all $\vertex\in\vertexSet$, $\val(\vertex)=\val'(\vertex)$ and\label{item:pOneOpt:one}
  \item the winning regions in the reachability games $(\arena, \reach{\target})$ and $(\arena', \reach{\target})$ coincide.\label{item:pOneOpt:two}
  \end{enumerate}
    
  For~\ref{item:pOneOpt:one}, we observe that for all $\vertex\in\vertexSet$, $\val'(\vertex)\geq\val(\vertex)$ by definition of $\game'$ ($\playerOne$ has less choices than in $\game$).
  In particular, if $\val(\vertex) = +\infty$, then $\val'(\vertex) = +\infty$.
  We show the other inequality of~\ref{item:pOneOpt:one} by induction on $\val(\vertex)$ for vertices of finite value.
  For the base case, let $\vertex\in\vertexSet$ such that $\val(\vertex)=0$.
  In this case, $\playerOne$ has a strategy that reaches $\target$ from $\vertex$ using edges with zero weight that moves only to vertices with zero value until a target is reached (otherwise $\playerOne$ could not ensure a cost of $0$ from $\vertex$).
  These edges are in $\edgeSet'$.
  This ends the argument for the base case.

  We now assume by induction that for all $\beta\leq\alpha$ and all $\vertex\in\vertexSet$, if $\val(\vertex) = \beta$, then $\val'(\vertex)=\beta$.
  Let $\vertex\in\vertexSet$ such that $\val(\vertex)=\alpha+1$ and let us show that $\val'(\vertex)=\alpha+1$.
  A strategy of $\arena'$ that ensures $\alpha+1$ from $\vertex$ can be obtained as follows.
  Fix a strategy $\stratOne$ of $\arena$ that is optimal from $\vertex$.
  All outcomes of $\stratOne$ from $\vertex$ eventually visit some vertex $\vertex'$ with $\val(\vertex') < \val(\vertex)$ (as a target is eventually reached).
  The earliest such vertex is reached using only edges in $\edgeSet'$, as otherwise there would be an outcome of $\stratOne$ with cost greater than $\val(\vertex)$, contradicting its optimality.
  By modifying $\stratOne$ so $\playerOne$ uses an optimal strategy of $\game'$ once this earliest vertex is reached (which exists by induction), we obtain a strategy $\stratOne'$ of $\arena'$ that ensures $\val(\vertex)$ from $\vertex$, implying $\val(\vertex) \geq \val'(\vertex)$ and ending the inductive argument.
  
  We now show that~\ref{item:pOneOpt:two} holds.
  Clearly any vertex that is winning in $(\arena', \reach{\target})$ also is in $(\arena, \reach{\target})$.
  Conversely, fix a vertex $\vertex$ that is winning for $\playerOne$ in $(\arena, \reach{\target})$.
  If its value is finite, the claim follows from~\ref{item:pOneOpt:one}.
  Therefore, assume that $\val(\vertex)=+\infty$.
  Let $\stratOne$ be a winning strategy of $\playerOne$ from $\vertex$ in $(\arena, \reach{\target})$.
  Its behaviours in vertices of infinite value need not be restricted to obtain a strategy of $\arena'$, as the outgoing edges from $\playerOne$ vertices of infinite value are the same in $\arena$ and $\arena'$.
  Furthermore, all outcomes of $\stratOne$ eventually reach a vertex of finite value.
  By changing $\stratOne$ so it conforms to a strategy optimal in $\game'$ from the earliest such visited vertex, we obtain a strategy $\stratOne'$ that is winning from $\vertex$ in $(\arena', \reach{\target})$.
  This ends the proof of~\ref{item:pOneOpt:two}.

  Finally, we prove the claim of the theorem.
  Let $\stratOne$ be a memoryless uniformly winning reachability strategy in the reachability game $(\arena', \reach{\target})$.
  It follows from~\ref{item:pOneOpt:two} above that is it also a uniformly winning strategy in $(\arena, \reach{\target})$.
  It remains to show that $\stratOne$ is optimal from all vertices with finite value.
  Let $\vertex_0\in\vertexSet$ such that $\val(\vertex_0)$ is finite.
  Let $\play = \vertex_0\vertex_1\ldots$ be consistent with $\stratOne$.
  We claim that $\costReach{\target}{\weight}(\play)\leq \val(\vertex_0)$.
  Let $\indexLast= \min\{\indexPosition\in\IN\mid \vertex_\indexPosition\in\target\}$ (it exists by~\ref{item:pOneOpt:two} and the fact $\val(\vertex_0)$ is finite).
  We have, by choice of $\edgeSet'$ and the second claim above, that
  \[\costReach{\target}{\weight}(\play) = \weight(\prefix{\play}{\indexLast}) = \sum_{\indexPosition=0}^{\indexLast-1} \weight(\vertex_\indexPosition, \vertex_{\indexPosition+1}) \leq \sum_{\indexPosition=0}^{\indexLast-1}\val(\vertex_{\indexPosition}) - \val(\vertex_{\indexPosition+1}) = \val(\vertex_0),\]
  because $\val(\vertex_\indexLast) = 0$.
  This shows that $\stratOne$ is optimal from $\vertex$ and ends the proof.
\end{proof}

\section{Existence of Nash equilibria in shortest-path games}\label{appendix:existence}
In this section, we prove the (unconditional) existence of Nash equilibria in shortest-path games with non-negative integer weights on arbitrary arenas.
To prove this existence result, we build on an existence result for NEs in finite arenas from~\cite{DBLP:conf/lfcs/BrihayePS13}.
We assume familiarity with the results of Section~\ref{section:zero-sum} and Section~\ref{section:outcomes} in the following.

The authors of~\cite{DBLP:conf/lfcs/BrihayePS13} provide sufficient conditions on cost functions that guarantee the existence of (finite-memory) NEs that apply to shortest-path games in finite arenas.
Intuitively, one can construct an NE as follows: all players follow a memoryless uniformly optimal strategy of their respective coalition game, and any player who deviates from the outcome of this profile is punished with a memoryless uniformly optimal punishing strategy of the adversary of their coalition game.
Even though such punishing strategies do not always exist in infinite arenas (Example~\ref{ex:no opti}), we can still show that the outcome obtained by having the players follow uniformly optimal memoryless strategies of their coalition game is an NE outcome with the characterisation in Theorem~\ref{thm:charac:short}.
In particular, this is sufficient to prove the existence of NE in shortest-path games on infinite arenas.

\begin{theorem}\label{thm:short:existence}
  Let $\arena=\arenaTuple$ be an $\nPlayer$-player arena, and let $\target_1, \ldots, \target_\nPlayer\subseteq\vertexSet$ and $\weight_1, \ldots, \weight_\nPlayer\colon\edgeSet\to\IN$ respectively targets and weight functions for each player.
  Let $\game=(\arena, (\costReach{\target_\playerIndex}{\weight_\indexPlayer})_{\playerIndex\in\playerSet})$ be a shortest-path game.
  There exists an NE in $\game$ from any initial vertex.
\end{theorem}

\begin{proof} Let $\vertex_0\in\vertexSet$ be an initial vertex.
  For all $\playerIndex\in\playerSet$, let $\stratI$ be a memoryless uniformly optimal strategy in the coalition game $\game_\playerIndex = (\arena_\playerIndex, \costReach{\target_\playerIndex}{\weight_\indexPlayer})$ provided by Theorem~\ref{thm:short:pOneOpt}.
  We define $\stratProfile=(\stratI)_{\playerIndex\in\playerSet}$.
  We argue that $\play = \outcome{\stratProfile}{\vertex_0}$ is the outcome of an NE by Theorem~\ref{thm:charac:short}.
  Establishing this implies the existence of an NE from $\vertex_0$.

  The first condition of Theorem~\ref{thm:charac:short} follows from the strategies $\stratI$ being memoryless uniformly winning strategies in the reachability game $(\arena_\playerIndex, \reach{\target_\playerIndex})$.
  The second condition follows from the uniform optimality of the strategies $\stratI$ in $\game_\playerIndex$: $\stratI$ ensures $\val(\vertex)$ from all $\vertex\in\vertexSet$ for all $\indexPlayer\in\playerSet$.
  Therefore, the inequality in the second condition of Theorem~\ref{thm:charac:short} must hold in all relevant cases.
\end{proof}

\end{document}